\documentclass{elsarticle} 
\usepackage{ifsym} 
\usepackage{mathsetc}
\usepackage{proofnet}
\usepackage{algpseudocode}
\usepackage{algorithm} 
\usepackage{wrapfig}
\usepackage{ stmaryrd } 
\usepackage{url}

\usepackage{subcaption}

\def\la{\langle}
\def\ra{\rangle}
\def\Nll{{\mathbf N}}

\def\lmul{\{\!\!\{}
\def\rmul{\}\!\!\}}

\def\la{\langle}
\def\ra{\rangle}

\tikzstyle{tracetree}=[draw, shape=isosceles triangle,isosceles triangle apex angle=60,
inner xsep=0.015cm, inner ysep=0.015cm
]

\tikzstyle{indot}=[draw, rectangle, dotted]
\tikzstyle{tight}=[inner xsep=0.01cm,inner ysep=0.01cm]
\tikzstyle{tightt}=[inner xsep=0.004cm,inner ysep=0.004cm]
\tikzstyle{edgename}=[opacity=0]

\newtheorem{theorem}{Theorem}
\newtheorem{definition}[theorem]{Definition}
\newtheorem{lemma}[theorem]{Lemma}
\newtheorem{coro}[theorem]{Corollary}

\def\claim{  \noindent\textit{Proof.} }
\def\endclaim{ }
\newenvironment{prf}{\claim}{\endclaim}

\usepackage{listings}
\usepackage{color}

\definecolor{dkgreen}{rgb}{0,0.6,0}
\definecolor{gray}{rgb}{0.5,0.5,0.5}
\definecolor{mauve}{rgb}{0.58,0,0.82}

\journal{Information and Computation}

\begin{document}

\begin{frontmatter}
  \date{}
  \title{Paths-based criteria and application to linear logic subsystems characterizing polynomial time}
  
  \author[focal]{Matthieu Perrinel} 
  \ead{matthieu.perrinel@ens-lyon.fr}
  \address[focal]{LIP, ENS Lyon, 46 allée d'Italie, 69007 Lyon, FRANCE}

  \begin{abstract}
    Several variants of linear logic have been proposed to characterize complexity classes in the proofs-as-programs correspondence. Light linear logic (LLL) ensures a polynomial bound on reduction time, and characterizes in this way polynomial time ($Ptime$). In this paper we study the complexity of linear logic proof-nets and propose three semantic criteria based on context semantics: stratification, dependence control and nesting. Stratification alone entails an elementary time bound, the three criteria entail together a polynomial time bound.

    These criteria can be used to prove the complexity soundness of several existing variants of linear logic. We define a decidable syntactic subsystem of linear logic: $SDNLL$. We prove that the proof-nets of SDNLL satisfy the three criteria, which implies that $SDNLL$ is sound for $Ptime$. Several previous subsystems of linear logic characterizing polynomial time ($LLL$, $mL^4$, maximal system of $MS$) are embedded in $SDNLL$, proving its $Ptime$ completeness.
  \end{abstract}
\end{frontmatter}

\section{Introduction}
\paragraph{Motivations for a type-system capturing polynomial time}
Programming is a notoriously error-prone process. The behaviours of the programs written by programmers on their first attempt often differ from their expected behaviours. Type systems can detect some of those mistakes so that programmers can correct them more easily. In this work, the property we are interested in is time complexity: the execution time of a program as a function of the size of its input. A type system $S$ enforcing a polynomial bound on the time complexity of a program would be useful in several ways:
\begin{itemize}
\item In some real-time applications (e.g. car control systems) programs can never miss a deadline, otherwise the whole system is a failure. It is not enough to verify that the system reacted fast enough during tests, we need an absolute certainty.
\item For some software, it seems enough to get an empirical estimate of the complexity by running tests. In this case, $S$ could be useful to find the origin of the slowness observed during tests (this requires the type inferrer to give useful information when it fails to type a term).
\item In complexity theory, the main method to prove that a problem is $NP$-complete, is to define a polynomial time reduction from another $NP$-complete problem. If $S$ is well-trusted, it could be used as a specialized proof assistant: the fact that the reduction is typable in $S$ would increase the trust in the proof. More generally, $S$ could be used in any proof relying on a complexity bound for a program~\cite{nowak2014formal,zhang2009computational}.
\end{itemize}
In this work, we define a subsystem $SDNLL$ of linear logic such that every proof-net normalizes in polynomial time. This property is called $Ptime$ {\em soundness}. And, for every function $f$ computable in polynomial time there exists a $SDNLL$ proof-net $G_f$ which computes $f$. This property is called $Ptime$ {\em extensional completeness}.

Determining if a proof-net normalizes in polynomial time is undecidable. So for every such system $S$, either determining if a proof-net $G$ belongs to $S$ is undecidable, or $S$ is not {\em intensional complete}: i.e. there exist programs which normalize in polynomial time and are not typable by $S$. The subsystem $SDNLL$ is in the second case. We take inspiration from previous decidable type systems characterizing $Ptime$ and relax conditions without losing neither soundness nor decidability. The more intensionally expressive $S$ is (i.e. the more terms are typable by $S$), the more useful $S$ is. Indeed, the three motivations for systems characterizing polynomial time we described earlier require $S$ to type programs written by non-specialists: people who may not have a thorough understanding of $S$.

\paragraph{Linear logic and proof-nets} Linear logic ($LL$\label{def_acronym_ll})~\cite{girard1987linear} can be considered as a refinement of System F where we focus especially on how the duplication of formulae is managed. In linear logic, the structural rules (contraction and weakening) are only allowed for formulae of the shape $\oc A$:

\begin{equation*}
  \begin{array}{cc}
    \AxiomC{$\Gamma, \oc A, \oc A \vdash B$}
    \RightLabel{$\contLab$}
    \UnaryInfC{$\Gamma, \oc A \hspace{1.4em} \vdash B$}
    \DisplayProof & \hspace{6em}
    \AxiomC{$\Gamma \hspace{1.3em} \vdash B$}
    \RightLabel{$\weakLab$}
    \UnaryInfC{$\Gamma, \oc A \vdash B$}
    \DisplayProof
  \end{array}
\end{equation*}

With the three following additional rules (promotion, dereliction and digging), linear logic is as expressive as System F, so the elimination of the $cut$ rule (corresponding to the $\beta$-reduction of $\lambda$-calculus) is not even primitive recursive.

\begin{equation*}
  \begin{array}{ccc}
    \AxiomC{$\hspace{0.3em}A_1,\cdots,\hspace{0.3em}A_n \vdash \hspace{0.3em}B$}
    \RightLabel{$\fpriLab$}
    \UnaryInfC{$\oc A_1, \cdots,\oc A_k \vdash \oc B$}
    \DisplayProof  \hspace{3em}&
    \AxiomC{$\Gamma,\hspace{0.3em}A \vdash B$} 
    \RightLabel{$\derLab$}
    \UnaryInfC{$\Gamma, \oc A \vdash B$}
    \DisplayProof \hspace{3em}&                     
    \AxiomC{$\Gamma, \oc \oc A \vdash B$} 
    \RightLabel{$\digLab$}
    \UnaryInfC{$\Gamma, \hspace{0.4em}\oc A \vdash B$}
    \DisplayProof
  \end{array}
\end{equation*}
However, because the structural rules are handled by 5 distinct rules, one can enforce a subtle control on the use of ressources by modifying one of them. If we restrict some of those rules, it restricts the duplication of formulae. For instance, in the absence of $\derLab$ and $\digLab$ rules, the cut-elimination normalizes in elementary time~\cite{danos2003linear}. The set of such proofs is defined as Elementary Linear Logic ($ELL$).

Proof-nets~\cite{girard1996proof} are an alternative syntax for linear logic, where proofs are considered up-to meaningless commutations of rules. Proof-nets are graph-like structures where nodes correspond to logical rules. One of the reasons we use proof-nets instead of proof derivations is that context semantics, the main tool we use in this article, is much simpler to define and use in proof-nets.

\paragraph{Context semantics}Context semantics is a presentation of geometry of interaction~\cite{gonthier1992geometry,danos1995proof} defined by tokens traveling across proof-nets according to some rules. The paths defined by those tokens are stable by reduction so they represent the reduction of the proof-net. Context semantics has first been used to study optimal reduction~\cite{gonthier1992linear}.

Recently, it has been used to prove complexity bounds on subsystems of System T~\cite{lago2005geometry} and linear logic~\cite{baillot2001elementary,lago2006context}. In~\cite{lago2006context}, Dal Lago defines for every proof-net $G$ a weight $W_G \in \mathbb{N} \cup \{\infty\}$ based on the paths of context semantics such that, whenever $G$ reduces to $H$, $W_G \geq W_H+1$. Thus $W_G$ is a bound on the length of the longest path of reduction starting from $G$. Then we can prove theorems of the shape ``whenever $G$ satisfies some property (for instance if $G$ belongs to a subsystem such as $LLL$), $W_G$ satisfies some bound (for instance $W_G \leq P(|G|)$ with $P$ a polynomial and $|G|$ the size of $G$).'' 

From this point of view, context semantics has two major advantages compared to the syntactic study of reduction. First, its genericity: some common results can be proved for different variants of linear logic, which allows to factor out proofs of complexity results for these various systems. Moreover, the bounds obtained stand for any strategy of reduction. On the contrary, most bounds proved by syntactic means are only proved for a particular strategy. There are several advantages to strong bounds: 
\begin{itemize}
\item Let us suppose we know a strong complexity bound for a system $S'$. We can prove the same strong complexity bound on a system $S$ if we find an embedding $\phi$ of $S$ programs in $S'$ programs such that, whenever $t$ reduces to $u$ in $S$, $\phi(t)$ reduces to $\phi(u)$ in $S'$ (with at least one step). We use such an embedding in Section~\ref{section_sdnll_lambda} to prove a strong bound for $\lambda$-terms typed by $SDNLL$. If we only had a weak complexity bound for system $S'$, we would have to prove that the reduction from $\phi(t)$ to $\phi(u)$ matches the reduction strategy entailing the bound, which is not always possible.
\item The languages we study here are confluent. However, if we consider an extension of linear logic or $\lambda$-calculus with side-effects (such as $\lambda^{!R}$ considered by Madet and Amadio in \cite{madet2011elementary}), the reduction strategy influences the result of a program execution. It is important that the programmer understands the strategy. If the reduction strategy corresponded to strategies frequently used by programming languages (such as left-to-right call-by-value), it would not be a problem. However, in some cases ($mL^4$ for instance~\cite{baillot2010linear}), the strategy is rather farfetched and difficult to understand for the programmer.
\end{itemize}


Our context semantics, presented in Section~\ref{section_def_cont_sem}, is slightly different from Dal Lago's context semantics. In particular, Dal Lago worked in intuitionnistic linear logic, and we work in classical linear logic. So the results of~\cite{lago2006context} cannot be directly applied. However most theorems of~\cite{lago2006context} have correspondents in our framework, with quite similar proofs. This is why we omit the proofs of most of the results of this section, complete proofs can be found in~\cite{perrinelMegathese}.

\paragraph{Our approach} Contrary to previous works, we do not directly define a linear logic subsystem. First, we define semantic criteria forbidding behaviours which can result in non-polynomial complexity. We define relations $\rightarrow$ on boxes (special subterms of proof-nets) such that $B \rightarrow C$ means that "the number of times $B$ is copied depends on the number of times $C$ is copied". More precisely, we define three relations ($\stratSNLL$, $\dcSim$ and $\nestSim$) representing different kinds of dependence. The acyclicity of these three relations ensures a bound on the number of times every box is copied so a bound on the length of normalization sequences.

Then (in Section~\ref{chapter_type_systems}), we define {\em Stratified Dependence control Nested Linear Logic} ($SDNLL$), a subsystem of linear logic such that $\stratSNLL$, $\dcSim$ and $\nestSim$ are acyclic on every proof-net of $SDNLL$. This entails a bound on the length of normalization for every $SDNLL$ proof-net. The relations ($\stratSNLL$, $\dcSim$ and $\nestSim$) are based on the paths of context semantics. We use the syntactic restrictions to define invariants along the paths, proving that if  $B \rightarrow C$ then the "types" of $B$ and $C$ are such that we cannot have $C \rightarrow B$. Finally, in Section~\ref{section_sdnll_lambda}, we transform $SDNLL$ into a type-system $SDNLL_{\lambda}$ for $\lambda$-calculus, which enforces a polynomial bound on $\beta$-reduction. This tranformation is similar to the transformation of $LLL$~\cite{girard1995light} into $DLAL$~\cite{baillot2004light}.

\begin{figure}\centering
  \begin{tikzpicture}
    \node (LLL) at (0,0) {$LLL$};
    \node (MS)  at ($(LLL)+( 30:1.6)$) {$MS$};
    \node (L4)  at ($(LLL)+(150:1.1)$) {$L^4$};
    \node (L40) at ($(L4) +(150:1.1)$) {$L^4_0$};
    \node (SLL) at ($(LLL)+(4.5,0)$) {$SLL$};
    \node (BLL) at ($(SLL)+(3,0)$) {$BLL$};
    \node (QBAL)at ($(BLL)+(0,0.8)$) {$QBAL$};
    \node (SNLL)at ($(LLL)+(0.4,1.5)$) {$\mathbf{SDNLL}$};
    \draw [<-] (SNLL)--(L4);
    \draw [<-,dashed] (SNLL)--(MS);
    \draw [<-] (L40)--(L4);
    \draw [<-] (L4)--(LLL);
    \draw [<-] (MS)--(LLL);
    \draw [<-] (QBAL)--(BLL);
    \draw ($(L40|-LLL)+(-0.4,-0.25)$) rectangle ($(SLL|-SNLL)+(1,0.2)$);
    \draw ($(BLL)+(-1.5,-0.25)$) rectangle ($(QBAL)+(1,0.7)$);
    \node [below left] at ($(SLL|-SNLL)+(1,0.25)$) {Decidable};
    \node [below right] at ($(QBAL)+(-1.5,0.75)$) {Undecidable ?};
  \end{tikzpicture}
  \caption{\label{fig_state_art}State of the art}
\end{figure}
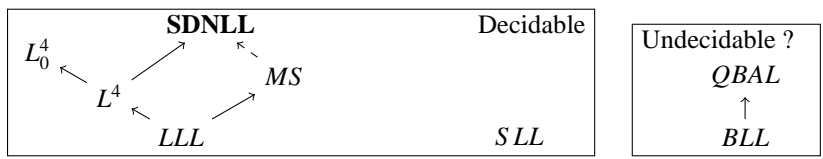

\paragraph{Previous polynomial time subsystems of Linear logic} There already exist several subsystems of linear logic characterizing polynomial time. The first such subsystem is $BLL$~\cite{girard1992bounded}, which enforces $Ptime$ soundness by labelling of $\oc$ modalities by polynomials. However, given a proof-net $G$, determining if $G$ is in $BLL$ (or its generalization $QBAL$~\cite{dal2009bounded}) seems undecidable. Thus, they do not fit in our approach.

The first decidable system was $LLL$~\cite{girard1995light} which is defined as the proof-nets of $ELL$ such that the contexts have at most one formula in every $\fpriLab$ rule\footnote{To keep some expressivity, Girard adds a new modality $\S$.}. A decidable type system $DLAL$ for $\lambda$-calculus was inspired by $LLL$~\cite{baillot2004light,atassi2007verification}.

Baillot and Mazza generalized $ELL$ with a subsystem $L^3$ of linear logic characterizing elementary time~\cite{baillot2010linear}. Then they defined $mL^4$ and $mL^4_0$, characterizing polynomial time, based on $L^3$ in the same way as $LLL$ is based on $ELL$. In a separate direction, Roversi and Vercelli also extended $LLL$ with $MS$\footnote{Which is a set of system rather than a unique system.}~\cite{roversi2010local}. Those three systems are obtained by decorating formulae with labels and adding local constraints on the labels. $mL^4$, $mL^4_0$ and $MS$ are trivially decidable on proof-nets: given a proof-net $G$ there exist only a finite number of ways to label the formulae of $G$. One can try every possibility and check whether the labels verify the constraints. Lafont defined $SLL$~\cite{lafont2004soft}, another subsystem of linear logic characterizing polynomial time. This system does not contain $LLL$, and none of the above generalizations of $LLL$ contains $SLL$.

Figure~\ref{fig_state_art} summarizes the state of the art. There is an arrow from the system $S$ to the system $T$ if there is a canonical embedding of $S$ in $T$. The arrow between $MS$ and $SDNLL$ is dotted because the embedding is only defined for one of the maximal systems of $MS$. In~\cite{perrinelMegathese}, we define $SwLL$ (based on the the ideas of this article) in which one can embed $SDNLL$, $SLL$, and every $MS$ polynomial subsystem.

This paper extends a previous work~\cite{perrinel2013pathsbased} by: providing a non-trivial nesting condition, defining a syntactic subsystem based on the semantic criteria, and providing most of the proofs (in~\cite{perrinel2013pathsbased} the proofs are only sketched). More details, and the technical proofs omitted in this paper can be found in Perrinel's thesis~\cite{perrinelMegathese}.

\section{Linear Logic and Context Semantics}\label{chapter_2}
\subsection{Linear Logic} 
Linear logic ($LL$)~\cite{girard1987linear} can be considered as a refinement of System F~\cite{girard1971extension} where we focus especially on how the duplication of formulae is managed. In this work we use neither the additives ($\oplus$ and $\with$) nor the constants. This fragment is usually named {\em Multiplicative Exponential Linear Logic with Quantifiers} (abbreviated by $MELL_\forall$). To simplify notations, we will abusively refer to it as {\em Linear Logic} (abreviated by $LL$). The set $\formLL$, defined as follows, designs the set of formulae of linear logic.
\begin{equation*}\label{def_fll}
  \formLL  = X \mid X^\perp  \mid \formLL \otimes \formLL \mid \formLL \parr \formLL \mid \forall X. \formLL \mid \exists X. \formLL \mid \oc \formLL \mid \wn \formLL 
\end{equation*}

We define inductively an involution $(\_)^\perp$ on $\formLL$, which can be considered as a negation:\label{def_perpformula} $(X)^\perp =X^\perp$, $(X^\perp)^\perp =X$, $(A \otimes B)^\perp = A^\perp \parr B^\perp$, $(A \parr B)^\perp = A^\perp \otimes B^\perp$, $(\forall X.A)^\perp = \exists X.A^\perp$, $(\exists X.A)^\perp = \forall X.A^\perp$, $(\oc A)^\perp = \wn(A^\perp)$ and $(\wn A)^\perp= \oc (A^\perp)$.

Linear logic is usually presented as a sequent calculus (as in the introduction). In this article, we will consider an alternative syntax: proof-nets~\cite{girard1996proof}.

\begin{definition}\label{def_proofnet}
  A {\em $LL$ proof-net} is a graph-like structure, defined inductively by the graphs of Figure~\ref{rules_labelling_proofnet} ($G$ and $H$ being $LL$ proof-nets). Every edge $e$ is labelled by $\beta(e) \in \formLL$ satisfying the constraints of Figure~\ref{rules_labelling_proofnet}. The set of edges is written $\dirEdges{G}$.

  A {\em proof-net} is a graph-like structure, whose edges are not labelled, defined inductively by the graphs of Figure~\ref{rules_labelling_proofnet} ($G$ and $H$ being proof-nets). The constraints of Figure~\ref{rules_labelling_proofnet} on labels are not taken into account.
\end{definition}

\begin{figure}
  \begin{tikzpicture}
  \tikzstyle{level}=[opacity = 0]
  \tikzstyle{type}=[black, midway, right=-0.08cm]

  \nvar{\ligneDeux}{2.3cm}
  \nvar{\ligneTrois}{2.6cm}

    \begin{scope}[scale = 0.85]
      \begin{scope}[shift={(4,1.8)}]
        \begin{scope}[shift={(-1,0)}]
          \node [ax] (ax) at (-0.5,0) {};
          \draw[ar] (ax) to [out=  0,in=90] node [type]              {$A^\perp$} ($(ax)+(0.5,-0.7)$);
          \draw[ar] (ax) to [out=180,in=90] node [type,left=-0.08cm] {$A$}      ($(ax)+(-0.5,-0.7)$);
        \end{scope}
        
        \node [proofnet, inner xsep=0.5cm] (G)   at (2.9,0)         {$G$};
        \node [proofnet, inner xsep=0.5cm] (H)   at ($(G)+(1.8,0)$) {$H$};
        \node [cut]      (cut) at ($(G)!0.5!(H)+(-0.15,-0.7)$) {};
        \draw [ar] (G.-65)  to [out=-90,in=180] node [type,left=-0.08cm] {$A$}      (cut);
        \draw [ar] (H.-145) to [out=-90,in=  0] node [type]              {$A^\perp$} (cut);
        \draw [multiar] (G.-158) --++ (0,-0.4);
        \draw [multiar] (H.-22) --++ (0,-0.4);
      \end{scope}
      
      \begin{scope}[shift={(-0.3,0)}]
        \node[proofnet, inner xsep=0.5cm] (G) at (0,0)           {$G$};
      \node[proofnet, inner xsep=0.5cm] (H) at ($(G)+(1.8,0)$) {$H$};
      \node[tensor]  (tens) at ($(G)!0.5!(H)+(0,-0.8)$) {};
      \draw [ar] (G.-50)  to [out=-90,in=150] node [type,left=-0.08cm] {$A$} (tens);
      \draw [ar] (H.-130) to [out=-90,in= 30] node [type]              {$B$} (tens);
      \draw [ar] (tens) --++ (0,-0.5) node [type] {$A \otimes B$};
      \draw [multiar] (G.-158) --++ (0,-0.5);
      \draw [multiar] (H.-22)  --++ (0,-0.5);
    \end{scope}
      
    \begin{scope}[shift={(4.7, 0)}]
      \node[proofnet, inner xsep=0.6cm] (G) at (0,0) {$G$};
      \node[par] (par) at ($(G.-160)!0.5!(G.-90)+(0,-0.8)$) {};
      \draw[ar] (G.-160) to [out=-90,in=120] node [type,left] {$A$} (par);
      \draw[ar] (G.-90)  to [out=-90,in= 60] node [type] {$B$} (par);
      \draw[ar] (par) --++ (0,-0.5) node [type] {$A \parr B$};
      \draw[multiar]  (G.-20) --++ (0,-0.5);
    \end{scope}
      
      \begin{scope}[shift={(8.5,0)}]
        \draw (0,0) node [proofnet, inner xsep=0.7cm] (G) {$G$};
        \node [exists] (ex) at ($(G.-164)+(0,-0.8)$) {};
        \draw [multiar](G.-16) --++(0,-0.4);
        \draw[ar] (G.-164) -- (ex) node [type] {$A[B/X]$};
        \draw[ar] (ex) --++ (0,-0.5) node [midway, right] {$\exists X. A$};
      \end{scope}
      
      \begin{scope}[shift={(11.8, 0)}]
        \draw (0,0) node [proofnet, inner xsep=0.6cm] (G) {$G$};
        \node [forall] (fa) at ($(G.-155)+(0,-0.8)$) {};
        \draw [multiar](G.-20) --++(0,-0.4);
        \draw[ar] (G.-155) -- (fa) node [type] {$A$};
        \draw[ar] (fa) --++ (0,-0.5) node [midway, right] {$\forall X. A$};
      \end{scope}

      \begin{scope}[shift={(9.1,- \ligneTrois)}]
        \draw (0,0) node [proofnet] (G) {$G$};
        \node [der] (fa) at ($(G.-145)+(0,-0.8)$) {};
        \draw [multiar](G.-35) --++(0,-0.4);
        \draw[ar] (G.-145) -- (fa) node [type] {$A$};
        \draw[ar] (fa) --++ (0,-0.5) node [midway, right] {$\wn A$};
      \end{scope}
      
      \begin{scope}[shift={(3.3,0.1cm-\ligneTrois)}]  
        \draw (0,0) node [proofnet, inner xsep=0.6cm] (G) {$G$};
        \node [cont] (cont) at ($(G.-65)+(0,-0.8)$) {};
        \draw[ar] (G.-20) to [out=-90,in= 30] node [type,pos=0.3] {$\wn A$} (cont);
        \draw[ar] (G.-145)to [out=-90,in=160] node [type,pos=0.3] {$\wn A$} (cont);
        \draw[->] (cont) --++ (0,-0.5) node [midway, right] {$?A$};
        \draw[multiar] (G.-160) --++(0,-0.8); 
      \end{scope}
      
      \begin{scope}[shift={(6,- \ligneTrois)}]
        \draw (0,0) node [proofnet] (G) {$G$};
        \draw [multiar] (G.-90) --++(0,-0.5);
        \node [weak] (weak) at ($(G)+(1,0)$) {};
        \draw [ar] (weak) --++(0,-0.8)  node [type] {$?A$};
      \end{scope}
      
      \begin{scope}[shift={(0,0.3cm-\ligneTrois)}]
        \draw (0,0.1) node [proofnet,inner xsep=0.8cm] (G) {$G$};
        \draw (G)++(1.1,-1.1) node [princdoor] (bang)   {};
        \draw (G)++(0.,-1.1) node [auxdoor] (whyn) {};
        \draw (G)++(-1.1,-1.1) node [auxdoor] (whyn2) {};
        \draw[ar] (whyn2 |- G.south) -- (whyn2) node [type] {$A_1$};
        \draw[ar] (whyn |- G.south) -- (whyn) node [type] {$A_n$};
        \draw[ar] (bang |- G.south) -- (bang) node [type] {$B$};
        \draw[ar] (whyn2) --++(0,-0.8) node [type] {$?A_1$};
        \draw[ar] (whyn) --++(0,-0.8) node [type] {$?A_n$};
        \draw[ar] (bang) --++(0,-0.8) node [type] {$\oc B$};
        \draw (whyn)--(bang) -| ++(0.45,1.5) -| ($(whyn2)+(-0.4,0)$) -- (whyn2);
        \draw [dotted] (whyn2) -- (whyn);
      \end{scope}
      
      \begin{scope}[shift={(11.7,- \ligneTrois)}]
        \draw (0,0) node [proofnet, inner xsep=0.6cm] (G) {$G$};
        \node [dig] (fa) at ($(G.-160)+(0,-0.8)$) {};
        \draw [multiar](G.-20) --++(0,-0.4);
        \draw[ar] (G.-160) -- (fa) node [type] {$\wn \wn A$};
        \draw[ar] (fa) --++ (0,-0.5) node [midway, right] {$\wn A$};
      \end{scope}
  \end{scope}

\end{tikzpicture}
\caption{ \label{rules_labelling_proofnet}Construction of $LL$ proof-nets. In the $\forall$ rule, $X$ can not be free in the other conclusions of $G$.}
\end{figure}
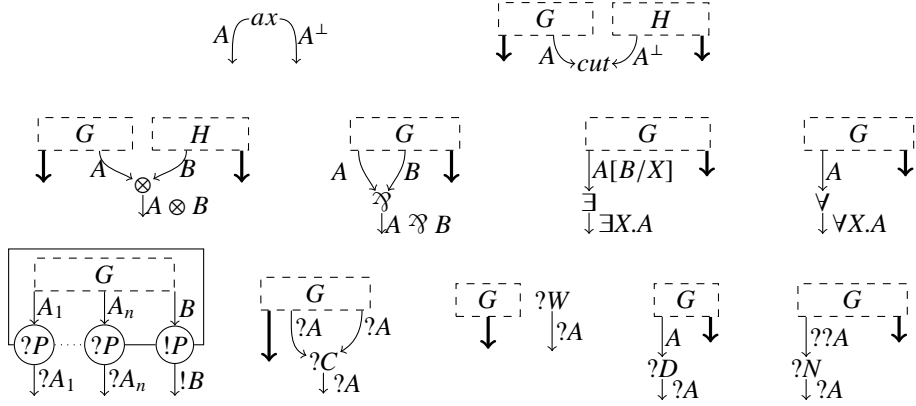

For the following definitions, we supposed fixed a proof-net $G$.

\paragraph{Directed edges}\label{directed_edges}The edges in the definition of proof-nets (the elements of $\dirEdges{G}$) are directed. We will often need to consider their inverted edges:  for any $(l,m)$ (the edge from $l$ to $m$), we denote its inverted edge $(m,l)$ (the edge from $m$ to $l$) by $\overline{(l,m)}$. We define the set $\edges{G}$ as $\dirEdges{G} \cup \Set{\overline{e}}{e \in \dirEdges{G}}$. In $LL$ proof-nets, we extend the labelling $\beta(\_)$ from $\dirEdges{G}$ to $\edges{G}$ by $\beta(\overline{e})=\beta(e)^\perp$.

\paragraph{Premises and conclusions}\label{def_premise} For any node $n$, the incoming edges of $n$ (in $\dirEdges{G}$) are named the {\em premises} of $n$.\label{def_conclusion} The outgoing edges of $n$ (in $\dirEdges{G}$) are named the {\em conclusions} of $n$. Some edges are not the premises of any node.\label{def_pending} Such edges are the {\em conclusions} of $G$.

\paragraph{Boxes}\label{def_box}The rectangle of Figure~\ref{rules_labelling_proofnet} with the $\fauxLab$ and $\fpriLab$ nodes is called a {\em box}. Formally a box is a subset of the nodes of the proof-net. We say that the edge $(m,n) \in \dirEdges{G}$ belongs to box $B$ if $n \in B$, in this case $(n,m)$ also belongs to box $B$. 

Let us call $B$ the box in Figure~\ref{rules_labelling_proofnet}.\label{def_principaldoor} The node labelled $\fpriLab$ is the {\em principal door} of $B$, its conclusion \label{def_sigmab}is written $\sigma(B)$, and is named the {\em principal edge} of $B$.\label{def_auxiliarydoor} The $\fauxLab$ nodes are the auxiliary doors of box $B$.\label{def_sigmaib} The edge going out of the $i$-th auxiliary door is written $\sigma_i(B)$ and is named an auxiliary edge of $B$. The doors of box $B$ are considered in $B$, they are exactly the nodes which are in $B$ but whose conclusions are not in $B$.

The number of boxes containing an element (box, node or edge) $x$ is its {\em depth} written $\depth{x}$.\label{def_maxdepth} $\partial_G$ is the maximum depth of an edge of $G$. The set of boxes of $G$ is $\boxset{G}$.

\tikzstyle{edgename}=[phantom]
\begin{figure}\centering
  \begin{tikzpicture}[baseline=0cm]
  \tikzstyle{edgename}=[opacity=0]
  
  \begin{scope}[shift={(7.8,0.7)}]
    \node [ax]  (ax)  at (0, 0) {};
    \node [etc] (etc) at ($(ax)+(1.1,0)$) {};
    \node [cut] (cut) at ($(ax)!0.6!(etc)+(0,-0.5)$) {};
    \draw [ar,out=  0,in=180] (ax) to (cut); 
    \draw [ar,out=-90,in=  0] (etc) to node [type,right=-0.08cm] {$A$} (cut);
    \draw [ar,out=180,in= 90] (ax)  to node [type,left=-0.08cm] {$A$} ($(ax)+(-0.5,-0.5)$);
     
    \draw [reduc] (1.5,-0.2) --++ (0.7,0) node [below left=-0.1cm] {$cut$};
    
    \node[etc] (etc2) at ($(etc)+(1.4,0)$) {};
    \draw[ar] (etc2) --++ (0,-0.5) node [type,right=-0.1cm] {$A$};
  \end{scope}

  \begin{scope}[shift={(0,0)}]
    \node [forall] (fa) at (0, 0) {};
    \node [exists] (ex) at ($(fa)+(0.7,0)$) {};
    \node [etc] (etcfa) at ($(fa)+(0,0.8)$) {};
    \node [etc] (etcex) at (etcfa-|ex) {};
    \node [cut] (cut)   at ($(ex)!0.5!(fa)+(0,-0.5)$) {};
    \draw [ar] (etcfa) -- (fa) node [type,left=-0.08cm] {$A$};
    \draw [ar] (etcex) -- (ex) node [type,right=-0.08cm] {$A^\perp[B/X]$};
    \draw [ar,out=-90,in=180] (fa) to node [type,below left=-0.08cm,pos=0.5] {$\forall X.A$} (cut);
    \draw [ar,out=-90,in=  0] (ex) to node [type,below right=-0.08cm,pos=0.5]{$\exists X.A^\perp$} (cut);
    
    \draw [reduc] (1.8,0) --++ (0.8,0) node [below left=-0.1cm] {$cut$};
    
    \node [etc] (etcfa2) at ($(etcfa)+(3.8,0)$) {};
    \node [etc] (etcex2) at ($(etcex)+(3.8,0)$) {};
    \node [cut] (cut) at ($(etcfa2)!0.5!(etcex2)+(0,-0.6)$) {};
    \draw [ar,out=-90,in=180] (etcfa2) to node [edgename, below left]  {$a$} node [type,left=-0.07,pos=0.25] {$A[B/X]$} (cut);
    \draw [ar,out=-90,in=  0] (etcex2) to node [edgename, below right] {$b$} node [type,right=-0.07,pos=0.25]      {$A^\perp[B/X]$} (cut);
  \end{scope}

  \begin{scope}[shift={(3.2,-1.6)}]
    \node [par]  (par)  at (0,0) {};
    \node [tensor] (tens) at ($(par)+(1.5,0)$) {};
    \node [cut]  (cut)  at ($(par)!0.5!(tens)+(0,-0.5)$) {};
    \node [etc] (pl) at ($(par) +(110:0.7)$) {};
    \node [etc] (pr) at ($(par) +( 70:0.7)$) {};
    \node [etc] (tl) at ($(tens)+(110:0.7)$) {};
    \node [etc] (tr) at ($(tens)+( 70:0.7)$) {};
    \draw [ar] (pl) -- (par) node [type,left] {$A$} node [edgename] {$a$};
    \draw [ar] (pr) -- (par) node [type]      {$B$} node [edgename,right] {$b$};
    \draw [ar] (tl) -- (tens)node [type,left] {$A^\perp$} node [edgename] {$e$};
    \draw [ar] (tr) -- (tens)node [type]      {$B^\perp$} node [edgename,right] {$f$};
    \draw [ar,out=-90,in=180] (par)  to node [edgename,below left] {$c$} node [type,left,pos=0.25] {$A \parr B$} (cut);
    \draw [ar,out=-90,in=  0] (tens) to node [edgename,below right]{$d$} node [type,pos=0.25] {$A^\perp \otimes B^\perp$}(cut);
  
    \draw [reduc] (2.4,0.4) --++(0.7,0) node [below left=-0.1cm] {$cut$};
    
    \node (pl2) at ($(pl)+(4, 0)$) {}; 
    \node (pr2) at ($(pr)+(4,0)$) {};
    \node (tl2) at ($(tl)+(4,0)$) {};
    \node (tr2) at ($(tr)+(4,0)$) {};
    \node [cut] (cutl) at ($(pl2)!0.5!(tl2)+(-0.3,-0.7)$) {};
    \node [cut] (cutr) at ($(pr2)!0.5!(tr2)+( 0.3,-0.7)$) {};
    \draw [ar,out=-90,in=180] (pl2) to node [type,pos=0.1,left] {$A$}      node [edgename,pos=0.15] {$a$} (cutl);
    \draw [ar,out=-90,in=  0] (tl2) to node [type,right=-0.08cm, pos=0.1]      {$A^\perp$} node [edgename,pos=0.1] {$e$} (cutl);
    \draw [ar,out=-90,in=180] (pr2) to node [type,pos=0.1,left] {$B$}      node [edgename,pos=0.1] {$b$} (cutr);
    \draw [ar,out=-90,in=  0] (tr2) to node [type,pos=0.1]      {$B^\perp$} node [edgename,pos=0.15] {$f$} (cutr);
  \end{scope}
\end{tikzpicture}
\caption{\label{cut_elim_non_exp_rules}Non-exponential cut-elimination steps. For the $\forall/\exists$ step, $[B/X]$ takes place on the whole net.}
\end{figure}
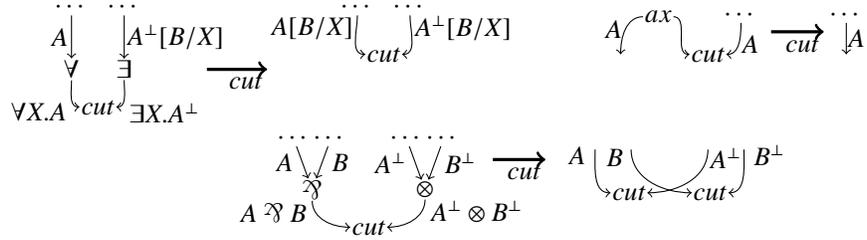

\paragraph{Cut-elimination} is a relation $\cutRel$ on ($LL$) proof-nets which is related to the $\beta$-reduction of $\lambda$-calculus. Figures \ref{cut_elim_non_exp_rules} and \ref{cut_elim_exp_rules} describe the rules of cut-elimination.

\begin{figure}\centering

\begin{subfigure}{\textwidth}
\begin{tikzpicture}[baseline=0.2cm]
  \tikzstyle{edgename}=[opacity=0]
  \node [proofnet,minimum width=2.5cm] (G) at (0,0) {$G'$};
  \node [princdoor] (pr) at ($(G.-13)+(0,-0.7)$) {};
  \node [above] at ($(pr)+(0.5,0)$) {$\mathbf{B}$};
  \node [auxdoor]   (a2) at ($(G.-90)+(0,-0.7)$) {};
  \node [auxdoor]   (a1) at ($(G.-167)+(0,-0.7)$) {};
  \draw (a2)--(pr) -| ++(0.7,1.3) -| ($(a1)+(-0.55,0)$) -- (a1);
  \draw [dotted] (a1)--(a2);
  \draw [ar] (G.-13)--(pr)    node [edgename] {$a$}   node [type] {$A$};     
  \draw [ar] (G.-90)--(a2)    node [edgename] {$a_k$} node [type] {$A_k$};    
  \draw [ar] (G.-167)--(a1)   node [edgename] {$a_1$} node [type] {$A_1$};    
  \draw [ar] (a2)--++(0,-0.65) node [edgename] {$c_k$} node [type] {$\wn A_k$};
  \draw [ar] (a1)--++(0,-0.65) node [edgename] {$c_1$} node [type] {$\wn A_1$};
  \node [der]  (der) at ($(pr)+(1.2,0)$) {};
  \node [cut]   (cut)at ($(pr)!0.5!(der)+(0,-0.6)$) {};
  \draw [ar,out=-90,in=180] (pr) to node [edgename,below left]  {$c$} node [type,pos=0.2] {$\oc A$} (cut);
  \draw [ar,out=-90,in=  0] (der) to node [edgename,right] {$d$} node [type,pos=0.3] {$\wn A^\perp$}(cut);
  \node [etc] (etc) at ($(der)+(0,1)$) {};
  \draw [ar] (etc) -- (der)   node [edgename] {$b$} node [type] {$A^\perp$};

  \draw [->,very thick] (4,-0.7) --++(1,0) node [below left=-0.1cm] {$cut$};

  \node [proofnet,minimum width=2.5cm] (H) at ($(G)+(8,0)$) {$G'$};
  \node [der] (b2) at ($(H.-90)+(0,-0.7)$) {};
  \node [der] (b1) at ($(H.-167)+(0,-0.7)$) {};
  \draw [dotted] (b1)--(b2);
  \draw [ar] (H.-90)--(b2)    node [edgename] {$a_k$} node [type] {$A_k$};    
  \draw [ar] (H.-167)--(b1)   node [edgename] {$a_1$} node [type] {$A_1$};    
  \draw [ar] (b1)--++(0,-0.6) node [edgename] {$c_1$} node [type] {$\wn A_1$};
  \draw [ar] (b2)--++(0,-0.6) node [edgename] {$c_k$} node [type] {$\wn A_k$};
  \node [etc] (etd) at ($(etc)+(8,0)$) {};
  \node [cut]   (cut)at ($(H.-13)!0.5!(etd)+(0,-0.8)$) {};
  \draw [ar,out=-90,in=180] (H.-13) to node [edgename,below left]  {$a$} node [type,pos=0.2] {$A$} (cut);
  \draw [ar,out=-90,in=  0] (etd)   to node [edgename,right] {$b$} node [type,pos=0.3] {$A^\perp$} (cut);
\end{tikzpicture}
\end{subfigure}
\vspace{1em}

\begin{subfigure}{\textwidth}
  \begin{tikzpicture}[baseline=0.2cm]
  \tikzstyle{edgename}=[opacity=0]
  \node [proofnet,minimum width=2.2cm] (G) at (0,0) {$G'$};
  \node [princdoor] (prg) at ($(G.-16)+(0,-0.7)$) {};
  \node [auxdoor]   (a2) at ($(G.-90)+(0,-0.7)$) {};
  \node [auxdoor]   (a1) at ($(G.-164)+(0,-0.7)$) {};
  \draw (a2)--(prg) -| ++(0.4,1.3) -| ($(a1)+(-0.4,0)$) -- (a1);
  \draw [dotted] (a1)--(a2);
  \draw [ar] (G.-16)--(prg)   node [type,right=-0.06cm] {$A$}       node [edgename] {$a$};     
  \draw [ar] (G.-90)--(a2)    node [type,right=-0.06cm] {$A_k$}     node [edgename] {$a_k$};
  \draw [ar] (a2)--++(0,-0.65) node [type,right=-0.06cm] {$\wn A_k$} node [edgename] {$c_k$} ;
  \draw [ar] (G.-164)--(a1)   node [type,right=-0.06cm] {$A_1$}     node [edgename] {$a_1$};    
  \draw [ar] (a1)--++(0,-0.65) node [type,right=-0.06cm] {$\wn A_1$} node [edgename] {$c_1$};

  \node [proofnet,minimum width=2.2cm] (H) at ($(G)+(3,0)$) {$H'$};
  \node [princdoor] (prh) at ($(H.-16)+(0,-0.7)$) {};
  \node [auxdoor]   (b2) at ($(H.-90)+(0,-0.7)$) {};
  \node [auxdoor]   (b1) at ($(H.-164)+(0,-0.7)$) {};
  \draw (b2)--(prh) -| ++(0.4,1.3) -| ($(b1)+(-0.4,0)$) -- (b1);
  \draw [dotted] (b1)--(b2);
  \draw [ar] (H.-16)--(prh)   node [type,right=-0.06cm] {$B$}       node [edgename] {$b$};     
  \draw [ar] (H.-90) --(b2)   node [type,right=-0.06cm] {$B_l$}     node [edgename] {$b_l$};    
  \draw [ar] (b2)--++(0,-0.65) node [type,right=-0.06cm] {$\wn B_l$} node [edgename] {$d_l$};
  \draw [ar] (H.-164)--(b1)   node [type,right=-0.06cm] {$A^\perp$} node [edgename] {$b_1$};  
  \draw [ar] (prh)--++(0,-0.65) node [type,right=-0.06cm] {$\oc B$}  node [edgename] {$d$};

  \node [cut] (cut) at ($(prg)!0.5!(b1)+(0,-0.65)$) {};
  \draw [ar,out=-90,in=180] (prg) to node [edgename, below left] {$c$}    node [type,pos=0.2] {$\oc A$}      (cut);
  \draw [ar,out=-90,in=  0] (b1)  to node [edgename, below right=-0.2cm] {$d_1$} node [type,pos=0.2] {$\wn A^\perp$} (cut);

  \draw [reduc] (4.4,0)--++(0.7,0) node [below left=-0.1cm] {$cut$};

  \node [proofnet,minimum width=2.2cm] (G) at ($(G)+(6.9,0)$) {$G'$};
  \node [auxdoor]   (a2) at ($(G.-90)+(0,-0.7)$) {};
  \node [auxdoor]   (a1) at ($(G.-164)+(0,-0.7)$) {};
  \node [proofnet,minimum width=2.5cm] (H) at ($(G)+(2.9,0)$) {$H'$};
  \node [princdoor] (prh) at ($(H.-14)+(0,-0.7)$) {};
  \node [auxdoor]   (b2) at ($(H.-41)+(0,-0.7)$) {};
  \node [auxdoor]   (b1) at ($(H.-159)+(0,-0.7)$) {};
  \node [cut]       (cut)at ($(G.-16)!0.5!(H.-170)+(0,-0.6)$) {};
  \draw [ar,out=-90,in=180] (G.-18)  to node [edgename, below left =-0.1]  {$a$}   node [type,pos=0.3] {$A$}     (cut);
  \draw [ar,out=-90,in=  0] (H.-170) to node [edgename, below right=-0.15] {$b_1$} node [type,pos=0.3] {$A^\perp$}(cut);
  \draw (b2)--(prh) -| ++(0.4,1.3) -| ($(a1)+(-0.4,0)$) -- (a1);
  \draw (a2)--(b1);
  \draw [dotted] (a1)--(a2);
  \draw [dotted] (b1)--(b2);
  \draw [ar] (G.-90)--(a2)     node [type] {$A_k$}     node [edgename] {$a_k$};
  \draw [ar] (a2)--++(0,-0.8)  node [type] {$\wn A_k$} node [edgename] {$c_k$};
  \draw [ar] (G.-164)--(a1)    node [type] {$A_1$}     node [edgename] {$a_1$};
  \draw [ar] (a1)--++(0,-0.8)  node [type] {$\wn A_1$} node [edgename] {$c_1$};
  \draw [ar] (H.-41)--(b2)     node [type] {$B_l$}     node [edgename] {$b_l$};
  \draw [ar] (b2)--++(0,-0.8)  node [type] {$\wn B_l$} node [edgename] {$d_l$};
  \draw [ar] (H.-159)--(b1)    node [type] {$B_2$}     node [edgename,left=-0.1cm] {$b_2$};
  \draw [ar] (b1)--++(0,-0.8)  node [type] {$\wn B_2$} node [edgename] {$d_2$};
  \draw [ar] (H.-14)--(prh)    node [type] {$B$}       node [edgename] {$b$};     
  \draw [ar] (prh)--++(0,-0.8) node [type] {$\oc B$}   node [edgename] {$d$}; 
\end{tikzpicture}
\end{subfigure}
\vspace{1em}

\begin{subfigure}{\textwidth}
  \begin{tikzpicture}[baseline=0.2cm]
  \tikzstyle{edgename}=[opacity=0]
  \node [proofnet,minimum width=2.2cm] (G) at (0,0) {$G'$};
  \node [princdoor] (pr) at ($(G.-16)+(0,-0.7)$) {};
  \node [auxdoor]   (a2) at ($(G.-90)+(0,-0.7)$) {};
  \node [auxdoor]   (a1) at ($(G.-164)+(0,-0.7)$) {};
  \draw (a2)--(pr) -| ++(0.5,1.3) -| ($(a1)+(-0.55,0)$) -- (a1);
  \draw [dotted] (a1)--(a2);
  \draw [ar] (G.-16)--(pr)    node [edgename] {$a$}   node [type,right=-0.06cm] {$A$};     
  \draw [ar] (G.-90)--(a2)    node [edgename] {$a_k$} node [type,right=-0.06cm] {$A_k$};
  \draw [ar] (G.-164)--(a1)   node [edgename] {$a_1$} node [type,right=-0.06cm] {$A_1$};
  \draw [ar] (a2)--++(0,-0.65) node [edgename] {$d_k$} node [type,right=-0.06cm] {$\wn A_k$};
  \draw [ar] (a1)--++(0,-0.65) node [edgename] {$d_1$} node [type,right=-0.06cm] {$\wn A_1$};
  \node [weak]  (wk) at ($(pr)+(1,0)$) {};
  \node [cut]   (cut)at ($(pr)!0.5!(wk)+(0,-0.6)$) {};
  \draw [ar,out=-90,in=180] (pr) to node [edgename,below left]  {$c$} node [type,pos=0.2] {$\oc A$} (cut);
  \draw [ar,out=-90,in=  0] (wk) to node [edgename,right]       {$d$} node [type,pos=0.3] {$\wn A^\perp$} (cut);

  \draw [->,very thick] (2.7,-0.7) --++(0.8,0) node [below left=-0.1cm] {$cut$};

  \node [weak] (b1) at ($(a1)+(4.8,0)$) {};
  \node [weak] (b2) at ($(a2)+(5.3,0)$) {};
  \draw [ar] (b1)--++(0,-0.65) node [edgename] {$d_1$} node [type] {$\wn A_1$};
  \draw [ar] (b2)--++(0,-0.65) node [edgename] {$d_k$} node [type] {$\wn A_k$};
  \draw [dotted] (b1)--(b2);
\end{tikzpicture}
\end{subfigure}
\vspace{1em}

\begin{subfigure}{\textwidth}
\begin{tikzpicture}[baseline=0.2cm]
  \tikzstyle{edgename}=[opacity=0]
  \node [proofnet,minimum width=2.5cm] (G) at (0,0) {$G'$};
  \node [princdoor] (pr) at ($(G.-16)+(0,-0.7)$) {};
  \node [above] at ($(pr)+(0.5,0)$) {$\mathbf{B}$};
  \node [auxdoor]   (a2) at ($(G.-90)+(0,-0.7)$) {};
  \node [auxdoor]   (a1) at ($(G.-164)+(0,-0.7)$) {};
  \draw (a2)--(pr) -| ++(0.7,1.3) -| ($(a1)+(-0.55,0)$) -- (a1);
  \draw [dotted] (a1)--(a2);
  \draw [ar] (G.-16)--(pr)    node [edgename] {$a$}   node [type] {$A$};     
  \draw [ar] (G.-90)--(a2)    node [edgename] {$a_k$} node [type] {$A_k$};    
  \draw [ar] (G.-164)--(a1)   node [edgename] {$a_1$} node [type] {$A_1$};    
  \draw [ar] (a2)--++(0,-0.8) node [edgename] {$c_k$} node [type] {$\wn A_k$};
  \draw [ar] (a1)--++(0,-0.8) node [edgename] {$c_1$} node [type] {$\wn A_1$};
  \node [dig]  (dig) at ($(pr)+(2,0)$) {};
  \node [cut]   (cut)at ($(pr)!0.5!(dig)+(0,-0.8)$) {};
  \draw [ar,out=-90,in=180] (pr) to node [edgename,below left]  {$c$} node [type,pos=0.2] {$\oc A$} (cut);
  \draw [ar,out=-90,in=  0] (dig) to node [edgename,right] {$d$} node [type,pos=0.3] {$\wn A^\perp$} (cut);
  \node [etc] (etc) at ($(dig)+(0,1)$) {};
  \draw [ar] (etc) -- (dig) node [edgename,right] {$f$} node [type] {$\wn \wn A^\perp$};

  \draw [->,very thick] (4.2,-0.7) --++(0.8,0) node [below left=-0.1cm] {$cut$};

  \node [proofnet,minimum width=2.5cm] (G1) at ($(G)+(7,0)$) {$G'$};
  \node [princdoor] (pri) at ($(G1.-13)+(0,-0.7)$) {};
  \node [above=-0.1] at ($(pri)+(0.7,0)$) {$\mathbf{B_i}$};
  \node [auxdoor]   (a2i) at ($(G1.-90)+(0,-0.7)$) {};
  \node [auxdoor]   (a1i) at ($(G1.-167)+(0,-0.7)$) {};
  \draw (a2i)--(pri) -| ++(1,1.3) -| ($(a1i)+(-0.55,0)$) -- (a1i);
  \draw [dotted] (a1i)--(a2i);
  \draw [ar] (G1.-13)--(pri) node [edgename] {$a$}   node [type] {$A$};     
  \draw [ar] (G1.-90)--(a2i) node [edgename] {$a_k$} node [type] {$A_k$};
  \draw [ar] (G1.-167)--(a1i)node [edgename] {$a_1$} node [type] {$A_1$};
  \node [princdoor] (pre) at ($(pri)+(0,-0.9)$) {};
  \node [above=-0.1] at ($(pre)+(0.8,0)$) {$\mathbf{B_e}$};
  \node [auxdoor]   (a2e) at ($(a2i)+(0,-0.9)$) {};
  \node [auxdoor]   (a1e) at ($(a1i)+(0,-0.9)$) {};
  \draw (a2e)--(pre) -| ++(1.1,2.3) -| ($(a1e)+(-0.6,0)$) -- (a1e);
  \draw [dotted] (a1e)--(a2e);
  \draw [ar] (pri)--(pre) node [edgename] {$b$}   node [type] {$\oc A$};
  \draw [ar] (a2i)--(a2e) node [edgename] {$b_k$} node [type] {$\wn A_k$};
  \draw [ar] (a1i)--(a1e) node [edgename] {$b_1$} node [type] {$\wn A_1$};
  \node [dig] (dig1) at ($(a1e)+(0,-0.7)$) {};
  \node [dig] (dig2) at ($(a2e)+(0,-0.7)$) {};
  \draw [ar] (a1e)--(dig1) node [edgename] {$e_1$} node [type] {$\wn \wn A_1$};
  \draw [ar] (a2e)--(dig2) node [edgename] {$e_k$} node [type] {$\wn \wn A_k$};
  \draw [ar] (dig1)--++(0,-0.5) node [edgename] {$c_1$} node [type] {$\wn A_1$};
  \draw [ar] (dig2)--++(0,-0.5) node [edgename] {$c_k$} node [type] {$\wn A_k$};
  \node [etc]  (etc) at ($(pri)!0.5!(pre)+(2,0)$) {};
  \node [cut]  (cut) at ($(pre)!0.5!(etc)+(0,-1.1)$) {};
  \draw [ar,out=-90,in=180] (pre) to node [edgename,below left]  {$c$} node [type,pos=0.2] {$\oc \oc A$} (cut);
  \draw [ar,out=-90,in=  0] (etc) to node [edgename,right] {$f$} node [type,pos=0.3] {$\wn \wn A^\perp$} (cut) ;
\end{tikzpicture}
\end{subfigure}
\vspace{1em}

\begin{subfigure}{\textwidth}
\begin{tikzpicture}[baseline=0.3cm]
  \tikzstyle{edgename}=[opacity=0]
  \node [proofnet,minimum width=2.2cm] (G) at (0,0) {$G'$};
  \node [princdoor] (pr) at ($(G.-16)+(0,-0.7)$) {};
  \node [above] at ($(pr)+(0.5,0)$) {$\mathbf{B}$};
  \node [auxdoor]   (a2) at ($(G.-90)+(0,-0.7)$) {};
  \node [auxdoor]   (a1) at ($(G.-164)+(0,-0.7)$) {};
  \draw (a2)--(pr) -| ++(0.7,1.3) -| ($(a1)+(-0.45,0)$) -- (a1);
  \draw [dotted] (a1)--(a2);
  \draw [ar] (G.-16)--(pr)    node [edgename] {$a$}   node [type] {$A$};     
  \draw [ar] (G.-90)--(a2)    node [edgename] {$a_k$} node [type] {$A_k$};    
  \draw [ar] (G.-164)--(a1)   node [edgename] {$a_1$} node [type] {$A_1$};
  \draw [ar] (a2)--++(0,-0.8) node [edgename] {$c_k$} node [type] {$\wn A_k$};    
  \draw [ar] (a1)--++(0,-0.8) node [edgename] {$c_1$} node [type] {$\wn A_1$};
  \node [cont]  (cont)at ($(pr)+(1.7,0)$) {};
  \node [cut]   (cut) at ($(pr)!0.5!(cont)+(0,-0.8)$) {};
  \draw [ar,out=-90,in=180] (pr)   to node [edgename,below left]  {$c$} node [type,pos=0.2] {$\oc A$} (cut);
  \draw [ar,out=-90,in=  0] (cont) to node [edgename,right] {$d$} node [type,pos=0.3] {$\wn A^\perp$} (cut);
  \node [etc] (etcl) at ($(cont)+(110:1)$) {};
  \node [etc] (etcr) at ($(cont)+( 70:1)$) {};
  \draw [ar] (etcl) -- (cont) node [edgename] {$f$} node [type,left] {$\wn A^\perp$};
  \draw [ar] (etcr) -- (cont) node [edgename,right] {$g$} node [type] {$\wn A^\perp$};

  \draw [reduc] (3.1,-0.9) --++ (0.8,0) node [below left=-0.1cm] {$cut$};

  \node [proofnet,minimum width=2.2cm] (Gl) at ($(G)+(5.2,0)$) {${G'}^l$};
  \node [princdoor] (prl) at ($(Gl.-16)+(0,-0.7)$) {};
  \node at ($(prl)+(0.41,0.17)$) {$\mathbf{B_l}$};
  \node [auxdoor]   (a2l) at ($(Gl.-90)+(0,-0.7)$) {};
  \node [auxdoor]   (a1l) at ($(Gl.-164)+(0,-0.7)$) {};
  \draw (a2l)--(prl) -| ++(0.6,1.4) -| ($(a1l)+(-0.27,0)$) -- (a1l);
  \draw [dotted] (a1l)--(a2l);
  \draw [ar] (Gl.-16)--(prl)  node [edgename] {$a^l$} node [type] {$A$};     
  \draw [ar] (Gl.-90)--(a2l)  node [edgename] {$a_k^l$} node [type] {$A_1$};
  \draw [ar] (Gl.-164)--(a1l) node [edgename] {$a_1^l$} node [type] {$A_k$};

  \node [proofnet,minimum width=2.2cm] (Gr) at ($(Gl)+(3.0,0)$) {${G'}^r$};
  \node [princdoor] (prr) at ($(Gr.-16)+(0,-0.7)$) {};
  \node at ($(prr)+(0.41,0.17)$) {$\mathbf{B_r}$};
  \node [auxdoor]   (a2r) at ($(Gr.-90)+(0,-0.7)$) {};
  \node [auxdoor]   (a1r) at ($(Gr.-164)+(0,-0.7)$) {};
  \draw (a2r)--(prr) -| ++(0.6,1.4) -| ($(a1r)+(-0.27,0)$) -- (a1r);
  \draw [dotted] (a1r)--(a2r);
  \draw [ar] (Gr.-16)--(prr)  node [edgename] {$a^r$} node [type] {$A$};     
  \draw [ar] (Gr.-90)--(a2r)  node [edgename] {$a_k^r$} node [type] {$A_1$};    
  \draw [ar] (Gr.-164)--(a1r) node [edgename] {$a_1^r$} node [type] {$A_k$};  

  \node [cont] (c1) at ($(a1l)!0.2!(a1r)+(0,-1.2)$) {};
  \node [cont] (c2) at ($(a2l)!0.35!(a2r)+(0,-1.2)$) {};
  \draw [ar] (a1l)--(c1) node [edgename,left,pos=0.4]               {$b_1^l$} node [type,left,pos=0.4] {$\wn A_1$};
  \draw [ar] (a1r)--(c1) node [edgename,below right=-0.05,pos=0.85] {$b_1^r$} node [type,below right=-0.05,pos=0.85] {$\wn A_1$};
  \draw [ar] (a2l)--(c2) node [edgename,left,pos=0.4]               {$b_k^l$} node [type,left,pos=0.4] {$\wn A_k$};
  \draw [ar] (a2r)--(c2) node [edgename,below right=-0.05,pos=0.85] {$b_k^r$} node [type,below right=-0.05,pos=0.85] {$\wn A_k$};
  \draw [ar] (c1)--++(0,-0.8) node [edgename] {$c_1$} node [type] {$\wn A_1$};
  \draw [ar] (c2)--++(0,-0.8) node [edgename] {$c_k$} node [type] {$\wn A_k$};

  \node [etc] (etcl2) at ($(etcl)+(8,0)$) {};
  \node [etc] (etcr2) at ($(etcr)+(8,0)$) {};
  \node [cut] (cutl) at ($(prl)!0.5!(etcl2)+(0,-1.1)$) {};
  \node [cut] (cutr) at ($(prr)!0.5!(etcr2)+(0,-1.1)$) {};
  \draw [ar,out=-40,in=180] (prl)  to node [edgename,below left,pos=0.85]  {$c^l$} node [type,below,pos=0.85] {$\oc A$} (cutl);
  \draw [ar,out=-90,in=  0] (etcl2) to node [edgename,below,pos=0.87] {$f$} node [type,below,pos=0.85] {$\wn A^\perp$} (cutl);
  \draw [ar,out=-90,in=180] (prr)  to node [edgename,below,pos=0.77]  {$c^r$} node [type,below,pos=0.8] {$\oc A$} (cutr);
  \draw [ar,out=-90,in=  0] (etcr2) to node [edgename,below,pos=0.8] {$g$} node [type,below,pos=0.8] {$\wn A^\perp$} (cutr);
\end{tikzpicture}
\end{subfigure}

\caption{\label{cut_elim_exp_rules}Exponential cut-elimination steps}
\end{figure}
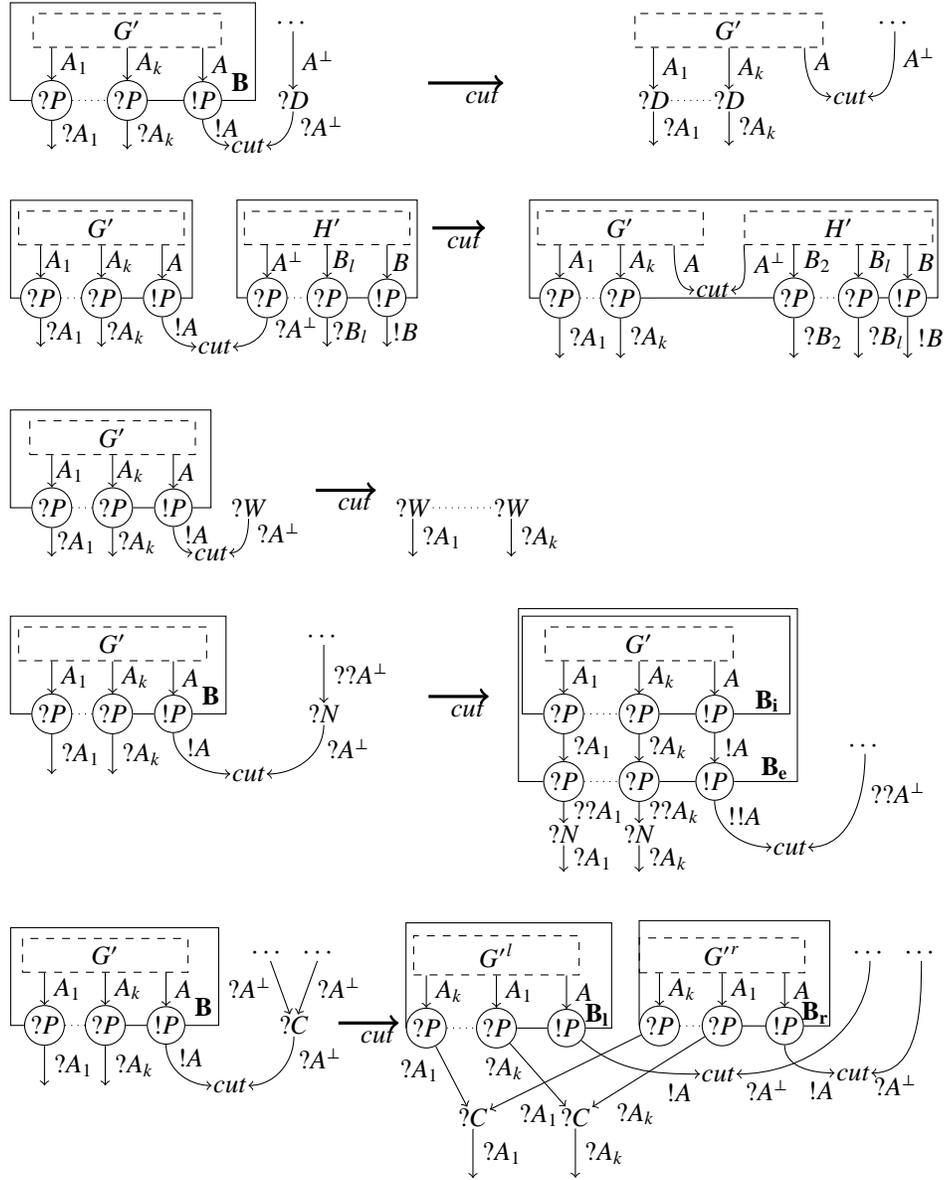
\tikzstyle{edgename}=[opacity=1, midway, left, black]

\begin{lemma}{\cite{girard1996proof}}
  Proof-nets and $LL$ proof-nets are stable under cut-elimination.
\end{lemma}

\subsection{Definition of Context Semantics}
\label{section_def_cont_sem}A common method to prove strong bounds on a rewriting system is to assign a weight $W_G \in \mathbb{N}$ to each term $G$ such that, if $G$ reduces to $H$, $W_G > W_H$. In $LL$, the $\oc P/ \wn C$ step makes the design of such a weight hard: a whole box is duplicated, increasing the number of nodes, edges, cuts,... The idea of context semantics is to define $W_G$ as $|A_G|$ with $A_G$ the edges which {\em appear} during reduction: edges of a net $G_k$ such that $G \cutRel G_1 \cutRel \cdots \cutRel G_k$\footnote{We identify edges which are unaffected by reduction: in Figure~\ref{fig_ex_duplicates}, $k$ only counts once in in $A_G$.}. We can notice that whenever $G \cutRel H$, we have $A_H \subseteq A_G$: if $e \in A_H$ then $e$ is an edge of a proof-net $H_k$ with $H \cutRel H_1 \cdots H_k$ so $G \cutRel H \cutRel H_1 \cdots H_k$. Moreover, this inclusion is strict because the premises of the cut reduced between $G$ and $H$ are in $A_G$ but not in $A_H$. Thus, $|A_G| \geq |A_H|+2$\footnote{We can not deduce that $|A_G| > |A_H|$ because they might be both infinite.}.

However, such a definition of $W_G$ would be impractical: proving a bound on $W_G$ does not seem easier than directly proving a bound on the number of reduction steps. The solution of context semantics is to consider for every edge $e$ of $E_G$, the set $\can{e}$ of {\em residues of $e$}: the elements of $A_G$ ``coming'' from $e$. For instance, in the leftmost proof-net of Figure~\ref{fig_ex_duplicates}, the residues of $e$ are $\{e,e_1,e_2,e_3,e_4\}$. The set $\can{E_G}$ of every edge residue is contained in $A_G$ but is not always equal to it (e.g. the premises of the two cuts in the middle proof-net of Figure~\ref{fig_ex_duplicates} are not residues of any edge of the leftmost proof-net so they are in $A_G$ but not in $\can{E_G}$. Nonetheless, we still have $|\can{E_G}| \geq |\can{E_H}|+1$ whenever $G \cutRel H$, so the length of any path of reduction beginning by $G$ is at most $|\can{E_G}|$ (Theorem~\ref{theo_dallago_edges}).

To bound $W_G$, we characterize edge residues by context semantics paths, simulating cut-elimination. Those paths are generated by {\em contexts} travelling across the proof-net according to some rules. The paths of context semantics in a proof-net $G$ are exactly the paths which are preserved by cut-elimination (such paths are called {\em persistent} in the literature~\cite{danos1995proof}). Computing those paths is somehow like reducing the proof-net. Proving bounds on the number of residues thanks to those paths rather than proving bounds directly on the reduction offers two advantages:
\begin{itemize}
\item Complex properties on proof-nets, which may be hard to manipulate formally, are transformed into existence (or absence) of paths of a certain shape.
\item For every $G \cutRel H$, we have $W_G > W_H$. Thus, the length of {\em any} normalization sequence is bounded by $W_G$. The bounds obtained in this paper do not depend on the reduction strategy.
\end{itemize}

To represent lists we use the notation $[a_1;\cdots;a_n]$.\label{def_arobase} To represent concatenation, we use $@$: $[a_1;\cdots;a_n]@[b_1;\cdots;b_k]$ is defined as $[a_1;\cdots;a_n;b_1;\cdots;b_k]$ and \label{def_insertion}$.$ represents ``push'' ($[a_1;\cdots;a_n].b$ is defined as $[a_1;\cdots;a_n;b]$).\label{def_listrestriction} $|[a_1;\cdots;a_j]|$ refers to $j$, the length of the list. If $X$ is a set, $|X|$ is the number of elements of $X$.

A context is a pair $((e,P),T)$ composed of a  {\em potential edge} $(e,P)$ representing an edge residue ($e$ is a directed edge of the proof-net) and a {\em trace} $T$ used to remember some information about the beginning of the path. This information is necessary to ensure that the paths are preserved by cut-elimination. The following definitions introduce the components of potential edges and traces.

\begin{figure}\centering
  \begin{tikzpicture}
    \draw [very thick,->] (3.2,0) --++ (0.5,0);
    \draw [very thick,->] (8.45,0) --++ (0.5,0) node [above] {$5$};

    \begin{scope}
      \node [princdoor] (bi) at (0,0) {};
      \node [par]       (p) at ($(bi)+(0,0.6)$) {};
      \node [ax]        (a) at ($(p)+(0,0.5)$) {};
      \draw [ar] (a) to [out= -20,in= 60] (p);
      \draw [ar] (a) to [out=-160,in=120] (p);
      \draw (bi) -| ++(0.6,1.24) -| ($(bi)+(-0.4,0)$) -- (bi);
      \node [above] at ($(bi)+(0.45,-0.05)$) {${\mathbf C}$};
      \draw [ar] (p)--(bi) node [edgename] {$e$};
      \node [princdoor] (b) at ($(bi)+(0,-0.8)$) {};
      \node[above] at ($(b)+(0.5,-0.05)$) {${\mathbf B}$};
      \draw (b) -| ++(0.67,2.1) -| ($(b)+(-0.45,0)$) -- (b);
      \draw [ar] (bi)  -- (b) node [edgename] {$a$};
      \node [cont] (c)  at ($(b) +(1.4,0)$) {};
      \node [cut]  (cut)at ($(b)!0.5!(c)+(0,-0.45)$) {};
      \node [weak] (w)  at ($(c) +(120:0.8)$) {};
      \node [auxdoor]  (d)  at ($(c) +( 60:0.8)$) {};
      \node [cont] (ci) at ($(d) +( 90:0.6)$) {};
      \node [tensor] (t) at ($(ci)+(0.9,0)$) {};
      \node [princdoor] (pri) at (d-|t) {};
      \node [ax]   (ab) at ($(ci)!0.5!(t)+(0,0.35)$) {};
      \node [ax]   (ah) at ($(ci)!0.5!(t)+(0,0.6)$) {};
      \draw [ar] (b) to [out=-90,in=180]  node [edgename,pos=0.11,left] {$b$} (cut);
      \draw [ar] (c) to [out=-90,in=  0]  node [edgename,right] {$c$}(cut);
      \draw [ar] (w) -- node [edgename,pos=0.3, below left=-0.1cm] {$j$} (c);
      \draw [ar] (d)  to [out=-90,in=60]  node [edgename,right] {$d$} (c);
      \draw [ar] (ci) to node [edgename] {$f$} (d);
      \draw [ar] (ab) to [out=180,in= 60] (ci);
      \draw [ar] (ah) to [out=180,in=120]  node [edgename,pos=0.5,above left=-0.05] {$g$} (ci);
      \draw [ar] (ab) to [out=-20,in=120] (t);
      \draw [ar] (ah) to [out=-20,in= 60]  node [edgename,pos=0.5,above right=-0.05] {$h$} (t);
      \draw [ar] (t) -- (pri)  node [edgename,right] {$i$};
      \draw [ar] (pri)--++(0,-0.5) node [edgename,right] {$k$};
      \draw (pri)-| ++(0.4,1.4) -| ($(d)+(-0.4,0)$) -- (d) -- (pri);
    \end{scope}

    \begin{scope}[shift={(4.3,0)}]
      \node [princdoor] (bi1) at (0,0) {};
      \node [par]       (p1) at ($(bi1)+(0,0.6)$) {};
      \node [ax]        (a1) at ($(p1)+(0,0.5)$) {};
      \draw [ar] (a1) to [out= -20,in= 60] (p1);
      \draw [ar] (a1) to [out=-160,in=120] (p1);
      \draw (bi1) -| ++(0.6,1.25) -| ($(bi1)+(-0.45,0)$) -- (bi1);
      \draw [ar] (p1)--(bi1) node [edgename,left] {$e_1$};
      \node [above] at ($(bi1)+(0.4,-0.05)$) {${\mathbf C_1}$};
      \node [princdoor] (b1) at ($(bi1)+(0,-0.8)$) {};
      \draw (b1) -| ++(0.65,2.1) -| ($(b1)+(-0.5,0)$) -- (b1);
      \node[above] at ($(b1)+(0.45,-0.05)$) {${\mathbf B_1}$};
      \draw [ar] (bi1)  -- (b1) node [edgename,left] {$a_1$};
      \node [weak] (w)  at ($(b1)+(0.95,0)$) {};
      \node [cut]  (c1) at ($(b1)!0.5!(w)+(0,-0.45)$) {};
      \draw [ar] (b1) to [out=-90,in=180] (c1);
      \draw [ar] (w)  to [out=-90,in=  0] (c1);

      \node [princdoor] (bi2) at ($(bi1)+(1.8,0)$) {};
      \node [above] at ($(bi2)+(0.4,-0.05)$) {${\mathbf C_2}$};
      \node [par]       (p2) at ($(bi2)+(0,0.6)$) {};
      \node [ax]        (a2) at ($(p2)+(0,0.5)$) {};
      \draw [ar] (a2) to [out= -20,in= 60] (p2);
      \draw [ar] (a2) to [out=-160,in=120] (p2);
      \draw (bi2) -| ++(0.6,1.25) -| ($(bi2)+(-0.45,0)$) -- (bi2);
      \draw [ar] (p2)--(bi2) node [edgename] {$e_2$};
      \node [princdoor] (b2) at ($(bi2)+(0,-0.8)$) {};
      \draw (b2) -| ++(0.65,2.1) -| ($(b2)+(-0.5,0)$) -- (b2);
      \node[above] at ($(b2)+(0.45,-0.05)$) {${\mathbf B_2}$};
      \draw [ar] (bi2)  -- (b2) node [edgename] {$a_2$};
      \node [auxdoor]  (d)  at ($(b2) +(1.2,0)$) {};
      \node [cut]  (c2) at ($(b2)!0.5!(d)+(0,-0.5)$) {};
      \draw [ar] (b2) to [out=-90,in=180] (c2);
      \draw [ar] (d) to  [out=-90,in=  0] (c2);

      \node [cont] (ci) at ($(d) +( 90:0.7)$) {};
      \node [tensor] (t) at ($(ci)+(0.7,0)$) {};
      \node [princdoor] (pri) at (d-|t) {};
      \draw (pri)-| ++(0.35,1.6) -| ($(d)+(-0.35,0)$) -- (d) -- (pri);
      \node [ax]   (ab) at ($(ci)!0.5!(t)+(0,0.4)$) {};
      \node [ax]   (ah) at ($(ci)!0.5!(t)+(0,0.7)$) {};
      \draw [ar] (b) to [out=-90,in=180] (cut);
      \draw [ar] (c) to [out=-90,in=  0] (cut);
      \draw [ar] (ci) to (d);
      \draw [ar] (ab) to [out=180,in= 60] (ci);
      \draw [ar] (ah) to [out=180,in=120] (ci);
      \draw [ar] (ab) to [out=-20,in=120] (t);
      \draw [ar] (ah) to [out=-20,in= 60] (t);
      \draw [ar] (t) -- (pri);
      \draw [ar] (pri) --++(0,-0.6) node [edgename] {$k$};
    \end{scope}

    \begin{scope}[shift={(10.2,-0.6)}]
      \node [tensor] (t) at ($(0,0.3)$) {};
      \node [princdoor] (pri) at ($(t)+(0,-0.6)$) {};
      \node [princdoor] (b1) at ($(t)+(140:0.7)$) {};
      \node [princdoor] (b2) at ($(t)+( 40:0.7)$) {};
      \node [par]       (p1) at ($(b1)+(0,0.6)$) {};
      \node [par]       (p2) at ($(b2)+(0,0.6)$) {};
      \node [ax]        (a1) at ($(p1)+(0,0.45)$) {};
      \node [ax]        (a2) at ($(p2)+(0,0.45)$) {};
      \draw (b1) -| ++(0.5,1.16) -| ($(b1)+(-0.37,0)$) -- (b1);
      \draw (b2) -| ++(0.5,1.16) -| ($(b2)+(-0.37,0)$) -- (b2);
      \draw [ar] (a1) to [out= -20,in= 60] (p1);
      \draw [ar] (a1) to [out=-160,in=120] (p1);
      \draw [ar] (a2) to [out= -20,in= 60] (p2);
      \draw [ar] (a2) to [out=-160,in=120] (p2);
      \draw [ar] (p1) -- (b1) node [edgename,right] {$e_3$}; 
      \draw [ar] (p2) -- (b2) node [edgename,right] {$e_4$};
      \draw [ar] (b1) -- (t) node [edgename, below left=-0.1cm] {$h$};
      \draw [ar] (b2) -- (t);
      \draw [ar] (t) --(pri) node [edgename, right] {$i$}; 
      \draw [ar] (pri)--++(0,-0.5) node [edgename,right] {$k$};
      \draw (pri) -| ++(1.15,2.24) -| ($(pri)+(-1,0)$) -- (pri);
    \end{scope}
  \end{tikzpicture}
  \caption{\label{fig_ex_duplicates}\label{path_example}Cut-elimination of a proof-net.}
\end{figure}
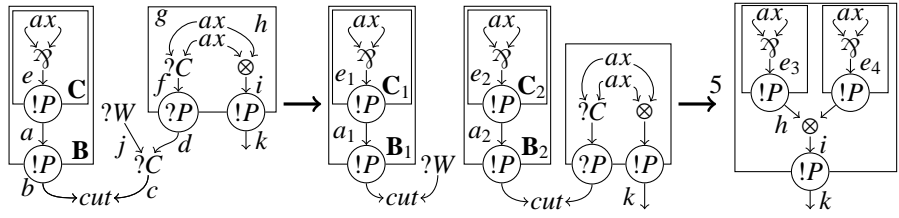

The language $\sig$ of {\it signatures} is defined by induction by the following grammar:\label{def_sig}
\begin{equation*}
\sig = \sige \mid \sigl(\sig) \mid \sigr(\sig) \mid \sigp(\sig) \mid \sign(\sig,\sig)
\end{equation*}
A signature corresponds to a list of choices of premises of $?C$ nodes, to designate a particular residue of a box. The signature $\sigr(t)$ means: ``I choose the right premise, and in the next $?C$ nodes I will use $t$ to make my choices''. The construction $\sign(t,u)$ allows to encapsulate two sequels of choices into one. It corresponds to the digging rule ($\oc\oc A \vdash B \rightsquigarrow \oc A \vdash B$, represented by the $\digLab$ node in proof-nets) which ``encapsulates'' two $\oc$ modalities into one. The $\sigp(t)$ construction is a degenerated case of the $\sign$ construction. Intuitively, $\sigp(t)$ corresponds to $\sign(\varnothing, t)$.

\label{def_potential}A {\em potential} is a list of signatures: a signature corresponds to the duplication of one box, but an element is copied whenever any of the boxes containing it is cut with a $?C$ node. The set of potentials is written $\Pot$. For every edge $e \in \edges{G}$, we define $Pot(e)$ as $\Set{(e,P)}{P \in Pot \text{ and }|P|=\depth{e}}$ such pairs are named {\em potential edges}.

Potentials are used to represent residues. For instance, the residues of $e$ in Figure~\ref{fig_ex_duplicates}, ($e$, $e_1$, $e_2$, $e_3$ and $e_4$) are respectively represented by the potential edges $(e,[\sige;\sige])$, $(e,[\sigl(\sige);\sige])$, $(e,[\sigr(\sige);\sige])$, $(e,[\sigr(\sige);\sigl(\sige)])$ and $(e,[\sigr(\sige);\sigr(\sige)])$.

\label{def_traceelement}A {\em trace element} is one of the following symbols:  $\parr_l , \parr_r, \otimes_l, \otimes_r, \forall, \exists,!_t,?_t$ with $t$ a signature. A trace element means ``I have crossed a node with this label, from that premise to its conclusion''.\label{def_trace}  A {\em trace } is a non-empty list of trace elements. The set of traces is $\trace$. A trace is a memory of the path followed, up to cut-eliminations. \label{def_dual_trace}We define duals of trace elements: $\parr_l^\perp = \otimes_l$, $!_t^\perp = ?_t$,... and extend the notion to traces by $[a_1; \cdots ; a_k]^\perp = [a_1^\perp; \cdots ; a_k^\perp]$.

\label{def_context}A {\em context} is a tuple $((e,P),T)$ with $(e,P)$ a potential edge and $T \in \trace$. It can be seen as a state of a token that travels around the net. It is located on edge $e$ (more precisely its residue corresponding to $P$) and carries information $T$ about its past travel. The set of contexts of $G$ is written $\context{G}$. We extend the mapping $(\_)^{\perp}$ on contexts by $((e,P),T)^{\perp}=((\overline{e},P),T^{\perp})$.

\nvar{\sepColonnes}{3cm}
\nvar{\sepRegles}{3.5cm}
\begin{figure}\centering
  \begin{tikzpicture}
    \node [cut] (n) at (0, 0) {};
    \draw [ar,out=-100,in= 20] ($(n)+( 30:0.7)$) to node [edgename, below right=-0.08cm] {$f$} (n.0);
    \draw [ar,out= -80,in=160] ($(n)+(150:0.7)$) to node [edgename, below left=-0.05cm] {$e$} (n.180);
    
    \node [ax] (n) at (\sepColonnes,0.2) {};
    \draw [ar,out=-160,in= 80] (n.180) to node [edgename,above left=-0.07cm] {$g$} ($(n)+(-150:0.7)$);
    \draw [ar,out= -20,in=100] (n.0) to node [edgename,above right=-0.07cm]{$h$} ($(n)+( -30:0.7)$);
    
    \draw (-1,-1) node [right] (rules) {
      $\begin{array}{lllll}
         ((e,P),&\hspace{-0.2em}T) &\hspace{-0.2em} \rightsquigarrow&\hspace{-0.2em} ((\overline{f},P),&\hspace{-0.2em}T)\\
         ((\overline{g},P),&\hspace{-0.2em}T) &\hspace{-0.2em}\rightsquigarrow&\hspace{-0.2em} ((h,P),&\hspace{-0.2em}T) 
       \end{array}$
     };
     
     \begin{scope}[shift={(0,-2.4)}]
       \node [par] (n) at (0,0) {};
       \draw[ar] ($(n)+(120:0.5)$)--(n) node [edgename, left] {$a$};
       \draw[ar] ($(n)+( 60:0.5)$)--(n) node [edgename, right] {$b$};
       \draw[ar] (n) --++(-90:0.4) node [edgename, right] {$c$};
       
       \node [tensor] (n) at (\sepColonnes,0) {};
       \draw[ar] ($(n)+(120:0.5)$)--(n) node [edgename, left] {$e$};
       \draw[ar] ($(n)+ (60:0.5)$)--(n) node [edgename, right] {$f$};
       \draw[ar] (n) --++(-90:0.4) node [edgename, right] {$g$};
       
       \draw (-1,-1.5) node [right] (rules) {
         $\begin{array}{lllll}
            ((a,P),&\hspace{-0.2em}T) &\hspace{-0.2em}\rightsquigarrow &\hspace{-0.2em}((c,P),&\hspace{-0.2em} T.\parr_l) \\
            ((b,P),&\hspace{-0.2em}T) &\hspace{-0.2em}\rightsquigarrow &\hspace{-0.2em}((c,P),&\hspace{-0.2em} T.\parr_r) \\
            ((e,P),&\hspace{-0.2em}T) &\hspace{-0.2em}\rightsquigarrow &\hspace{-0.2em}((g,P),&\hspace{-0.2em} T.\otimes_l) \\
            ((f,P),&\hspace{-0.2em}T) &\hspace{-0.2em}\rightsquigarrow &\hspace{-0.2em}((g,P),&\hspace{-0.2em} T.\otimes_r) 
          \end{array}$
        };
      \end{scope}
      \begin{scope}[shift={(6.5,0)}]
        \node[forall] (fa) at (0,0.1) {};
        \draw[ar] ($(fa)+(0,0.45)$)--(fa) node [midway, left] {$e$};
        \draw[ar] (fa) --++(0,-0.45) node [midway, left] {$f$};
        
        \node[exists] (ex) at (\sepColonnes,0.1) {};
        \draw[ar] ($(ex)+(0,0.45)$) -- (ex) node [midway, left] {$g$};
        \draw[ar] (ex) --++(0,-0.45) node [midway, left] {$h$};
        \draw (-1,-1) node [right] (rules) {
          $\begin{array}{lllll}
             ((e,P),&\hspace{-0.2em}T) &\hspace{-0.2em}\rightsquigarrow&\hspace{-0.2em} ((f,P),&\hspace{-0.2em} T.\forall) \\
             ((g,P),&\hspace{-0.2em}T) &\hspace{-0.2em}\rightsquigarrow&\hspace{-0.2em} ((h,P),&\hspace{-0.2em} T.\exists)
           \end{array}$
         };
         
       \end{scope}
       \begin{scope}[shift={(6,-2.35)}]
         \node[der] (n) at (0,0) {};
         \draw[ar] ($(n)+(0,0.45)$)--(n) node [edgename, left] {$e$};
         \draw[ar] (n) --++(0,-0.45) node [edgename, left] {$f$};
         
         \node[cont] (n) at (\sepColonnes,0) {};
         \draw[ar] ($(n)+(120:0.45)$)--(n) node [edgename] {$g$};
         \draw[ar] ($(n)+( 60:0.45)$)--(n) node [edgename, right] {$h$};
         \draw[ar] (n) --++(-90:0.45) node [edgename] {$i$};
         
         \draw (-1,-1.3) node [right] (rules) {
           $\begin{array}{lllll}
              ((e,P),&\hspace{-0.2em} T )    &\hspace{-0.2em}\rightsquigarrow &\hspace{-0.2em}((f,P),&\hspace{-0.2em} T.?_{\sige})\\
              ((g,P),&\hspace{-0.2em} T.?_t) &\hspace{-0.2em}\rightsquigarrow &\hspace{-0.2em}((i,P),&\hspace{-0.2em} T.?_{\sigl(t)})\\
              ((h,P),&\hspace{-0.2em} T.?_t) &\hspace{-0.2em}\rightsquigarrow &\hspace{-0.2em}((i,P),&\hspace{-0.2em} T.?_{\sigr(t)})\\
            \end{array}$
          };
        \end{scope}

        \begin{scope}[shift={(0,-5.8)}]
          \node [dig] (n) at (2,0) {};
          \draw[ar] ($(n)+(90:0.5)$) -- (n) node [edgename] {$g$};
          \draw[ar] (n) --++(-90: 0.5) node [edgename] {$h$};
          
          
          \draw (-1,-1.15) node [right] (rules) {
            $\begin{array}{lllll}
               ((g,P),&\hspace{-0.2em} T.\wn_{t_1}.\wn_{t_2}) &\hspace{-0.2em}\rightsquigarrow&\hspace{-0.2em} ((h,P),&\hspace{-0.2em} T.\wn_{\sign(t_1,t_2)})\\ 
               ((g,P),&\hspace{-0.2em} [\wn_t]) &\hspace{-0.2em}\rightsquigarrow&\hspace{-0.2em} ((h,P),&\hspace{-0.2em} [\wn_{\sigp(t)}]) 
             \end{array}$
           };
         \end{scope}

         \begin{scope}[shift={(6.8,-5.3)}]
           \node [auxdoor] (n) at (0, 0) {};
           \draw[ar] ($(n)+(0,0.5)$) -- (n) node [edgename] {$e$};
           \draw[ar] (n) --++(0,-0.5)  node [edgename] {$f$};
           \node [princdoor] at (\sepColonnes,0)  (m) {};
           \draw[ar] ($(m)+(0,0.5)$) -- (m) node [edgename] {$g$};
           \draw[ar] (m) --++(0,-0.5) node [edgename] {$h$};
           \draw (n) -| ++(-.5,.6) (n)--(m) -| ++(.5,.6);
           
           \draw (-1,-1.4) node [right] (rules) {
             $\begin{array}{lllll}
                ((e, P.t),&\hspace{-0.2em} T) &\hspace{-0.2em}\rightsquigarrow \hspace{-0.2em}& ((f,P),&\hspace{-0.2em} T.?_t) \\
                ((g, P.t),&\hspace{-0.2em} T) &\hspace{-0.2em}\rightsquigarrow \hspace{-0.2em}& ((h,P),&\hspace{-0.2em} T.!_t) \\
                ((\overline{f},P),&\hspace{-0.2em} [!_t]) &\hspace{-0.2em}\hookrightarrow \hspace{-0.2em}& ((h, P),&\hspace{-0.2em} [!_t]) \\
              \end{array}$};
          \end{scope}
        \end{tikzpicture}
        \caption{\label{exponential_context_semantic} Rules of the context semantics}
\end{figure}
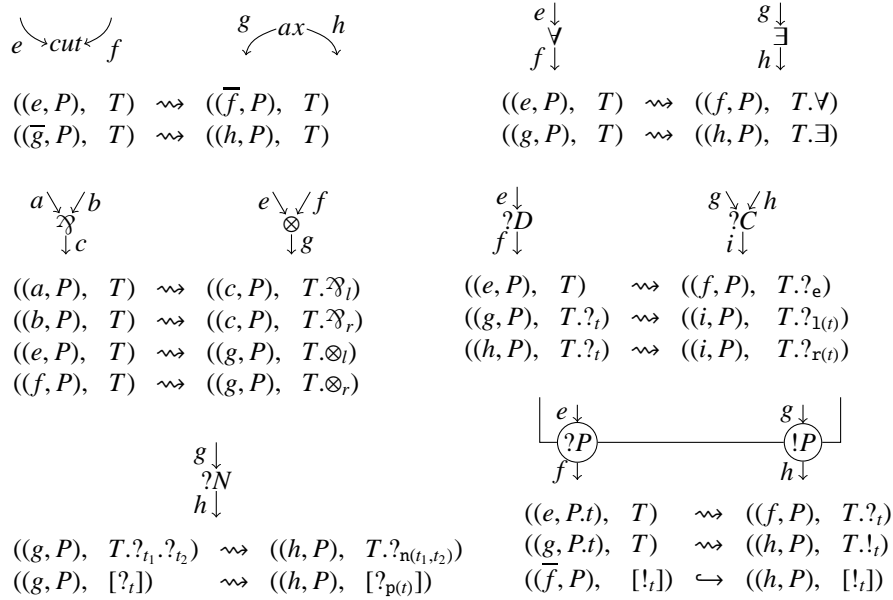

\label{def_nojump}\label{def_onlyjump}The nodes define two relations $\noJump$ and $\onlyJump$ on contexts (Figure~\ref{exponential_context_semantic}). For any rule $C \noJump D$ presented in Figure~\ref{exponential_context_semantic}, we also define the dual rule $D^{\perp} \noJump C^{\perp}$. \label{def_csrel}We define $\csRel$ as the union of $\noJump$ and $\onlyJump$. In other words, $\csRel$ is the smallest relation on contexts including every instance of $\noJump$ rules in Figure~\ref{exponential_context_semantic} together with every instance of their duals and every instance of the $\onlyJump$ rule.

The rules are sound: if $((e,P),T) \csRel ((f,Q),U)$, then $\depth{e}=|P|$ iff $\depth{f}=|Q|$. Those relations are deterministic. In particular, if $C=((e,P),T.\wn_t)$ with $e$ the premise of a $\digLab$ node or $C=((\overline{\sigma_i(B)},P),T.\oc_t)$, the context $D$ such that $C \csRel D$ depends on the size of $T$: there is a rule in the case $T=[]$ and another in the case $T \neq []$. Let us notice that $\noJump$ is injective (Lemma~\ref{lemma_nojump_injective}). It is not the case for the $\csRel$ relation. Indeed, if $B$ is a box with two auxiliary doors then, for every potential $P$ and signature $t$, we have $((\overline{\sigma_1(B)},P),[\oc_t]) \csRel ((\sigma(B),P),[\oc_t])$ and  $((\overline{\sigma_2(B)},P),[\oc_t]) \csRel ((\sigma(B),P),[\oc_t])$.

\begin{lemma}{\label{lemma_nojump_injective}}
  If $C_1 \noJump D$ and $C_2 \noJump D$ then $C_1=C_2$
\end{lemma}

Finally, we can observe that for every sequence $((e_1,P_1),T_1) \noJump ((e_2,P_2),T_2) \noJump \cdots ((e_n,P_n),T_n)$, the sequence of directed edges $e_1,\cdots, e_n$ is a path (i.e. the head of $e_i$ is the same node as the tail of $e_{i+1}$). The $\onlyJump$ relation breaks this property as it is non-local: it deals with two non-adjacent edges. 
The $\csRel$-paths represent the reduction of a proof-net because they are stable along reduction. For example, the path in the first proof-net of Figure~\ref{path_example}:
%

\begin{align*}
&((e,[\sigr(\sige);\sigl(\sige)]),[\parr_r]) \csRel ((a,[\sigr(\sige)]),[\parr_r;\oc_{\sigl(\sige)}]) \csRel ((b,[]),[\parr_r;\oc_{\sigl(\sige)};\oc_{\sigr(\sige)}]) \csRel\\
& ((\overline{c},[]),[\parr_r; \oc_{\sigl(\sige)};\oc_{\sigr(\sige)}]) \csRel ((\overline{d},[]),[\parr_r; \oc_{\sigl(\sige)};\oc_{\sige}]) \csRel ((\overline{f},[\sige]),[\parr_r; \oc_{\sigl(\sige)}]) \csRel \\
& ((\overline{g},[\sige]),[\parr_r; \oc_{\sige}]) \hspace{-0.15em}\csRel\hspace{-0.15em} ((h,[\sige]),[\parr_r; \oc_{\sige}]) \hspace{-0.15em}\csRel\hspace{-0.15em} ((i,[\sige]),[\parr_r; \oc_{\sige}; \otimes_r]) \hspace{-0.15em}\csRel\hspace{-0.15em} ((k,[]),[\parr_r;\oc_{\sige};\otimes_r;\oc_{\sige}])
\end{align*}
becomes the path $((e_3,[\sige;\sige]),[\parr_r]) \csRel ((h,[\sige]),[\parr_r;\oc_{\sige}]) \csRel ((i,[\sige]),[\parr_r;\oc_{\sige};\otimes_r]) \csRel ((k,[]),[\parr_r;\oc_{\sige};\otimes_r;\oc_{\sige}])$ in the third proof-net of Figure~\ref{path_example}.

\subsection{Dal Lago's weight theorem}\label{section_capturing_residues}
As written earlier, potential edges are intended to ``correspond'' to residues. To precise this correspondence we first define, for every $G \cutRel H$ step, a partial mapping $\pi_{G \rightarrow H}(\_)$ from $\pot{H}$ to $\pot{G}$. For edges $e$ which are not affected by the step, we can define $\pi_{G\rightarrow H}(e,P)=(e,P)$. If the reduction step is a $\oc P / \wn C$ step (bottom of Figure~\ref{cut_elim_exp_rules}) and $e \in E_H$ is contained in $B_l$ (respectively $B_r$) then we define $\pi_{G \rightarrow H}(e,P.t@Q)=(e,P.\sigl(t)@Q)$ (respectively $(e,P.\sigr(t)@Q)$) with $|P|=\partial(B)$. If the reduction step is a $\oc P / \wn N$ step, $e$ is immediately contained in $B_e$ and $f$ is contained in $B_i$, then we define $\pi_{G \rightarrow H}(e,P.t)=(e,P.\sigp(t))$ and $\pi_{G \rightarrow H}(f,P.t.u@Q)=(f,P.\sign(t,u)@Q)$. If the reduction step is a $\oc P / \wn D$ step and $e \in E_H$ belongs to $G'$ then we define $\pi_{G \rightarrow H}(e,P@Q)=(e,P.\sige @Q)$. We do not detail every case and exception because the only purpose of this definition in this paper is to guide intuition. A more precise definition of the mapping is given (on contexts) in Definition~12 of~\cite{perrinelMegathese}.

Let us suppose $G_1 \cutRel G_2 \cdots \cutRel G_k$ and $e'$ an edge of $G_k$. A potential edge $(e,P) \in \pot{\edges{G_1}}$ {\em corresponds} to $e'$ if $\pi_{G_1 \rightarrow G_2} \circ \cdots \circ \pi_{G_{k-1} \rightarrow G_k} (e',[\sige;\cdots;\sige]) = (e,P)$.

Let $e \in \edges{G}$, there are potential edges in $\pot{e}$ which do not correspond to residues of $e$. For instance, in Figure~\ref{fig_ex_duplicates} $a$ has three residues: $a$, $a_1$ and $a_2$. The residue $a_1$ is obtained by choosing the left box during the duplication of box $B$, so it is represented by $(a,[\sigl(\sige)])$. Similarly, $a$ and $a_2$ are represented by $(a,[\sige])$ and $(a,[\sigr(\sige)])$. However, $(a,[\sigr(\sigl(\sige))])$ does not represent any residue. The potential node $(a,[\sigr(\sigl(\sige))])$ means that whenever the box $B_2$ is cut with a $\contLab$ node, we choose the left box. But this situation never happens. It can be observed by the following path: 
\begin{align*}
((\sigma(B),[]),[\oc_{\sigr(\sigl(\sige))}]) \csRel ((\overline{c},[]),[\oc_{\sigr(\sigl(\sige))}]) \csRel ((\overline{d},[]),[\oc_{\sigl(\sige)}]) \csRel ((k,[]),[\oc_{\sigl(\sige)}]) \not \csRel
\end{align*}
The $\sigl(\_)$ has not been used because we did not encounter a second $\wn C$ node. On the contrary, the signatures corresponding to residues are entirely used:
\begin{align*}
&((\sigma(B),[]),[\oc_{\sige}]) \csRel^0 ((b,[]),[\oc_{\sige}]) \\
&((\sigma(B),[]),[\oc_{\sigl(\sige)}]) \csRel^2 ((\overline{j},[]),[\oc_{\sige}]) \hspace{2em} ((\sigma(B),[]),[\oc_{\sigr(\sige)}]) \csRel^2 ((\overline{d},[]),[\oc_{\sige}])
\end{align*}

This is why, in the absence of $\digLab$ nodes, we define the {\em canonical potentials} of $e \in B_{\depth{e}} \subset \cdots B_2 \subset B_1$ as the potential edges $(e,[p_1;\cdots;p_{\depth{e}}])$ such that, for $1 \leq i \leq \depth{e}$, we have $((\sigma(B_i),[p_1;\cdots;p_{i-1}]),[\oc_{p_i}]) \csRel^* ((\_,\_),[\oc_{\sige}])$ (throughout the article, we use $\_$ to denote an object whose name and value are not important to us, for example $C \csRel \_$ means $\exists D \in \context{G}, C \csRel D$).

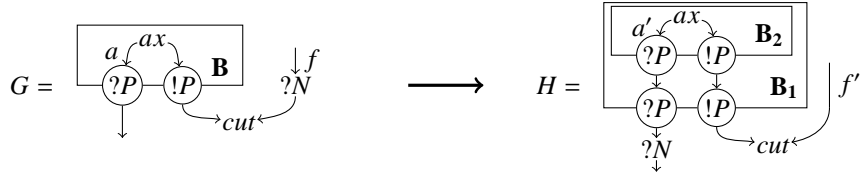
\begin{figure}\centering
  \begin{tikzpicture}
    \draw [->,very thick] (3,0) --++ (1,0); 

    \node [princdoor] (p) at (0,0) {};
    \node [auxdoor]   (a) at ($(p)+(-0.8,0)$) {};
    \node [ax]        (ax)at ($(a)!0.5!(p)+(0,0.6)$) {};
    \draw [ar,out=-170,in= 80] (ax) to node [pos=0.6,left] {$a$} (a);
    \draw [ar,out= -10,in=100] (ax) to (p);
    \draw (p) -| ++ (0.8,0.8) -| ($(a)+(-0.6,0)$) -- (a)-- (p);
    \node [above] at ($(p)+(0.5,0)$) {$\mathbf{B}$};
    \node at ($(p)+(-2,0)$) {$G=$}; 
    \draw [ar] (a) --++ (0,-0.7);
    \node [dig] (d) at ($(p)+(1.5,0)$) {};
    \node [cut] (c) at ($(p)!0.5!(d)+(0,-0.5)$) {};
    \draw [ar,out=-90,in=170] (p) to (c);
    \draw [ar,out=-90,in= 10] (d) to (c);
    \draw [ar] ($(d)+(0,0.5)$) to node [edgename,right] {$f$} (d);

    \begin{scope}[shift={(7.1,0.4)}]
      \node [princdoor] (p) at (0,0) {};
      \node at ($(p)+(-2.1,-0.4)$) {$H=$}; 
      \node [auxdoor]   (a) at ($(p)+(-0.8,0)$) {};
      \node [ax]        (ax)at ($(a)!0.5!(p)+(0,0.5)$) {};
      \draw [ar,out=-170,in= 80] (ax) to node [pos=0.6,left] {$a'$} (a);
      \draw [ar,out= -10,in=100] (ax) to (p);
      \draw (p) -| ++ (1,0.65) -| ($(a)+(-0.6,0)$) -- (a) -- (p);
      \node [above] at ($(p)+(0.7,0)$) {$\mathbf{B_2}$};
      \node [princdoor] (p1) at ($(p)+(0,-0.7)$) {};
      \node [auxdoor]   (a1) at (a|-p1) {};
      \draw (p1) -| ++ (1.2,1.4) -| ($(a1)+(-0.7,0)$) -- (a1) -- (p1);
      \node [above] at ($(p1)+(0.9,0)$) {$\mathbf{B_1}$};
      \draw [ar] (a)--(a1);
      \draw [ar] (p)--(p1);
      \node [dig] (d) at ($(a1)+(0,-0.55)$) {};
      \draw [ar] (a1)--(d);
      \draw [ar] (d) --++ (0,-0.3);
      \coordinate (etc) at ($(p1)+(1.5,0)$) ;
      \node [cut] (c) at ($(p1)!0.5!(etc)+(0,-0.5)$) {};
      \draw [ar,out=-90,in=170] (p1) to (c);
      \draw [ar,out=-90,in= 10] (etc) to (c);
      \draw ($(etc)+(0,0.6)$) -- (etc) node [edgename,right] {$f'$};
    \end{scope}
  \end{tikzpicture}
  \caption{\label{fig_intui_digging}The potential edge $(a,[\sign(t_2,t_1)]$ corresponds to $(a',[t_1;t_2])$. }
\end{figure}

Now we will consider what happens when $\wn N$ nodes are allowed. Let us consider the node $a$ in Figure~\ref{fig_intui_digging}. The residues of $a$ are exactly the residues of $a'$ and $a$ itself, and ``$(a',[t_1;t_2])$ corresponds to a residue of $a$'' is successively equivalent to: 
\begin{align*}
&\left \{ \hspace{-0.4em}\begin{array}{ll}
((\sigma(B_2),[t_1]),[\oc_{t_2}])\hspace{-0.4em}&\csRel^* ((\_,\_),[\oc_{\sige}])\\
((\sigma(B_1),[]),[\oc_{t_1}])   &\csRel^* ((\_,\_),[\oc_{\sige}])\\
\end{array}\right.
&&\hspace{-1.2em}\Leftrightarrow\left \{ \hspace{-0.4em} \begin{array}{ll}
((\overline{f'},[]),[\oc_{t_2};\oc_{t_1}]) \hspace{-0.4em} &\csRel^* ((\_,\_),[\oc_{\sige}])\\
((\overline{f'},[]),[\oc_{t_1}]) &\csRel^* ((\_,\_),[\oc_{\sige}])
\end{array}\right.\\
\Leftrightarrow&\left \{ \hspace{-0.4em} \begin{array}{ll}
((\overline{f},[]),[\oc_{t_2};\oc_{t_1}])\hspace{-0.4em} &\csRel^* ((\_,\_),[\oc_{\sige}])\\
((\overline{f},[]),[\oc_{t_1}]) &\csRel^* ((\_,\_),[\oc_{\sige}])
\end{array}\right.
&&\hspace{-1.2em}\Leftrightarrow\left \{ \hspace{-0.4em} \begin{array}{ll}
((\sigma(B),[]),[\oc_{\sign(t_2,t_1)}]) \hspace{-0.4em} &\csRel^* ((\_,\_),[\oc_{\sige}])\\
((\sigma(B),[]),[\oc_{\sigp(t_1)}]) &\csRel^* ((\_,\_),[\oc_{\sige}])
\end{array}\right.
\end{align*}
Thus, $(a,[\sign(t_2,t_1)])$ corresponds to a residue of $a$ iff both $\sign(t_2,t_1)$ and $\sigp(t_1)$ are entirely used by their $\csRel$-paths. Let us notice that a box may encounter several $\digLab$ nodes during cut-elimination. To check every case, we define a relation $\simpl$ on signatures such that, in particular, $\sign(t_2,t_1) \simpl \sign(t_2,t_1)$ and $\sign(t_2,t_1) \simpl \sigp(t_1)$.

\begin{definition}\label{def_standard}A signature is {\em standard} if it does not contain the constructor $\sigp$.\label{def_quasistandard} A signature $t$ is {\em quasi-standard} iff for every subtree $\sign(t_1,t_2)$ of $t$, the signature $t_2$ is standard. 
\end{definition}

\label{def_simpl}The binary relation $\simpl$ on $\sig$ is defined by induction as follows: $\sige     \simpl \sige$ and, if we suppose that $t \simpl t'$, then $\sigl(t)\simpl \sigl(t')$, $\sigr(t)\simpl \sigr(t')$, $\sigp(t) \simpl \sigp(t')$, $\sign(u,t)  \simpl \sigp(t')$ and $\sign(t,u) \simpl  \sign(t',u)$. We write that $t'$ is a {\em simplification} of $t$, when $t  \simpl t'$.\label{def_sqsubset} We write $t \simplStrict t'$ for ``$t \simpl t'$ and $t \neq t'$''. We can observe that $\simpl$ is an order and $\simplStrict$ a strict order.

\begin{lemma}[\cite{perrinelMegathese}]\label{lemma_total_order}
  Let $t \in \sig$, then $\simpl$ is a total order on $\Set*{ u \in \sig}{ t \simpl u }$.
\end{lemma}

\begin{definition}\label{def_quasistandard_context} A context $((e,P),[\oc_t]@T)$ is said {\em quasi-standard} if $t$ is quasi-standard and every signature in $P$ and $T$ is standard.
\end{definition}

If $u \compl t$ with $t$ standard $((\sigma(B),P),[\oc_t])$ is quasi-standard, and quasi-standard contexts are stable by $\csRel$~\cite{perrinelMegathese}. So every context we study in this work is quasi-standard.

We capture the notion of residue by {\em canonical potentials}. The definition of canonical potentials relies on {\em copies}. A copy represents the choices for one box, a canonical potential for an element $x$ is a list of copies: one copy for each box containing $x$.

\begin{definition}{\label{def_copycontext}}  A {\em copy context} is a context of the shape $((e,P),[\oc_t]@T)$ such that for every $u \compl t$, there exists a path of the shape $((e,P),[\oc_u]@T) \csRel^* ((\_,\_),[\oc_{\sige}]@\_)$.

\label{def_copy}Let $(B,P) \in \pot{B_G}$, the set $\cop{B,P}$ of {\em copies} of $(B,P)$ is the set of standard signatures $t$ such that $((\sigma(B),P),[\oc_t])$ is a copy context.
\end{definition}

For instance, in Figure~\ref{fig_ex_duplicates}, the copies of $(B,[])$ are $\sige$, $\sigl(\sige)$ and $\sigr(\sige)$ which respectively corresponds to $B$ itself, $B_1$ and $B_2$. So $(C,[\sigl(\sige)])$ and $(C,[\sigr(\sige)])$ correspond respectively to $C_1$ and $C_2$. We can notice that $C_2$ is duplicated while $C_1$ can not be duplicated, in terms of context semantics $((\sigma(C),[\sigr(\sige)]),[\oc_{\sigl(\sige)}]) \csRel^5 ((\overline{g},[\sige]),[\oc_{\sige}])$ so $\sigl(\sige)$ is a copy of $(C,[\sigr(\sige)])$ while  $((\sigma(C),[\sigl(\sige)]),[\oc_{\sigl(\sige)}]) \csRel^3 ((\overline{j},[]),[\oc_{\sigl(\sige)};\oc_{\sige}]) \not \csRel$ so $\sigl(\sige)$ is not a copy of $(C,[\sigl(\sige)])$. 

\begin{definition}\label{def_canonical}Let $x$ be an edge (resp. box, node) of $G$ with $x \in B_{\depth{x}} \subset ... \subset B_1 $. The set $\can{x}$ of {\it canonical edges} (resp. box, node) for $x$  is the set of tuples $(x,[p_{1} ;  ... ; p_{\depth{x}}])$ with $p_1,\cdots,p_{\depth{x}}$ signatures such that:
  \begin{equation*}
    \forall 1\leq i\leq \depth{x}, p_i \in \cop{B_i, [p_{1}; \cdots ; p_{i-1}]}
  \end{equation*}
\end{definition}

For instance, in the proof-net of Figure~\ref{fig_ex_duplicates}, we have $\cop{B,[]}=\{\sige,\sigl(\sige),\sigr(\sige)\}$, $\cop{C,[\sige]}=\cop{C,[\sigl(\sige)]}=\{\sige\}$ and $\cop{C,[\sigr(\sige)]}=\{\sige,\sigl(\sige),\sigr(\sige)\}$. So, by definition, $\can{e}=\{e\} \times \{[\sige;\sige],[\sigl(\sige);\sige],[\sigr(\sige);\sige],[\sigr(\sige);\sigl(\sige)];[\sigr(\sige);\sigr(\sige)]\}$. Those canonical potentials correspond respectively to $e$, $e_1$, $e_2$, $e_3$ and $e_4$. We can notice that $|\can{e}|=5$. In the middle proof-net of this Figure, we have $\can{e_1}=\{(e_1,[\sige;\sige])\}$ and $\can{e_2} = \{(e_2,[\sige;\sige]), (e_2,[\sige;\sigl(\sige)]), (e_2,[\sige; \sigr(\sige)])\}$ so $|\can{e_1}|+|\can{e_2}|=4 < |\can{e}|$ (the number of residues decreases because there is no edge corresponding to $(e,[\sige,\sige])$ in the middle proof-net).

The set of canonical edges of $G$ is represented by $\can{\edges{G}}$. \label{remarque_can_same}Let us notice that the canonical edges for $e$ only depend on the boxes containing $e$: if $e$ and $f$ are contained in the same boxes then $\can{e}=\Set*{(e,P)}{(f,P) \in \can{f}}$.

\begin{definition}\label{def_wg}
  For any proof-net $G$, we define $W_G =\left|\can{\edges{G}}\right| \in \mathbb{N} \cup \{\infty\}$.
\end{definition}

In~\cite{perrinelMegathese}, to prove that $W_G$ is a bound on reduction, we first build a strict injection from the canonical nodes of $H$ to the canonical nodes of $G$. This injection is based on a mapping from contexts of $H$ to contexts of $G$ which preserves $\csRel$-paths. Then we prove that the number of canonical nodes is bounded by $W_G$

\begin{definition}\label{def_pi}
 Let us suppose that $G \cutRel H$ then we defined (in~\cite{perrinelMegathese}) a partial mapping $\pi(\_)$ from $Cont_H$ to $Cont_G$ such that, whenever $\pi(C)$ and $\pi(D)$ are defined,
  \begin{equation*}
    C \csRel^* D  \Rightarrow \pi(C) \csRel^* \pi(D) \hspace{3em}
    C \csRel^+ D  \Leftarrow \pi(C) \csRel^+ \pi(D)
  \end{equation*}
\end{definition}

Theorem~\ref{theo_dallago_edges} is a slight variation of the Lemma 6 of Dal Lago in~\cite{lago2006context}. This result allows to prove strong complexity bounds for several systems.

\begin{theorem}[\cite{perrinelMegathese}]\label{theo_dallago_edges}
  If $G$ is a normalizing proof-net, then $W_G \in \mathbb{N}$. The length of any path of reduction, and the size of any proof-net of the path, is bounded  by $W_G$.
\end{theorem}

Execution time depends on the implementation of proof-nets and $cut$-elimination. In a basic implementation based on graphs, every step can be done in constant time except for the box rules, which can be done in a time linear in the size of the box, so linear in the size of the proof-net. Thus, according to Theorem~\ref{theo_dallago_edges}, the execution time of $G$ is in $O(W_G^2)$. The complexity classes we study in this article are stable by polynomial. Thus, to establish the soundness of a $LL$ subsystem with respect to polynomial time/elementary time, it is enough to prove a polynomial/elementary bound on $W_G$.

\begin{lemma}[\cite{perrinelMegathese}]\label{lemma_acyclicity}
  Let $G$ be a normalizing proof-net, there is no path of the shape $((e,P),[\oc_t]) \csRel^+ ((e,P),[\oc_u])$ with $(e,P)$ a canonical edge.
\end{lemma}

\section{Paths criteria for elementary time}\label{chapter_3}
\subsection{History and motivations}
A stratification refers to a restriction of a framework, which forbids the contraction (or identification) of two subterms belonging to two morally different ``strata''. Stratification restrictions might be applied to several frameworks (naive set theory, linear logic, lambda calculus and recursion theory) to entail coherence or complexity properties~\cite{baillot2010linear}. To define a stratification condition on Linear Logic we define, for every proof-net $G$, a {\em stratification relation} $>$ between the boxes of $G$. Then, we consider that $B$ belongs to a higher stratum than $C$ if $B (>)^+ C$. The relation $>$ must be defined such that there exists a function $f$ such that:
\begin{equation}\label{eq_f_dup}
  \left|\cop{B,P}\right| \leq f \left( \max_{\substack{B > C\\(B,P) \in \pot{B}}}\left| \cop{C,Q} \right|, \left|E_G\right| \right)
\end{equation}

\label{definition_stratugen}One says that $G$ is $>$-stratified if $>$ is acyclic. In this case, for every box $B$ of $G$, we define the $>$-stratum of $B_1$ (written $\stratu{>}{B_1}$) as the greatest $i \in \mathbb{N} \cup \{\infty \}$ such that there exists $B_2,\cdots,B_{i} \in S$ such that $B_1 > B_{2} > \cdots > B_{i}$. We define $|>|$ as $\max_{B \in \boxset{G}}\stratu{>}{B}$. If $t$ is $>$-stratified, the $>$-stratum of every box is in $\mathbb{N}$ because $\left|\boxset{G}\right|$ is finite. Thus, one can bound $\left|\cop{B}\right|$ by induction on $\stratu{>}{B}$ (thanks to Equation~\ref{eq_f_dup}). Because $W_G \leq |\edges{G}| \cdot \left( \max_{(B,P) \in \pot{\boxset{G}}}|\cop{B,P}|\right)^{\depthG{G}}$, this gives us a bound on $W_G$.

In most previous works, the stratum $s(~)$ is rather explicit while $>$ is left implicit (it can be defined by ``$B > C$ iff $s(B) > s(C)$''). Concretely, in~\cite{girard1995light} and~\cite{danos2003linear}, the stratum of a box is defined as its depth (the number of boxes containing it). To enforce Equation~\ref{eq_f_dup}, digging and dereliction ($\wn N$ and $\wn D$ nodes) are forbidden. In~\cite{baillot2010linear}, Baillot and Mazza label the edges with a natural number. To enforce Equation~\ref{eq_f_dup}, Baillot and Mazza define some local conditions that those labels have to satisfy. Those works are presented as subsystems of Linear Logic: $ELL$~\cite{girard1995light} and $L^3$~\cite{baillot2010linear}. In both cases, the function $f$ in Equation~\ref{eq_f_dup} is an elementary function (tower of exponential of fixed height). Because this class of functions is stable by composition and maximum, $ELL$ and $L^3$ proof-nets normalize in a number of steps bounded by an elementary function of its size, and this function only depends on $\max_{B \in B_G}s(B) \leq |B_G|$.

When they defined $L^3$~\cite{baillot2010linear}, Baillot and Mazza did more than improving the intensional expressivity of $ELL$, they showed that ``exponential boxes and stratification levels are two different things''. This clarified the notion of stratification and enabled the present work. Here, we go further in that direction: we disentangle three principles (stratification, dependence control and nesting) which are implicit in $LLL$ and $L^4$. These principles are presented as the acyclicity of relations on boxes (respectively $\stratSNLL$, $\dcSim$ and $\nestSim$) whose intuitive meanings are described below. The meanings are voluntarily vague because those principles are not limited to the representations given in this paper, there are many variations possible~\cite{perrinelMegathese}. The intuitive meanings are not given in terms of linear logic but in the larger setting of models of computation based on rewriting. Indeed, we applied those principles both to linear logic and $\lambda$-calculus. We believe them to be relevant in other frameworks based on rewriting such as interaction nets, recursion theory and term rewriting systems. In this larger setting, the relations are between {\em parts} of a programs (boxes for linear logic, subterms for $\lambda$-calculus).
\begin{itemize}
\item Stratification (Section~\ref{section_def_stratSNLL_simple}): $B \stratSNLL C$ means that $B$ will interact with a part $C'$ (i.e. during reduction there is a rewriting step involving $B$ and $C'$) which will be created by a rewriting rule involving $C$. For instance, let us consider $t=\lambda x.(\lambda y.(y) \lambda w.w) \lambda z.(z)x$, we have $\lambda w.w \stratSNLL \lambda z.(z)x$ because $t \betared \lambda x.(\lambda z.(z)x) \lambda w.w \betared \lambda x.(\lambda w.w)x \betared \lambda x.x$ so the last step is a rewriting step involving both $\lambda w.w$ and $B'=x$, which is created during a step involving $\lambda z.(z)x$ (the second step).  
\item Dependence control (Section~\ref{section_dependence_control_simple}): $B \dcSim C$ means that several parts of $C$ will be substituted by $B$. Those parts will not be duplicated inside $C$. For instance, let us consider $t=\lambda y.(\lambda x.(x)(x)(\lambda w.w)y)\lambda z.z$, we have $\lambda z.z \dcSim (\lambda x.(x)(x)(\lambda w.w)y)$ because the two occurrences of $x$ in $\lambda x.(x)(x)(\lambda w.w)y$ will indeed be replaced by $\lambda z.z$. None of those occurrences of $x$ will be duplicated during a normalization of $\lambda x.(x)(x)(\lambda w.w)y$. 
\item Nesting (Section~\ref{subsection_depcontrol_digging}): $B \nestSim C$ means that a part of $C$ will be substituted by $B$. Those free variables may be duplicated inside $C$. For instance, let us consider the $\lambda$-term $t=(\lambda y.(\lambda x.(y)x)\lambda z.z)\lambda w.(w)w$ we have $\lambda z.z \nestSim \lambda x.(y)x$ we can notice that the occurrence of $x$ in $\lambda x.(y)x$ will indeed be replaced by $\lambda z.z$. This occurrence of $x$ may be duplicated, with the reduction $t \betared (\lambda x.(\lambda w.(w)w)x)\lambda z.z \betared (\lambda x.(x)x)\lambda z.z$. 
\end{itemize}

The acyclicity of $\stratSNLL$ entails an elementary bound on $W_G$ (Theorem~\ref{theoStratElementaryBound}), the acyclicity of the three relations entails a polynomial bound on $W_G$ (Corollary~\ref{coro_bound_poly_nest}). We want to find characterizations of complexity classes which are as intensionally expressive as possible. So we try to find the smallest possible relation $>$ (with respect to inclusion) whose acyclicity entails a bound of the shape of Equation~\ref{eq_f_dup}. Indeed if, for every proof-net the relation $R_1$ on boxes is a subset of $R_2$, then the acyclicity of $R_2$ implies the acyclicity of $R_1$. So more proof-nets are $R_1$-stratified than $R_2$-stratified.

We want to prove a bound on the number of copies of boxes. Let us consider a potential box $(B,P)$ and a copy $t$ of $(B,P)$, by definition of copies there exists a path $((\sigma(B),P),[\oc_t]) \csRel^* ((e,Q),[\oc_{\sige}])$. Our idea to prove Corollary~\ref{coro_bound_poly_nest} is to determine entirely $t$ from a partial information on $(e,Q)$ and on the $\onlyJump$ steps of the path. Because there is a bounded number of possibilities for those information, we have a bound on the number of copies of $(B,P)$.

\begin{itemize}
\item {\em Stratification:} When $\stratSNLL$ is acyclic, one can {\em trace back} $\noJump$-paths: let us suppose that $C_k \noJump^* C_1 \noJump C_0$, with some partial information on $C_0$ we can deduce a partial information on $C_1$, $C_2$,... $C_k$. In particular, we can deduce the edges of all those contexts.
\item {\em Dependence control:} When $\dcSim$ is acyclic, one can {\em trace back} the $\onlyJump$ steps. Thus, if $\stratSNLL$ and $\dcSim$ are acyclic and $C_k \csRel^* C_1 \csRel C_0$, we only need a bounded amount of information to deduce the edges of the contexts. This gives us a bound on the number of sequences $e_k,\cdots,e_1,e_0$ of edges such that there exists a path of the shape $((\sigma(B),P),[\oc_{t}]) \csRel ((e_k,\_),\_) \csRel \cdots ((e_1,\_),\_) \csRel ((e_0,\_),[\oc_{\sige}])$.
\item {\em Nesting:} If there is no $\digLab$ node, then a copy $t$ of $(B,P)$ is a list of $\sigl$ and $\sigr$ which is entirely determined by the sequence $e_k,\cdots,e_0$ of edges of the path $((\sigma(B),P),[\oc_{t}]) \csRel ((e_k,\_),\_) \csRel \cdots ((e_1,\_),\_) \csRel ((e_0,\_),[\oc_{\sige}])$. Combined with the acyclicity of $\stratSNLL$ and $\dcSim$, this gives us a bound on $|\cop{B,P}|$.
\end{itemize}

\label{section_stratification_simple}
\subsection{Definition of $\stratSNLL$-stratification}
\label{section_def_stratSNLL_simple}
To prove the complexity bounds for $ELL$ and $LLL$, one usually uses a round-by-round cut-elimination procedure. During round $i$, we reduce every cut at depth $i$. We can bound the number of $\wn C$ node residues at depth $i+1$ and, because the boxes at depth $i$ can only be duplicated by $\wn C$ nodes at depth $i+1$, it gives us a bound on the number of times boxes at depth $i+1$ are duplicated. We will proceed similarly: we will prove a bound on the number of nodes (in particular the $\wn C$ and $\wn N$ nodes) obtained after $i$ rounds of cut-elimination, and prove that it gives us a bound on the number of duplication during round $i+1$ by tracing back paths corresponding to copies from the $((e,P),[\oc_{\sige}])$ context\footnote{One can observe that we can restrict $e$ to be either the principal door of $B$, or a (reverse) premise of a $\wn C$ or $\wn N$ node, because crossing those nodes upwards are the only step modifying the signature of the left-most trace element.} back to $((\sigma(B),P),[\oc_t])$ and showing that the potential edge $(e,P)$ (corresponding to a residue) determines $t$ in a unique way.

To understand the definition of $\stratSNLL$, let us first define a relation $\stratELL$ on boxes by: $B \stratELL C$ iff there exists a path of the shape $((\sigma(B),\_),[\oc_{\_}]) \csRel^* ((e,\_),\_)$ with $e \in C$.

Let us notice that, if $((\sigma(B),P),[\oc_t]) \csRel^* ((e_k,P_k),[\oc_{t_k}]@T_k) \noJump^k ((e_0,P_0),[\oc_{t_0}])$, one only needs to know $e_0$ and $P_0$ to deduce $(e_i,P_i,T_i)_{1 \leq i \leq k}$ (because $\noJump$ is injective). By definition of $\stratELL$, for every box $C$ containing $e_i$, we have $B \stratELL C$. Thus, there are only $\left|\edges{G}\right| \cdot \left( \max_{B \stratELL C}|\cop{C,Q}|\right)^{\depthG{G}}$ such paths (it is enough to fix $e_0 \in \edges{G}$ and a copy for every box containing $e_0$). 

The idea of this section is to identify {\em unnecessary} $B \stratELL C$ pairs. It is to say, boxes $B$ and $C$ such that $B \stratELL C$ but tracing back $\noJump$-paths originating from $((\sigma(B),P),[\oc_t])$ does not depend on an element of $\cop{C,Q}$. The first such example is whenever $B \subset C$ and no $\csRel$ path from $((\sigma(B),P),[\oc_t])$ to $((e,R),[\oc_u])$ leaves the box $C$. In this case, the signature corresponding to $C$ never changes along the path. So, whenever $((\sigma(B),P),[\oc_t]) \csRel^* ((e_k,P_k),[\oc_{t_k}]) \noJump^k ((e_0,P_0),[\oc_{t_0}])$ the signature corresponding to $C$ is the same in $P$, $P_k$ and $P_0$. This signature never goes to the trace, so knowing it is not necessary to trace back the path. In this case, knowing $max_Q |\cop{C,Q}|$ is not necessary to bound the number of $\noJump$-paths originating from $((\sigma(B),P),[\oc_t])$.

Thus, $B \stratELL C$ couples are necessary only if there is a $\csRel$ path from $((\sigma(B),P),[\oc_t])$ which enters $C$ by one of its doors (either auxiliary or principal). In fact, we prove that the $B \stratELL C$ couples are necessary only if there is a $\csRel$ path from $((\sigma(B),P),[\oc_t])$ which enters $C$ by its {\em principal} door. To understand why, we study an example. In Figure~\ref{fig_ex_principal_nec}, if $((\sigma(B),P),[\oc_t]) \noJump^* ((\overline{d},q),[\oc_{\sige}])$, we need to know $q$ to trace back the path (i.e. to deduce the list of edges of those paths) because:
\begin{equation}\label{eq_paths_bc}
  \left\{
    \begin{array}{l}
      ((\sigma(B),[\sigr(\sige)]),[\oc_{\sigr(\sige)}]) \noJump^4 ((g,[]),[\oc_{\sige};\wn_{\sigl(\sige)}]) \noJump^3 ((\overline{d},[\sigl(\sigr(\sige))]),[\oc_{\sige}])\\
      ((\sigma(B),[\sigr(\sige)]),[\oc_{\sigl(\sige)}]) \noJump^4 ((h,[]),[\oc_{\sige};\wn_{\sigr(\sige)}]) \noJump^3 ((\overline{d},[\sigr(\sigr(\sige))]),[\oc_{\sige}])
    \end{array}
  \right .
\end{equation}
So $B \stratELL C$ is a necessary pair. Tracing those paths backwards, the difference in the potential corresponding to $C$ becomes a difference in a $\wn_{\_}$ trace element (in the $((\overline{\sigma(C)},[]),[\oc_{\sige};\wn_q]) \csRel ((\overline{d},[q]),[\oc_{\sige}])$ step). And because of this difference on a $\wn_{\_}$ trace element, the reverse paths separate when the paths cross a $\contLab$ node downwards: $((g,[]),[\oc_{\sige};\wn_{\sige}]) \csRel ((f,[]),[\oc_{\sige};\wn_{\sigl(\sige)}])$ and $((h,[]),[\oc_{\sige};\wn_{\sige}]) \csRel ((f,[]),[\oc_{\sige};\wn_{\sigr(\sige)}])$.

On the contrary, if $((\sigma(D),P),[\oc_t]) \noJump^* ((\overline{w},[q_A;q_B]),[\oc_{\sige}])$, we only need to know $q_B$ to trace back the path. Indeed the paths do not enter $A$ by its principal door, so $q_A$ can only appear on $\oc_{\_}$ trace elements, never on $\wn_{\_}$ trace elements.

\begin{figure}\centering
  \begin{tikzpicture}
    \node[auxdoor]   (baux) at (0,0) {};
    \draw[ar] ($(baux)+(0,0.6)$) --(baux);
    \draw[ar] (baux)--++(0,-0.6);
    \node[princdoor] (bpri) at ($(baux)+(0.8,0)$) {};
    \node [weak] (bweak) at ($(bpri)+(0,0.7)$) {};
    \draw [ar]   (bweak) -- (bpri) node [edgename] {$w$};
    \draw (bpri)-| ++(0.7,1.1) -| ($(baux)+(-0.4,0)$) -- (baux) -- (bpri);
    \node[cont] (bcont) at ($(bpri)+(1.3,0)$) {};
    \node[cut]  (bcut)  at ($(bpri)!0.5!(bcont)+(0,-0.4)$) {};
    \draw[ar,out=-90,in=180] (bpri) to (bcut);
    \draw[ar,out=-90,in=  0] (bcont)to (bcut);
    \node[ax] (bax1) at ($(bcont)+(0.8,0.3)$) {};
    \node[ax] (bax2) at ($(bax1)+(0,0.25)$) {};
    \node [auxdoor]  (aaux1) at ($(bax1)+(0.2,-1.5)$) {};
    \node [auxdoor]  (aaux2) at ($(aaux1)+(0.9,0)$) {};
    \node [princdoor](apri)  at ($(aaux1)+(-3.5,0)$) {};
    \draw[ar] ($(apri)+(0,0.6)$)--(apri);
    \draw [ar,out=180,in= 60] (bax1) to (bcont);
    \draw [ar,out=180,in=120] (bax2) to (bcont);
    \draw [ar,out=  0,in= 90] (bax1) to (aaux1);
    \draw [ar,out=  0,in= 90] (bax2) to (aaux2);
    \draw (apri) -- (aaux1) -- (aaux2) -| ++ (0.4,2.5) -| ($(apri)+(-0.7,0)$) -- (apri);
    \node [cont] (ccont1) at ($(aaux1)!0.5!(aaux2)+(0,-0.75)$) {};
    \draw [ar] (aaux1) -- (ccont1) node [edgename] {$g$};
    \draw [ar] (aaux2) -- (ccont1) node [edgename,right] {$h$};
    \node [princdoor] (cpri) at ($(ccont1)+(2.4,0)$) {};
    \node [der]       (cder) at ($(cpri)+(0,1)$) {};
    \node [princdoor] (dpri) at ($(cder)+(0,1)$) {};
    \draw (dpri) -| ++(0.7,0.9) -| ($(dpri)+(-0.4,0)$) -- (dpri);
    \draw (cpri) -| ++(0.8,3)   -| ($(cpri)+(-0.5,0)$) -- (cpri);
    \draw [ar] (dpri) -- (cder);
    \draw [ar] (cder) -- (cpri) node [edgename] {$d$};
    \node [cut] (ccut) at ($(ccont1)!0.6!(cpri)+(0,-0.5)$) {};
    \draw [ar,out=-90,in=180] (ccont1) to node [edgename,above right=-0.05] {$f$} (ccut);
    \draw [ar,out=-90,in=  0] (cpri)  to (ccut);
    \node [above] at ($(apri)+(-0.5,0)$) {${\mathbf A}$};
    \node [above] at ($(bpri)+( 0.5,0)$) {${\mathbf B}$};
    \node [above] at ($(cpri)+( 0.5,0)$) {${\mathbf C}$};
    \node [above] at ($(dpri)+( 0.5,0)$) {${\mathbf D}$};

    \coordinate      (extdig)   at ($(apri)+(-1.8,0)$);
    \node [cont]     (extcont1) at ($(extdig)+(0,0)$) {};
    \node [weak]     (extweak1) at ($(extcont1)+( 60:0.8)$) {};
    \node [der]      (extder1)  at ($(extcont1)+(120:0.8)$) {};
    \node [cut]      (extcut2) at ($(apri)!0.5!(extdig)+(0,-0.6)$) {};
    \node [ax]       (extax)   at ($(extder1)+(-0.5,0.6)$) {};
    \draw [ar] (extder1)  -- (extcont1);
    \draw [ar] (extweak1) -- (extcont1);
    \draw [ar] (extax) to [out=  0, in=90] (extder1);               
    \draw [ar] (extax) to [out=180, in=90] ($(extax)+(-0.5,-0.8)$); 
    \draw [ar] (extcont1)to [out=-90,in=170] (extcut2);               
    \draw [ar] (apri)to [out=-90,in= 10] (extcut2);
  \end{tikzpicture}
  \caption{\label{fig_ex_principal_nec}$D \stratELL A$ but it is an ``unnecessary'' couple because $|\cop{B,P}$ does not depend on $|\cop{A,[]}|$.}
\end{figure}

We define a relation $\stratSNLL$ between boxes of proof-nets. $B \stratSNLL C$ means that there is a path beginning by the principal door of $B$ which enters $C$ by its principal door.

\begin{definition}\label{def_stratPot} Let $B, C \in \boxset{G}$, we write $B \stratSNLL C$ if there is a path of the shape:
\begin{equation*}\label{def_twoheadrightsquigarrow}
  ((\sigma(B), P), [\oc_t]) \noJump^* ((\overline{\sigma(C)},Q),T)
\end{equation*} 
\end{definition}

We can notice that for every proof-net, $\stratSNLL \subseteq \stratELL$. For example, in the proof-net of Figure~\ref{fig_ex_principal_nec}, we have $B \stratELL A$, $B \stratELL C$, $D \stratELL A$, $D \stratELL B$ and $D \stratELL C$ while the only pair for $\stratSNLL$ are $B \stratSNLL C$ and $D \stratSNLL B$. 

As shown in Equation~\ref{eq_paths_bc}, to trace back the $\noJump$-path from $((\sigma(B),[\sigr(\sige)]),[\oc_{\sigr(\sige)}])$ to $((\overline{d},[q]),[\oc_{\sige}])$ one needs information on $q=\sigl(\sigr(\sige))$. However, let us notice that it is not necessary to know $q$ entirely. The only information needed to trace back the path is that it is of the form $\sigl(x)$. Knowing that $x=\sigr(\sige)$ is useless because the information in $x$ would only be used if the path entered $A$ by its principal door and that is not the case. 

The following intuitions (formalized in Section~\ref{section_restricted_copies}) capture the notion of the information needed to trace back the paths. As we stated earlier, a canonical potential of a box corresponds to a residue of this box along reduction, a {\em $\csRel_S$-canonical potential} of a box corresponds to a residue obtained without firing cuts involving the principal door of a box outside $S$. It is to say, a $\csRel_S$-canonical potential of a box corresponds to a residue of this box along reduction such that, for every step of this reduction involving the principal door of a box $B$, $B$ is a residue of a box of $S$.

More formally, we first define the $\csRel_S$-copies of $(B,P)$ as the copies $t$ of $(B,P)$ such that in the paths $((\sigma(B),P),[\oc_{t}]) \csRel^* ((\_,\_),[\oc_{\sige}])$, every $\onlyJump$ step of the path is on a box of $S$. For instance, in the proof-net of Figure~\ref{fig_ex_principal_nec}, the $\csRel_{\{C\}}$-copies of $(C,[])$ are $\{ \sige, \sigl(\sige), \sigr(\sige) \}$ while the $\csRel_{\{C,A\}}$-copies of $(C,[])$ are $\{ \sige, \sigl(\sige), \sigr(\sige), \sigl(\sigl(\sige)),\sigl(\sigr(\sige)),\sigr(\sigl(\sige)),\sigr(\sigr(\sige)) \}$. Then, we define $\csRel_S$-canonical potentials from the notion of $\csRel_S$-copies in the same way as we defined canonical potentials from the notion of copies.

Let us suppose that we know that $\_ \noJump^7 ((\overline{d},[q]),[\oc_{\sige}])$ and $\_ \noJump^7 ((\overline{d},[q']),[\oc_{\sige}])$ and we want to prove that those paths take the same edges. We only need to know that the $\csRel_{\{C\}}$-copies of $(C,[])$ ``corresponding'' to $q$ and $q'$ are equal. We define $x$ (resp. $x'$) as the ``biggest'' $\csRel_{\{C\}}$-copy of $(C,[])$ which is a ``truncation'' of $q$ (resp. $q'$). For instance, if $q=\sigr(\sigl(\sige))$ and $q'=\sigr(\sigr(\sige))$, then $q \neq q'$ but we have $x=x'=\sigr(\sige)$. This is enough to know that $q$ and $q'$ are of the shape $\sigr(\_)$ and this information is enough to trace back the paths, so to prove that the paths take the same edges.

The $\csRel_S$-copy of $(B,P)$ corresponding to $t$ is written $\restrSig{\csRel_S}{((\sigma(B),P),[\oc_t])}$. It represents the part of $t$ which is used if we refuse the $\onlyJump$ steps over the potential boxes which are not in $S$. For instance, in Figure~\ref{fig_ex_principal_nec}, $\restrSig{\csRel_{\{(C\}}}{((\sigma(C),[]),[\oc_{\sigr(\sigl(\sige))}])}= \sigr(e)$ because, if we refuse to jump over $(A,[])$, only $\sigr(\_)$ is consumed in the $\csRel$ paths starting from this context. Then, $\restrPot{\csRel_S}{e,P}$ is defined from the $\restrSig{\csRel_S}{((\sigma(B),P),[\oc_t])}$ construction in the same way as canonical potentials are defined from copies.

\subsection{Restricted copies and canonical potentials}\label{section_restricted_copies}
Now that we gave the intuitions, we can state the formal definitions.

\begin{definition}\label{def_mapstoset} Let $G$ be a proof-net and $S \subset \boxset{G}$. We define $\csRel_{S}$ and $\noJump_S$ as follows:
  \begin{equation*}
    C \csRel_{S}D \Leftrightarrow \left \{ \begin{array}{l} C \csRel D \\ \text{If }C= ((\sigma(B),P),[\oc_t]) \text{, then }  B \in S \end{array} \right.
\end{equation*}
\begin{equation*}
    C \noJump_{S} D \Leftrightarrow \left \{ \begin{array}{l} C \noJump D \\ \text{If } D= ((\overline{\sigma(B)},P),T.\wn_t) \text{, then }  B \in S \end{array} \right.
  \end{equation*}
\end{definition}

If $\csRel$ corresponds to cut-elimination, $\csRel_S$ corresponds to cut-elimination restricted by allowing reduction of cuts involving the principal door of a box $B$ only if $B$ is a residue of a box of $S$. In the following, we suppose given a relation $\rightarrow$ on contexts such that $\rightarrow \subseteq \csRel$.

\begin{definition}\label{def_copyrel}\label{def_arrowcopy}
  A $\rightarrow$-copy context is a context of the shape $((e,P),[\oc_t]@T)$ such that for every $u \compl t$, there exists a path of the shape $((e,P),[\oc_u]@T) \rightarrow^* ((\_,\_),[\oc_{\sige}])$.
  
  Let $(B,P) \in \pot{B_G}$, the set $\copRel{\rightarrow}{B,P}$ of $\rightarrow$-copies of $(B,P)$ is the set of standard signatures $t$ such that $((\sigma(B),P),[\oc_t])$ is a $\rightarrow$-copy context.
\end{definition}

For example, for any box $B$ and set $S$ such that $B \not \in S$, $\copRel{\csRel_S}{B,P}=\{\sige\}$ (because $((\sigma(B),P),[\oc_t]) \not \csRel_S$). In Figure~\ref{fig_ex_principal_nec}, we have $\copRel{\csRel_{\{C\}}}{C,[]}=\{\sige,\sigl(\sige),\sigr(\sige)\}$ whereas $\copRel{\csRel_{\{A,C\}}}{C,[]}=\{\sige,\sigl(\sige),\sigr(\sige),\sigl(\sigl(\sige)),\sigl(\sigr(\sige)),\sigr(\sigl(\sige)),\sigr(\sigr(\sige))\}$.

\begin{definition}\label{def_canonicalrel}Let $e$ be an edge of $G$ such that $e \in B_{\depth{e}} \subset ... \subset B_1 $. We define $\canRel{\rightarrow}{e}$ as the set of potentials $[s_{1} ;  ... ; s_{\depth{e}}]$ such that:
  \begin{equation*}\forall 1\leq i\leq \depth{x}, s_i \in \copRel{\rightarrow}{B_i, [s_{1}; \cdots ; s_{i-1}]}
  \end{equation*}
\end{definition}

For instance, in Figure~\ref{fig_ex_principal_nec}, $\canRel{\csRel_{\{B\}}}{w}=\{(w,[\sige;\sige]),(w,[\sige;\sigl(\sige)]),(w,[\sige;\sigr(\sige)])\}$ and $\canRel{\csRel_{\{C\}}}{d}=\{(d,[\sige]),(d,[\sigl(\sige)]),(d,[\sigr(\sige)])\}$.

We can notice that, in particular, the definitions of $\copRel{\csRel}{B,P}$ and $\canRel{\csRel}{x}$  match respectively the definitions of $\cop{B,P}$ (Definition~\ref{def_copy}) and $\can{x}$ (Definition~\ref{def_canonical}). Finally, we define in Definition~\ref{def_canonicalrel_contexts} a notion of $\rightarrow$-canonical contexts. Intuitively\footnote{This property is not true for every $\rightarrow$ relation, but is true if $\rightarrow$ is of the shape $\csRel_S$.}, every context reachable from $((\sigma(B),P),[\oc_t])$ by a $\rightarrow$-path with $(B,P) \in \canRel{\rightarrow}{B}$ and $t \in \copRel{\rightarrow}{B,P}$, is a $\rightarrow$-canonical context.

\begin{definition}\label{def_canonicalrel_contexts}A {\em $\rightarrow$-canonical context} is a context $((e,[P_1;\cdots;P_{\depth{e}}),[T_1;\cdots;T_k])$ such that $(e,P) \in \canRel{\rightarrow}{e}$ and:
  \begin{itemize}
  \item For every $T_i=\oc_t$, $((e,[P_1;\cdots;P_{\depth{e}}]),[\oc_t;T_{i+1};\cdots;T_k])$ is a $\rightarrow$-copy context.
  \item For every $T_i=\wn_t$, $((\overline{e},[P_1;\cdots;P_{\depth{e}}]),[\oc_t;T^{\perp}_{i+1};\cdots;T^{\perp}_k])$ is a $\rightarrow$-copy context.
  \end{itemize}
\end{definition}

Let us consider a potential box $(B,P)$ and $t \in \cop{B,P}$, then there exists a context $((e,Q),[\oc_{\sige}])$ such that $((\sigma(B),P),[\oc_t]) \csRel^* ((e,Q),[\oc_{\sige}])$. If some of those $\csRel$ steps are not in $\rightarrow$, we may have $((\sigma(B),P),[\oc_t]) \rightarrow^* ((f,R),[\oc_v]) \not \rightarrow$ with $v \neq \sige$. In this case, $t$ would not be a $\rightarrow$-copy of $(B,P)$. However, there exist ``truncations'' of $t$ which are $\rightarrow$-copy of $(B,P)$ (at least, $\sige$ verify those properties).

\begin{definition}{\label{def_truncation}}
  We define ``$t$ is a truncation of $t'$'' (written $t \prune t'$) by induction on $t$. For every signature $t,t',u,u'$, we set $\sige \prune t$ and if we suppose $t \prune t'$ and $u \prune u'$ then $\sigl(t) \prune \sigl(t')$, $\sigr(t) \prune \sigr(t')$, $\sigp(t) \prune \sigp(t')$ and $\sign(t,u) \prune \sign(t',u')$.
\end{definition}

As hinted earlier, we want to define $\restrSig{\rightarrow}{((\sigma(B),P),[\oc_t])}$ as the ``biggest'' $\rightarrow$-copy $u$ of $(B,P)$ such that $u \prune t$. But we have not precised the meaning of ``biggest'' yet. The solution we chose is to first maximize the rightmost branch. Then, once this branch is fixed, we maximize the second rightmost branch and so on. Formally, we define ``biggest'' as ``the maximum for the order $\pruneceq$'' with $\pruneceq$ defined as follows.

\begin{definition}{\label{def_prunec}}
  We first define a strict order $\prunec$ on signatures by induction. For every signature $t,t',u,v$, we set $\sige \prunec t$. And, if we suppose $t \prunec t'$, then $\sigl(t) \prunec \sigl(t')$, $\sigr(t) \prunec \sigr(t')$, $\sigp(t) \prunec \sigp(t')$, $\sign(u,t) \prunec \sign(v,t')$ and $\sign(t,u) \prunec \sign(t',u)$.

  Then we define an order $\pruneceq$ on signatures by: $t \pruneceq t'$ iff either $t=t'$ or $t \prunec t'$.
\end{definition}

\begin{lemma}[\cite{perrinelMegathese}]\label{lemma_pruneceq_total}
  Let $t$ be a signature, then $\pruneceq$ is a total order on $\Set{ u \in \sig}{ u \prune t}$.
\end{lemma}

Thanks to Lemma~\ref{lemma_pruneceq_total}, the set $\restr{\rightarrow}{(\sigma(B),P),[\oc_t]}$ defined below is totally ordered by $\pruneceq$ and finite (if $t$ is of size $k$, it has at most $2^k$ truncations) so it admits a maximum for $\pruneceq$, written $\restrSig{\rightarrow}{((\sigma(B),P),[\oc_t])}$.

\begin{definition}{\label{def_cop_restr}}
  Let $((e,P),[\oc_t]@T) \in \context{G}$, we define $\restr{\rightarrow}{(e,P),[\oc_t]@T}$ as the set of signatures $u$ such that $u \prune t$ and $((e,P),[\oc_u]@T)$ is a $\rightarrow$-copy context. Then, we define $\restrSig{\rightarrow}{((e,P),[\oc_t]@T)}$ as the maximum (for $\pruneceq$) element of $\restr{\rightarrow}{(e,P),[\oc_t]@T}$.
\end{definition}

For example, in the proof-net of Figure~\ref{fig_ex_principal_nec}, $\restr{\csRel_{\{C\}}}{(\sigma(C),[]),[\oc_{\sigl(\sigr(\sige))}]}=\{\sige,\sigl(\sige)\}$ so we have $\restrSig{\csRel_{\{C\}}}{((\sigma(C),[]),[\oc_{\sigl(\sigr(\sige))}])}=\sigl(\sige)$.

In Figure~\ref{fig_ex_principal_nec}, $((\overline{\sigma_1(A)},[]),[\oc_u]) \csRel_S ((\sigma(A),[]),[\oc_u])$ for any $u \in \sig$. So, for any $t \in \sig$, $\restrSig{\csRel_S}{((\overline{\sigma_1(A)},[]),[\oc_t])}=\restrSig{\csRel_S}{((\sigma(A),[]),[\oc_{\sige}])}$. Lemma~\ref{lemma_restr_eq} generalizes this observation.

\begin{lemma}[\cite{perrinelMegathese}]\label{lemma_restr_eq} Let $t \in \sig$. We suppose that, for every $u \prune t$ and $v \compl u$, we have $((e,P),[\oc_v]@T) \rightarrow ((f,Q),[\oc_v]@U)$. Then, $\restrSig{\rightarrow}{((e,P),[\oc_t]@T)}=\restrSig{\rightarrow}{((f,Q),[\oc_t]@U)}$.
\end{lemma}

Now, for any potential edge $(e,P)$, we want to define $\restrPot{\rightarrow}{e,P}$ as the ``biggest'' truncation $P'$ of $P$ such that $(e,P')$ is a $\rightarrow$-canonical edge. We first maximize the leftmost signature, then the second, and so on.

\begin{definition}\label{def_potrestr} For every potential edge $(e,P)$, we define $\restrPot{\rightarrow}{e,P}$ by induction on $\depth{e}$. If $\depth{e}=0$, then we set $\restrPot{\rightarrow}{e,[]}=(e,[])$. Otherwise we have $P=Q.t$, let $B$ be the deepest box containing $e$, $(\sigma(B),Q')=\restrPot{\rightarrow}{\sigma(B),Q}$ and $t'=\restrSig{\rightarrow}{((\sigma(B),Q'),[\oc_t])}$ then we set $\restrPot{\rightarrow}{e,Q.t}=(e,Q'.t')$.
\end{definition}

For example, in the proof-net of Figure~\ref{fig_ex_principal_nec}, $\restrPot{\csRel_{\{B\}}}{w,[\sigr(\sige);\sigl(\sige)]}=(w,[\sige;\sigl(\sige)])$.

\begin{definition}{\label{def_prune_pot}} We extend $\prune$ on $\Pot$ by $[p_1;\hspace{-0.04em}\cdots\hspace{-0.04em};p_k]\hspace{-0.17em}\prune\hspace{-0.17em}[p'_1;\hspace{-0.04em}\cdots\hspace{-0.04em};p'_k]$ iff for $1 \hspace{-0.2em}\leq\hspace{-0.1em} i \hspace{-0.2em}\leq\hspace{-0.1em} k$, $p_i \hspace{-0.17em}\prune\hspace{-0.17em} p'_i$.
\end{definition}

We can notice that, in the same way as the definition of $\can{e}$ only depends on the boxes containing $e$ (cf. page~\pageref{remarque_can_same}), the definition of $\restrPot{\rightarrow}{e,P}$ only depends on the boxes containing $e$. We formalize it with the next lemma.

\begin{lemma}\label{lemma_restrpot_eq}
  If $e,\hspace{-0.05em}f \hspace{-0.12em}\in\hspace{-0.08em} \dirEdges{G}$ belong to the same boxes, \hspace{-0.05em}$\restrPot{\rightarrow}{e,P} \hspace{-0.2em}=\hspace{-0.1em} (e,P')$ iff $\restrPot{\rightarrow}{f,P}\hspace{-0.2em}=\hspace{-0.1em}(f,P')$.
\end{lemma}

Let us suppose that $((\sigma(B),P),[\oc_t]) \noJump^* ((e,Q),[\oc_{\sige}])$ and $S$ is the set of boxes which are entered by their principal door by this path. Then, we prove that it is enough to know $\restrPot{\csRel_S}{e,Q}$ to trace back the path (Lemma~\ref{lemma_remonter_simple}). To do so, we need to prove that for every intermediary step $(((e_k,P_k),T_k) \noJump ((e_{k+1},P_{k+1}),T_{k+1})$ we have enough information about $P_{k+1}$ and $T_{k+1}$ to determine $e_k$. This is the role of the following definition. As an intuition, if $((e,P'),T')=\restrCont{\rightarrow}{((e,P),T)}$ then $((e,P'),T')$ is the ``biggest'' $\rightarrow$-canonical context which is a truncation of $((e,P),T)$.

\begin{definition}{\label{def_cont_restr}} For $((e,P),[T_n;\cdots;T_1]) \in \context{G}$ we define $\restrCont{\rightarrow}{((e,P),[T_n;\cdots;T_1])}$ as $((e,P'),[T'_n;\cdots;T'_1])$ with $(e,P')=\restrPot{\rightarrow}{e,P}$ and $T'_i$ defined by induction on $i$ as follows:
  \begin{itemize}
  \item If $T_i=\oc_t$, then $T'_i=\oc_{t'}$ with $t'=\restrSig{\rightarrow}{((e,P'),[\oc_t;T'_{i-1}\hspace{0.4em};\cdots;T'_1\hspace{0.3em}])}$.
  \item If $T_i=\wn_t$, then $T'_i=\wn_{t'}$ with $t'=\restrSig{\rightarrow}{((\overline{e},P'),[\oc_{t};{T'_{i-1}}^{\hspace{-0.6em}\perp};\cdots;{T'_1}^{\hspace{-0.05em}\perp}])}$.
  \item Otherwise, $T'_i=T_i$.
  \end{itemize}
\end{definition}

%

Lemma~\ref{lemma_restrCont_left} is a generalization of Lemma~\ref{lemma_restr_eq} to contexts. For example, in Figure~\ref{fig_ex_principal_nec}, for every $S \subseteq \boxset{G}$ and trace $T$ we have, $((d,[\sigr(\sign(\sige,\sige))]),T) \noJump_S^{3} ((\overline{h},[]),T.\oc_{\sign(\sige,\sige)})$ and $((h,[]),T^{\perp}.\wn_{\sign(\sige,\sige)}) \csRel_S^3 ((\overline{d},[\sigr(\sign(\sige,\sige))]),T^\perp)$. So for every $t,u \in \sig$, there exist $v,w \in \sig$ such that 
\begin{align*}
  \restrCont{\csRel_S}{((\overline{h},[]),[\oc_{t};\wn_{u};\oc_{\sign(\sige,\sige)}])}&=((\overline{h},[]),[\oc_{v};\wn_{w};\oc_{\_}])\\
  \restrCont{\csRel_S}{((d,[\sigr(\sign(\sige,\sige))]),[\oc_{t};\wn_{u}])}&=((d,[\_]),[\oc_{v};\wn_{w}])
\end{align*}

\begin{lemma}\label{lemma_restrCont_left}
  Let $(e,P),(e,Q)$ be potential edges and $U,V$ be lists of trace elements. Let us suppose that, for every trace element list $T$, $((e,P),T@U) \rightarrow ((f,Q),T@V)$ and $((\overline{f},Q),T^\perp @ V^\perp) \rightarrow ((\overline{e},P),T^\perp @ U^\perp ) $. Then, for any trace $T$, $\restrCont{\rightarrow}{((e,P),T@U)}$ and $\restrCont{\rightarrow}{((f,Q),T@V)}$ are of the shape $(\_,T'@U')$ and $(\_,T'@V')$ with $|T|=|T'|$.
\end{lemma}
\begin{proof}
  Let us write $[T_k;\cdots;T_1]$ for $T$, $(\_,[T'_k;\cdots;T'_1]@U')$ for $\restrCont{\rightarrow}{((e,P),T@U)}$ and $(\_,[T''_k;\cdots;T''_1]@V')$ for $\restrCont{\rightarrow}{((f,Q),T@V)}$. We prove $T'_i=T''_i$ by induction on $i$. 

  If $T_i=\oc_{t}$, then we have $T'_i=\oc_{t'_i}$ with $t'_i=\restrSig{\rightarrow}{((e,P),[\oc_t;T'_{i-1};\cdots;T'_1]@U')}$ and $T''_i=\oc_{t''_i}$ with $t''_i=\restrSig{\rightarrow}{((f,Q),[\oc_t;T''_{i-1};\cdots;T''_1]@V')}$. By the induction hypothesis, we have $[T''_{i-1};\cdots;T''_1]=[T'_{i-1};\cdots;T'_1]$. By assumption $((e,\hspace{-0.05em}P),[\oc_t;\hspace{-0.05em}T'_{i-1};\cdots;T'_{1}]\hspace{-0.05em}@U') \hspace{-0.1em}\rightarrow^* ((f,Q),[\oc_t;T'_{i-1};\cdots;T'_{1}]@V')$. Thus, by Lemma~\ref{lemma_restr_eq}, $t'_i=t''_i$ so $T'_i=T''_i$. 
  
  The case $T_i=\wn_t$ is similar (using the $((\overline{f},Q),T^\perp @ V^\perp) \rightarrow ((\overline{e},P),T^\perp @ U^\perp ) $ hypothesis).
\end{proof}

\subsection{Elementary bound for $\stratSNLL$-stratified proof-nets}\label{section_elementary_bound_simple}

We consider the following theorem as the main technical innovation of this paper. It uses the notions of the previous section to trace back $\noJump$-paths. In order to bound $W_G$, we need to bound the number of copies of potential boxes. The usual way to prove the elementary bound on $LLL$ is a round-by-round cut-elimination procedure: we first reduce every cut at depth $0$. Because of the absence of dereliction in $ELL$, none of these step creates new cuts at depth $0$. So this round terminates in at most $|\edges{G}|$ steps. Because each step may at most double the size of the proof-net, the size of the proof-net at the end of round $0$ is at most $2^{|\edges{G}|}$. Then we reduce the cuts at depth $1$, because of the previous bound there at most $2^{|\edges{G}|}$ such cuts, and the reduction of those cuts does not create any new cut... 

The original proof of the elementary bound of $L^3$ relies on a similar round-by-round procedure which is more complex because reducing a cut at level i can create new cuts at level $i$, and a box of level $i$ can be contained in a box of higher level. While Dal Lago adapted to context semantics the round-by-round procedure of $ELL$ concisely in~\cite{lago2006context}, the round-by-round procedure of $L^3$ was only adapted to context semantics by Perrinel~\cite{perrinel2013pathsbased} (a work which is the basis of this article).

Theorem~\label{injection_lemma_simple} allows us to bring round-by-round procedures where strata differ from depth, to context semantics. We explained that $\restrPot{\csRel_S}{e,P}$ corresponds to 
a residue $e'$ of $e$, such that we only fired cuts involving principal door of boxes of $S$. In a round-by-round procedure, after the $i$-th round we have a bound on the number of such $\restrPot{\csRel_S}{e,P}$. By tracing back a path from $((e,P),[\oc_{\sige}])$ until a potential box $(B,P)$ using only the information $\restrPot{\csRel_S}{e,P}$, we show that there is only one residue of $e'$ which will be cut with $B$ (more precisely its residue corresponding to $(B,P)$). This allows us to prove a bound on the number of copies of $(B,P)$.

While we will use other criteria and technical results to deal with the $\onlyJump$ steps, both the proofs of the elementary bound and the proofs of the polynomial bound rely on Theorem~\ref{injection_lemma_simple}.
 
  \begin{theorem}\label{injection_lemma_simple}
    Let $G$ be a proof-net and $S \subset \boxset{G}$. Let $C_e$, $C_f$ and $C'_f$ be contexts such that $C_e \noJump_S C_f$ and $\restrCont{\csRel_S}{C_f} = \restrCont{\csRel_S}{C_f'}$, then there exists a context $C'_e$ such that $C_e' \noJump_S C_f'$ and $\restrCont{\csRel_S}{C_e}= \restrCont{\csRel_S}{ C_e'}$.
  \end{theorem}
\begin{prf}
  We detail an easy step (crossing a $\parLab$ node upward). Most of the other steps are quite similar. For the steps which offer some particular difficulty, we only detail the points which differ from crossing a $\parLab$ upward.

      \tikzsetnextfilename{elem_par2}
      \begin{wrapfigure}{l}{1.5cm}
        \begin{tikzpicture}
          \node[par] (par) at (0,0) {};
          \draw[ar] ($(par)+(0,-0.5)$) -- (par)  node [edgename] {$e$};
          \draw[ar] (par) -- ($(par)+(120:0.5)$) node [edgename] {$f$};
          \draw     ($(par)+( 60:0.5)$) -- (par);
        \end{tikzpicture}
      \end{wrapfigure}
      Let us suppose that $C_e=((e,P),T.\otimes_l) \noJump_S ((f,P),T)=C_f$ (crossing a $\parr$ upwards, such that $\overline{f}$ is not a principal edge) and $\restrCont{\csRel_S}{C_f} = \restrCont{\csRel_S}{ C'_f}$. So $C'_f$ is of the shape $((f,P'),T')$. We set $C'_e=((e,P'),T'.\otimes_l)$. Let $((f,P''),T'')=\restrCont{\csRel_S}{C_f}$, then $\restrPot{\csRel_S}{f,P}=\restrPot{\csRel_S}{f,P'}=(f,P'')$. So, by Lemma~\ref{lemma_restrpot_eq}, $\restrPot{\csRel_S}{e,P}=\restrPot{\csRel_S}{e,P'}=(e,P'')$. Moreover, by Lemma~\ref{lemma_restrCont_left}, $\restrCont{\csRel_S}{C_e}=((e,P''),T''.\otimes_l)$ and $\restrCont{\csRel_S}{C'_e}=((e,P''),T''.\otimes_l)$ so $\restrCont{\csRel_S}{C_e}=\restrCont{\csRel_S}{C'_e}$.

      \tikzsetnextfilename{elem_cut2}
      \begin{wrapfigure}{l}{2cm}
        \begin{tikzpicture}
          \node[princdoor] (aux) at (0,0) {};
          \node      (etc) at ($(aux)+(1,0)$) {};
          \node[cut] (cut) at ($(aux)!0.5!(etc)+(0,-0.4)$) {};
          \draw[ar]  ($(aux)+(0,0.4)$) -- (aux);
          \draw[ar,out=-90,in=180] (aux) to node [edgename] {$e$} (cut);          
          \draw[ar,out=0,in=-90]   (cut) to node [edgename,right=-0.1] {$f$} (etc);
          \draw (aux)--++(0.5,0) (aux)--++(-0.5,0);
        \end{tikzpicture}
      \end{wrapfigure}
      Let us consider the case where $e$ is the principal edge of a box $B$ (we consider the case where we cross a $cut$), we suppose that we have $C_e=((e,P),T.\oc_t) \noJump_S ((f,P),T.\oc_t) = C_f$\footnote{If the rightmost trace element of $C_e$ is not of the shape $\oc_t$, the proof does not offer any additional difficulty compared to the step presented above.}. So $C'_f$ is of the shape $((f,P'),T'.\oc_{t'})$. We set $C'_e=((e,P'),T'.\oc_{t'})$. By supposition, $\restrCont{\csRel_S}{C_f}=\restrCont{\csRel_S}{C'_f}=((f,P''),T''.\oc_{t''})$. In particular $\restrSig{\csRel_s}{((f,P''),[\oc_{t}])}=\restrSig{\csRel_s}{((f,P''),[\oc_{t'}])}$. If $B \in S$, by Lemma~\ref{lemma_restrCont_left}, we have $\restrCont{\csRel_S}{C_e}=\restrCont{\csRel_S}{C'_e}=((e,P''),T''.\oc_{t''})$. Otherwise, we have $\restrCont{\csRel_S}{C_e}=\restrCont{\csRel_S}{C'_e}=((e,P''),T''.\oc_{\sige})$.

      \tikzsetnextfilename{elem_cutprinc_3}
      \begin{wrapfigure}{l}{2cm}
        \begin{tikzpicture}
          \node[princdoor] (aux) at (0,0) {};
          \node      (etc) at ($(aux)+(1,0)$) {};
          \node[cut] (cut) at ($(aux)!0.5!(etc)+(0,-0.4)$) {};
          \draw[ar]  ($(aux)+(0,0.4)$) -- (aux);
          \draw[ar,out=180,in=-90] (cut) to node [edgename,left=0.2cm] {$f$} (aux);          
          \draw[ar,out=-90,in=0]   (etc) to node [edgename,right] {$e$} (cut);
          \draw (aux)--++(0.5,0) (aux)--++(-0.5,0);
        \end{tikzpicture}
      \end{wrapfigure}
      Let us consider the case where $\overline{f}$ is the principal edge of a box $B$ (we consider the case where we cross a $cut$) with $C_e=((e,P),T.\wn_t) \noJump_S ((f,P),T.\wn_t) = C_f$. So $C'_f$ is of the shape $((f,P'),T'.\wn_{t'})$. We set $C'_e=((e,P'),T'.\wn_{t'})$. By supposition, $\restrCont{\csRel_S}{C_f}=\restrCont{\csRel_S}{C'_f}=((f,P''),T''.\wn_{t''})$. By definition of $\noJump_S$, $B$ is in $S$. So, we can notice that $C'_e \noJump_S C'_f$ and, using Lemma~\ref{lemma_restrCont_left}, we have $\restrCont{\csRel_S}{C_e}=\restrCont{\csRel_S}{C'_e}=((e,P''),T''.\wn_{t''})$.

      \tikzsetnextfilename{elem_princ2}
      \begin{wrapfigure}{l}{1cm}
        \begin{tikzpicture}
          \node[princdoor] (aux) at (0,0) {};
          \draw[ar] (aux) -- ($(aux)+(0,0.5)$) node [edgename] {$f$};
          \draw[ar] ($(aux)+(0,-0.5)$) --(aux) node [edgename] {$e$};
          \draw (aux)--++(0.5,0) (aux)--++(-0.5,0);
        \end{tikzpicture}
      \end{wrapfigure}
      Let us suppose that $C_e=((e,P),T.\wn_{t}) \noJump_S ((f,P.t),T)=C_f$ (crossing the principal door of $C$ upwards). Then, $C'_f$ must be of the shape $((f,P'.t'),T')$. We set $C'_e=((e,P'),T'.\wn_{t'})$. The only particular point is to prove that $\restrSig{\csRel_S}{(\restrPot{\csRel_S}{\overline{e},P},[\oc_t])}=\restrSig{\csRel_S}{(\restrPot{\csRel_S}{\overline{e},P'},[\oc_{t'}])}$. By definition, $\restrPot{\csRel_S}{f,P.t}=(f,Q.u)$ with $\restrPot{\csRel_S}{\overline{e},P}=(\overline{e},Q)$ and $\restrSig{\csRel_S}{(\overline{e},Q),[\oc_t])}=u$. Similarly, $\restrPot{\csRel_S}{f,P'.t'}=(f,Q'.u')$ with $\restrPot{\csRel_S}{\overline{e},P'}=(\overline{e},Q')$ and $\restrSig{\csRel_S}{((\overline{e},Q'),[\oc_{t'}])}=u'$. We know that $\restrCont{\csRel_S}{C_f}= \restrCont{\csRel_S}{C'_f}$, so $\restrPot{\csRel_S}{f,Q.t}=\restrPot{\csRel_S}{f,Q'.t'}$. Thus $u=u'$, i.e. $\restrSig{\csRel_S}{(\restrPot{\csRel_S}{\overline{e},P},[\oc_t])}=\restrSig{\csRel_S}{(\restrPot{\csRel_S}{\overline{e},P'},[\oc_{t'}])}$.

      \tikzsetnextfilename{elem_princ3}
      \begin{wrapfigure}{l}{1cm}
        \begin{tikzpicture}
          \node[princdoor] (aux) at (0,0) {};
          \draw[ar] ($(aux)+(0,0.5)$) -- (aux) node [edgename] {$e$};
          \draw[ar] (aux) -- ($(aux)+(0,-0.5)$) node [edgename] {$f$};
          \draw (aux)--++(0.5,0) (aux)--++(-0.5,0);
        \end{tikzpicture}
      \end{wrapfigure}

      Let us suppose that $C_e=((e,P.t),T) \noJump_S ((f,P),T.\oc_t)=C_f$ (crossing the principal door of $B$ downwards). Then $C'_f$ must be of the shape $((f,P'),T'.\oc_{t'})$. We set $C'_e=((e,P'.t'),T')$. The only particular point is to prove that $\restrPot{\csRel_S}{e,P.t}=\restrPot{\csRel_S}{e,P'.t'}$. By definition of $\restrPot{\csRel_S}{\_,\_}$, $\restrPot{\csRel_S}{e,P.t}=(e,Q.u)$ with $\restrPot{\csRel_S}{f,P}=(f,Q)$ and $\restrSig{\csRel_S}{((f,Q),[\oc_t])}=u$. Similarly, $\restrPot{\csRel_S}{e,P'.t'}=(e,Q'.u')$ with $\restrPot{\csRel_S}{f,P'}=(f,Q')$ and $\restrSig{\csRel_S}{((f,Q'),[\oc_{t'}])}=u'$. By supposition,
      \begin{align*}
        \restrCont{\csRel_S}{((f,P),T.\oc_t)} &= \restrCont{\csRel_S}{((f,P'),T'.\oc_{t'})}\\
        \left( \restrPot{\csRel_S}{f,P},\_@ \left[ \oc_{\restrSig{\csRel_S }{(\restrPot{\csRel_S}{f,P},[\oc_t])}} \right] \right ) &= \left ( \restrPot{\csRel_S}{f,P'}, \_@ \left[ \oc_{ \restrSig{\csRel_S}{(\restrPot{\csRel_S}{f,P'},[\oc_t'])} } \right] \right)\\
        ((f,Q),\_@[\oc_u]) &=((f,Q'),\_@[\oc_{u'}])\\
        Q.u &= Q'.u'
      \end{align*}

      The steps crossing auxiliary doors are similar to the steps crossing principal doors (dealt with above). To deal with the $\contLab$ node, one has to notice that $t \simpl u$ iff $\sigl(t) \simpl \sigl(u)$. The steps crossing $\digLab$ nodes are quite technical, but do not bring any insight on the result. Those steps are described in~\cite{perrinelMegathese}.
      \qed
    \end{prf}

    Theorem~\ref{injection_lemma_simple} allows us to trace back some $\noJump$ paths provided that we have some information about the last context of the path. In this subsection, we show how this implies an elementary bound (Lemma~\ref{theoStratElementaryBound}). But, first, we need some technical lemmas.

    \begin{lemma}[\cite{perrinelMegathese}]{\label{lemma_last_jump}}
      Let $\rightarrow \subseteq \csRel$. If $((\sigma(B),P),[\oc_t]) \rightarrow^* C$, then there exists a unique context $((\sigma(B'),P'),[\oc_{t'}])$ such that $((\sigma(B),P),[\oc_t]) \rightarrow^* ((\sigma(B'),P'),[\oc_{t'}]) (\noJump \cap \rightarrow)^* C$.
    \end{lemma}

Let $G$ be a $\stratSNLL$-stratified proof-net and $n \in \mathbb{N}$, we set $S_n=\Set*{B \in \boxset{G}}{\stratu{\stratSNLL}{B} \leq n}$. Let us notice that, if $\stratu{\stratSNLL}{B} \leq n$, the set of boxes $C$ such that $B \stratSNLL C$ is included in $S_{n-1}$. So, we will be able (thanks to Lemma~\ref{injection_lemma_simple}) to bound the number of copies of boxes of $S_n$ depending on the maximum number of copies of boxes of $S_{n-1}$. This corresponds to the round-by-round cut-elimination procedure used to prove the bounds on $ELL$, $LLL$, $L^3$ and $L^4$.

To make notations readable, we write $\csRel_n$ for $\csRel_{S_n}$, $\noJump_n$ for $\noJump_{S_n}$, $\restrSig{n}{((e,P),T)}$ for $\restrSig{\csRel_{S_n}}{((e,P),T)}$, $\copRel{n}{B,P}$ for $\copRel{\csRel_{S_n}}{B,P}$ and so on.

    \begin{lemma}\label{lemma_remonter_simple}Let $n \in \mathbb{N}$. If $((\sigma(B),P),[\oc_t]) \csRel_n C_k \cdots \csRel_n C_0$ and $\restrCont{n-1}{C_0}=\restrCont{n-1}{C'_0}$ then there exists $(C'_i)_{0 \leq i \leq k}$ such that $C'_k \csRel_n \cdots \csRel_n C'_0$ and, for $0 \leq i \leq k$, $\restrCont{n-1}{C_i}=\restrCont{n-1}{C'_i}$.
    \end{lemma}
    \begin{proof}
      We prove (by induction on $i$) the existence of a context $C'_i$ such that $C'_i \csRel_n C'_{i-1}$ and $\restrCont{n-1}{C_i}=\restrCont{n-1}{C'_i}$. If $i=0$, $C'_0$ satisfies the property by assumption. Otherwise, by induction hypothesis we know that there exists a context $C'_{i-1}$ such that $\restrCont{n-1}{(C_{i-1})}=\restrCont{n-1}{(C'_{i-1})}$.

      If the $C_i \csRel_n C_{i-1}$ step is a $\onlyJump$ step, it is of the shape $C_i=((\overline{\sigma_{j}(D)},Q),[\oc_u]) \onlyJump ((\sigma(D),Q),[\oc_u])=C_{i-1}$. So $C'_{i-1}$ is of the shape $((\sigma(D),Q'),[\oc_{u'}])$ with $\restrPot{n-1}{\sigma(D),Q}=\restrPot{n-1}{\sigma(D),Q'}=(\sigma(D),Q'')$ and $\restrSig{n-1}{((\sigma(D),Q''),[\oc_u])}=\restrSig{n-1}{((\sigma(D),Q''),[\oc_{u'}])}=u''$. Let us set $C'_i=((\overline{\sigma_j(D)},Q'),[\oc_{u'}])$. By Lemma~\ref{lemma_restrpot_eq}, $\restrPot{n-1}{\overline{\sigma_j(D)},Q}=\restrPot{n-1}{\overline{\sigma_j(D)},Q'}=(\overline{\sigma_j(D)},Q'')$. By Lemma~\ref{lemma_restr_eq}, $\restrSig{\csRel_S}{((\overline{\sigma_j(D)},Q''),[\oc_u])}=\restrSig{\csRel_S}{((\overline{\sigma_j(D)},Q''),[\oc_{u'}])}=u''$. So $\restrCont{\csRel_S}{C_i}=\restrCont{\csRel_S}{C'_i}=((\overline{\sigma_j(D)},Q''),[\oc_{u''}])$.

      Otherwise, $C_i \noJump C_{i-1}$ step so there exists a context of the shape $((\sigma(D),Q),[\oc_u])$ such that $((\sigma(B),P),[\oc_t]) \csRel_n ((\sigma(D),Q),[\oc_u]) (\noJump \cap \csRel_n)^+ C_{i-1}$ (Lemma~\ref{lemma_last_jump}). And, by definition of $\csRel_n$, $D \in S_n$. We prove that the last step of the path is in $\noJump_{n-1}$. We suppose $C_{i-1}$ is of the shape $((\overline{\sigma(D_i)},Q_i),[\oc_v])$ (otherwise, it is immediate by definition of $\noJump_{n-1}$). We can notice that $D \stratSNLL D_i$ so $\stratu{\stratSNLL}{D_i} < \stratu{\stratSNLL}{D} \leq n$, which means that $D_i \in S_{n-1}$. Thus, we have $C_{i} \noJump_{n-1} C_{i-1}$. By Theorem~\ref{injection_lemma_simple}, there exists a context $C'_i$ such that $C'_i \noJump_{n-1} C'_{i-1}$ and $\restrCont{n-1}{C_i}=\restrCont{n-1}{C'_i}$. 
    \end{proof}

\begin{lemma}[\cite{perrinelMegathese}]{\label{lemma_oclefttrace_remonter}}
  Let $S \subseteq \boxset{G}$. If $((e,P),[\oc_t]@T) \csRel_S^* ((f,Q),[\oc_u]@U)$, for every $u' \in \sig$, there exists $t' \in \sig$ such that $((e,P),[\oc_{t'}]@T) \csRel_S^* ((f,Q),[\oc_{u'}]@U)$.
\end{lemma}
\begin{proof}
  It is enough to prove it for one step. We can examine every possible step, each case is straightforward: the steps sometimes depend on $t$, never on $u$. For instance, let us suppose that $((e,P),[\oc_{\sigl(u)}]) \csRel_S ((f,Q),[\oc_u])$ (crossing a $\wn C$ upwards). Then, for every $u' \in \sig$, we have $((e,P),[\oc_{\sigl(u')}]) \csRel_S ((f,Q),[\oc_{u'}])$ so we can set $t'=\sigl(u')$.
\end{proof}

    \begin{lemma}[strong acyclicity]\label{lemma_strong_acyclicity_simple}
      Let $G$ be a normalizing proof-net. For every $n \in \mathbb{N}$, if $((\sigma(B),P),[\oc_t]) \csRel_n^* ((e,Q),[\oc_u]) \csRel_n^+ ((e,Q'),[\oc_v])$ then $(e,Q)^{n-1} \neq (e,Q')^{n-1}$. 
    \end{lemma}
    \begin{proof}
      We prove it by contradiction. We suppose that $((\sigma(B),P),[\oc_t]) \csRel_n^l ((e,Q),[\oc_u])$ and $((\sigma(B),P),[\oc_t]) \csRel^{l+m}_{S_n} ((e,Q'),[\oc_{u'}])=D'$, and $(e,Q)^{n-1} = (e,Q')^{n-1}$. Then, $\restrCont{n-1}{((e,Q),[\oc_{u'}])} = \restrCont{n-1}{D'}$. By Lemma~\ref{lemma_remonter_simple}, there exists a context $C'_1$ such that $C'_1 \csRel^{l+m} ((e,Q),[\oc_{u'}])$ and $\restrCont{n-1}{C'_1}=\restrCont{n-1}{((\sigma(B),P),[\oc_t])}$. So $C'_1$ is of the shape $((\sigma(B),P_1),[\oc_{t'_1}])$. By Lemma~\ref{lemma_oclefttrace_remonter}, there exists a signature $t_1$ such that $((\sigma(B),P_1),[\oc_{t_1}]) \csRel^{l+m} ((e,Q),[\oc_u])$ so $((\sigma(B),P_1),[\oc_{t_1}]) \csRel^{l+2m} ((e,Q'),[\oc_{u'}])$. 

      We define $C_1$ as the context $((\sigma(B),P_1),[\oc_{t_1}])$. For $k \in \mathbb{N}$, we can define by induction on $k$ a context $C_k=((\sigma(B),P_k),[\oc_{t_k}])$ such that $C_k \csRel_{n}^{l+k\cdot m} D$ and $C_k \csRel_n^{l+(k+1)\cdot m} D'$.

      Thus, if $m > 0$, we define an infinite path. In particular, this path goes through infinitely many contexts of shape $((\sigma(B),P'),[\oc_{t'}])$. According to Corollary \ref{theo_dallago_edges}, the number of canonical potentials for an edge is finite. So there is some $(\sigma(B),P') \in \can{\dirEdges{G}}$ and $v,v' \in Sig$ such that $((\sigma(B),P'),[\oc_v]) \csRel^+ ((\sigma(B),P'),[\oc_{v'}])$. This is impossible because normalizing proof-nets are acyclic (Lemma \ref{lemma_acyclicity}). This is a contradiction, so our hypothesis is wrong, $m=0$. There is no path of the shape $((\sigma(B),P),[\oc_t]) \csRel_n^* ((e,Q),[\oc_u]) \csRel^+_{n} ((e,Q'),[\oc_v])$ with $(e,Q)^{n-1}=(e,Q')^{n-1}$.
    \end{proof}

    \begin{figure}  \centering
      \tikzsetnextfilename{app_gh}
      \begin{tikzpicture}
        \node [proofnet,minimum width=1cm] (G) at (0,0) {G};
        \node [tensor] (tens) at ($(G)+(2,0)$) {};
        \node [cut]    (cut)  at ($(G)!0.5!(tens)+(0,-0.5)$) {};
        \draw [ar,out=-90,in=180] (G.-90) to (cut);
        \draw [ar,out=-90,in=  0] (tens)  to (cut);
        \node (etc) at ($(tens)+(1.1,-0.2)$) {};
        \node [ax] (ax) at ($(tens)!0.6!(etc)+(0,0.6)$) {};
        \draw [ar] (ax) to [out=-170,in=60] (tens);
        \draw [ar] (ax) to [out=-10, in=90] (etc);
        \node [proofnet,minimum width=1cm] (H) at ($(tens)+(-0.5,0.7)$) {$H$};
        \draw [ar] (H)--(tens);
      \end{tikzpicture}
      \caption{\label{fig_app_gh}This proof-net, written $(G)H$, corresponds to the application of a function $G$ to an argument $H$.}
    \end{figure}
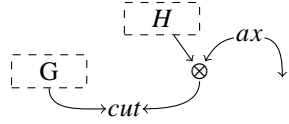

    \begin{lemma}[\cite{perrinelMegathese}]\label{lemma_numb_sig_depth}
      The number of signatures whose depth is $\leq d$ is at most $2^{2^{2\cdot d}}$
    \end{lemma}
    
    \begin{lemma}{\label{lemma_bound_cop_acyc}} If $\left |\Set*{ \restrPot{n-1}{e,Q} }{ \exists t,u \in \sig, ((\sigma(B),P),[\oc_t]) \csRel_n^* ((e,Q),[\oc_u]) } \right |\leq M$, then $\left|\copRel{n}{B,P}\right|$ is bounded by $2^{2^{2\cdot M}}$.
    \end{lemma}
    \begin{proof}
      Let us consider $u \in \sig$ such that there exists $t \in \copRel{n}{B,P}$ such that $t \simpl u$. By definition of $\copRel{\_}{\_}$ (Definition~\ref{def_arrowcopy}, in page~\pageref{def_arrowcopy}), there exists a path of the shape $((\sigma(B),P),[\oc_u]) \csRel_n^* ((\_,\_),[\oc_{\sige}])$. We consider $u$ as a tree. During the path beginning by $((\sigma(B),P),[\oc_u])$, the height of the left-most branch of $u$ (viewed as a tree) decreases to $0$ (the height of $\sige$). The height of the left-most branch decreases only by crossing a $\wn C$ or $\wn N$ nodes upwards (which corresponds to contexts of the shape $((e,Q),[\oc_v])$) and during those steps it decreases by exactly $1$. So the height of the left-most branch of $u$ is inferior to the number of instances of contexts of the shape $((e,Q),[\oc_v])$ through which the path goes. From Lemma~\ref{lemma_strong_acyclicity_simple}, each $\restrPot{n-1}{e,Q}$ is represented at most once. So the height of the left-most branch of $u$ is inferior to  $M$.
      
      Let $t$ be a $\csRel_n$-copy of $(B,P)$, then the height of $t$ is the height of its deepest branch. Once we consider signatures as trees, a simplification $u$ of $t$ can be viewed as a subtree of $t$ obtained as follows: we choose a branch of $t$ and $u$ is the part of $t$ on the right of this branch, in particular this branch becomes the leftmost branch of $u$. So there exists a simplification $u$ of $t$ such that the leftmost branch of $u$ is the deepest branch of $t$. So the heigth of $t$ is equal to the heigth of the leftmost branch of $u$. By the preceding paragraph, the height of the leftmost branch of $u$ is at most $M$ so the height of $t$ is at most $M$. The result is obtained by Lemma~\ref{lemma_numb_sig_depth}.
    \end{proof}
    
    In order to express elementary bounds, we define the notation $2^x_n$ (with $n \in \mathbb{N}$ and $x \in \mathbb{R}$) by induction on $n$: $2^x_0=x$ and $2^x_{n+1}=2^{2^x}_n$. So $2^x_n$ is a tower of exponentials of height $n$ with top exponent $x$.

    \begin{theorem}\label{theoStratElementaryBound}
      If a proof-net $G$ normalizes and is $\stratSNLL$-stratified, then the length of its longest path of reduction is bounded by  $2^{|\dirEdges{G}|}_{3 \stratG{G}}$
    \end{theorem}
    \begin{proof}
      By Lemma~\ref{lemma_bound_cop_acyc} and definition of $\canRel{n-1}{\_}$ we have:
      \begin{align*}    
        \max_{(B,P) \in \pot{\boxset{G}}}\left| \copRel{n}{B,P} \right| \leq & 2^{2^{2\cdot \left|\canRel{n-1}{\dirEdges{G}}\right|}}\\
        \max_{e \in \edges{G} }\left| \canRel{n}{e} \right| \leq & \left ( 2^{2^{2\cdot \left|\canRel{n-1}{\dirEdges{G}}\right|}} \right) ^{\depthG{G}}\\
        \left | \canRel{n}{\dirEdges{G}} \right | \leq & \left | \dirEdges{G} \right | \left ( 2^{\depthG{G} \cdot 2^{2\cdot \left|\canRel{n-1}{\dirEdges{G}}\right|}} \right)
      \end{align*}
      We define $u_n$ as $2^{\left|\dirEdges{G}\right|}_{3 \cdot n}$. We show by induction that, for every $n \in \mathbb{N}$, $\left | \canRel{n}{\dirEdges{G}} \right| \leq u_n$. For $n=0$, we can notice that for every $e \in \dirEdges{G}$, we have $|\canRel{0}{e}|=1$ (the only canonical potentials are lists of $\sige$) so $\left | \canRel{0}{\dirEdges{G}} \right | \leq \left|\dirEdges{G}\right| \leq u_0$. If $n \geq 0$, let us notice that $G$ has at least two boxes so $\left|\dirEdges{G}\right|\geq 4$. We have the following inequalities (to simplify the equations, we write $s$ for $\left|\dirEdges{G}\right|$):
      \begin{align*}
        \left | \canRel{n+1}{\dirEdges{G}} \right| &\leq s \left ( 2^{\depthG{G} \cdot 2^{2\cdot \left|\canRel{n}{\dirEdges{G}}\right|}} \right) \leq s \left ( 2^{\depthG{G} \cdot 2^{2\cdot u_n}} \right) \leq 2^{\frac{s}{2}} \left ( 2^{s \cdot 2^{2\cdot u_n}} \right)\\
        \log \left( \left | \canRel{n+1}{\dirEdges{G}} \right| \right) &\leq \frac{s}{2}+ s \cdot 2^{2\cdot u_n} \leq (2 \cdot s)  \cdot 2^{2\cdot u_n} \leq 2^{s+2\cdot u_n} \leq 2^{4 u_n} \leq 2^{2^{u_n}}\\
        \left | \canRel{n+1}{\dirEdges{G}} \right| &\leq 2^{u_n}_3= 2^{s}_{3n+3} = u_{n+1}
      \end{align*}
      Then, Theorem~\ref{theo_dallago_edges} gives us the announced bound.
    \end{proof}
    
    Let us consider the application of a proof-net $G$ to $H$ (Figure~\ref{fig_app_gh}). If $\stratSNLL$ is acyclic on $(G)H$, then $\left|\stratSNLL\right| \leq \left| \boxset{(G)H} \right| \leq \left|\boxset{G}\right| + \left|\boxset{H}\right|$. It is reasonable\footnote{More details at the end of Section~\ref{subsection_depcontrol_digging}.} to assume that the number of boxes does not depend on the argument of the function. So, by Theorem~\ref{theoStratElementaryBound}, the length of the normalization sequence is bounded by $e_G(x)$ with $x$ the size of the argument and $e_G$ an elementary function which does not depend on the argument.

\section{Paths criteria for polynomial time}
\label{section_polytime_simple}
\subsection{Dependence control}\label{section_dependence_control_simple}
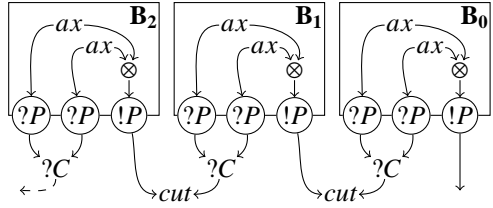
\begin{figure}
  \centering
  \tikzsetnextfilename{exp_proofnet}
  \begin{tikzpicture}
    \tikzstyle{door}=[draw, circle, inner sep=0.02cm]
    \node [princdoor] (bang1) at (0,0) {};
    \node at ($(bang1)+(0.2,1.3)$) {$\mathbf{B_2}$}; 
    \node [auxdoor] (ancr1) at ($(bang1)+(-0.65,0)$) {};
    \node [auxdoor] (ancl1) at ($(ancr1)+(-0.65,0)$) {};
    \node [cont] (why1) at ($(ancl1)!0.5!(ancr1)+(0,-0.7)$) {};
    \node [tensor]  (tens1) at ($(bang1)+(0,0.6)$) {};
    \node [ax] (ax1b) at ($(tens1)+(-0.4,0.3)$) {};
    \node [ax] (ax1h) at ($(tens1)+(-0.8,0.6)$) {};
    \draw [ar] (tens1) -- (bang1);
    \draw [ar] (ax1h) to [out=0,in=70]  (tens1);
    \draw [ar] (ax1b) to [out=0,in=110] (tens1);
    \draw [ar] (ax1h) to [out=180,in=90] (ancl1);
    \draw [ar] (ax1b) to [out=180,in=100] (ancr1);
    \node [princdoor] (bang2) at (2.2,0) {};
    \node at ($(bang2)+(0.2,1.3)$) {$\mathbf{B_1}$}; 
    \node [auxdoor] (ancr2) at ($(bang2)+(-0.65,0)$) {};
    \node [auxdoor] (ancl2) at ($(ancr2)+(-0.65,0)$) {};
    \node [cont] (why2) at ($(ancl2)!0.5!(ancr2)+(0,-0.7)$) {};
    \node [below right] at (why2) {};
    \node [tensor]  (tens2) at ($(bang2)+(0,0.6)$)   {};
    \node [ax]      (ax2b)  at ($(tens2)+(-0.4,0.3)$){};
    \node [ax]      (ax2h)  at ($(tens2)+(-0.8,0.6)$){};
    \draw [ar] (tens2)--(bang2);
    \draw [ar] (ax2h) to [out=0,in=70]  (tens2);
    \draw [ar] (ax2b) to [out=0,in=110] (tens2);
    \draw [ar] (ax2h) to [out=180,in=90] (ancl2);
    \draw [ar] (ax2b) to [out=180,in=100] (ancr2);
    \node [princdoor] (bang3) at (4.4,0) {};
    \node at ($(bang3)+(0.2,1.3)$) {$\mathbf{B_0}$}; 
    \node [auxdoor] (ancr3) at ($(bang3)+(-0.65,0)$) {};
    \node [auxdoor] (ancl3) at ($(ancr3)+(-0.65,0)$) {};
    \node [cont] (why3) at ($(ancl3)!0.5!(ancr3)+(0,-0.7)$) {};
    \node [below right] at (why3) {};
    \node [cut](cut1) at ($(bang1)!0.5!(why2)+(-0,-0.7)$) {};
    \node [cut](cut2) at ($(bang2)!0.5!(why3)+(-0,-0.7)$) {};
    \draw [ar](bang1) to [out=-80, in=180] (cut1);
    \draw [ar] (why2) to [out=-100, in=0]  (cut1);
    \draw [ar](bang2) to [out=-80, in=180] (cut2);
    \draw [ar](why3)  to [out=-100,in=0]   (cut2);
    \node [tensor] (tens3) at ($(bang3)+(0,0.6)$)    {};
    \node [ax]     (ax3b)  at ($(tens3)+(-0.4,0.3)$) {};
    \node [ax]     (ax3h)  at ($(tens3)+(-0.8,0.6)$) {};
    \draw [ar] (tens3)--(bang3);
    \draw [ar] (ax3h) to [out=0,in=70]  (tens3);
    \draw [ar] (ax3b) to [out=0,in=110] (tens3);
    \draw [ar] (ax3h) to [out=180,in=90] (ancl3);
    \draw [ar] (ax3b) to [out=180,in=100] (ancr3);
    \draw [ar] (bang3) -- ++ (0,-1);
    \draw (bang1) -|++(0.4,1.5) -| ($(ancl1)+(-0.3,0)$) -- (ancl1)--(ancr1)--(bang1);
    \draw (bang2) -|++(0.4,1.5) -| ($(ancl2)+(-0.3,0)$) -- (ancl2)--(ancr2)--(bang2);
    \draw (bang3) -|++(0.4,1.5) -| ($(ancl3)+(-0.3,0)$) -- (ancl3)--(ancr3)--(bang3);
    \draw [ar] (ancl1) to [bend right] (why1);
    \draw [ar] (ancr1) to [bend left]  (why1);
    \draw [ar] (ancl2) to [bend right] (why2);
    \draw [ar] (ancr2) to [bend left]  (why2);
    \draw [ar] (ancl3) to [bend right] (why3);
    \draw [ar] (ancr3) to [bend left] (why3);
    \node (suite) at ($(why1)+(-0.6,-0.3)$) {};
    \draw [ar,dashed] (why1) to [out=-90,in=0] (suite);
  \end{tikzpicture}
  \caption{\label{exp}This proof-net (if extended to $n$ boxes) reduces in ${\cal O}(2^n)$ reduction steps}
\end{figure}

Though $\stratSNLL$-stratification gives us a bound on the length of the reduction, elementary time is not considered as a reasonable bound, as it rises extremely fast with the size of the input. Cobham-Edmons thesis asserts that $Ptime$ corresponds to feasible problems. It suffers some limits:
\begin{itemize}
\item When one is only interested in very small inputs, the asymptotical complexity is not a concern
\item It does not account for constants and exponents.
\end{itemize}
However, in practice, the programs which we consider tractable mostly correspond to programs enjoying a polynomial bound on their time complexity. This is why we look for criteria entailing a polynomial bound on $W_G$. Figure~\ref{exp} shows us a way for the complexity to arise despite $\stratSNLL$-stratification. On this proof-net, $B_1$ has two residues. Each residue of $B_1$ creates two residues of $B_2$ (so $4$ residues in total). If we extend this sequence of boxes, $B_n$ has at least $2^n$ residues. From a context semantics perspective, $\left|\cop{B_i,[]}\right|$ depends non-additively on $\left|\cop{B_{i-1},[]}\right|$. Indeed, for any $t \in \cop{B_{i-1},[]}$, both $\sigl(t)$ and $\sigr(t)$ are in $\cop{B_i,[]}$. Thus, for every copy in $B_0$ there exist at least $2^{i}$ copies of $B_i$.

This proof-net is similar to the $\lambda$-term $(\lambda x.\la x,x\ra)\cdots (\lambda x.\la x,x\ra)y$ (in $\lambda$-calculus with pairs) which reduces to a $\lambda$-term of size $O(2^n)$ (with $n$ the number of successive applications of $\lambda x.\la x,x\ra$). Let us observe that the number of $\betared$ steps depends on the strategy: call-by-name normalizes in $\Theta(2^n)$ steps while call-by-value normalizes in $\Theta(n)$ steps (but, because the term size grows exponentially, the exectution time is in $\Theta(2^n)$ independently of the reduction strategy). The exponential blow-up happens because there are two free occurrences of $x$ in $\lambda x.\la x,x\ra$ (this corresponds in Figure~\ref{exp} to the two auxiliary doors by box which come from the same contraction node).

In~\cite{roversi2009some}, this situation is called a chain of {\em spindles}\label{def_spindle}. We call {\em dependence control condition} any restriction on linear logic which aims to tackle this kind of spindle chains. The dependence control in $LLL$~\cite{girard1995light} is to limit the number of auxiliary doors of each $\fpriLab$-box to $1$. The dependence control in $SLL$~\cite{lafont2004soft} is to forbid auxiliary doors above contraction nodes. 

However, those conditions forbid many proof-nets normalizing in polynomial time. For instance, the proof-net of Figure~\ref{expb} normalizes in linear time, even if the boxes have two auxiliary doors one of which is above a $\contLab$ node. The copies of $C_i$ depend on the copies of $C_{i-1}$ because $\cop{C_i,[]}=\{\sige,\sigr(\sige)\} \cup \Set*{\sigl(t)}{t \in \cop{C_{i-1},[]}}$. But the dependence is additive: $|\cop{C_i,[]}| = 2+|\cop{C_{i-1},[]}|$.

In terms of context semantics, to give a bound on the number of copies of a potential box, we want to trace back a path $((\sigma(B),P),[\oc_t]) \csRel^* ((e,Q),[\oc_{\sige}])$ with as little information on the path as possible. Theorem~\ref{injection_lemma_simple} (and the injectivity of $\noJump$) allows us to trace back $\noJump$ steps. However, we need additional information to trace back $\onlyJump$ steps because $\onlyJump$ is not injective. For instance, in Figure~\ref{exp}, we have:
\begin{align*}
  ((\sigma(B_2),[]), [\oc_{\sigl(\sige)}]) \noJump^2 C_e=((\overline{\sigma_1(B_1)},[]),[\oc_{\sige}]) \onlyJump ((\sigma(B_1),[]),[\oc_{\sige}])=C_f \\
  ((\sigma(B_2),[]), [\oc_{\sigr(\sige)}]) \noJump^2 C'_e=((\overline{\sigma_2(B_1)},[]),[\oc_{\sige}]) \onlyJump ((\sigma(B_1),[]),[\oc_{\sige}]) =C_f
\end{align*}
Let us consider a $((\sigma(B),P),[\oc_t]) (\noJump_S \cup \onlyJump)^* ((e,Q),[\oc_{\sige}])$ path. Thanks to Theorem~\ref{injection_lemma_simple} and Lemma~\ref{injection_lemma_jump}, we can trace it back (determine every edge of the path) provided we know $\restrCont{\csRel_S}{((e,Q),[\oc_{\sige}])}$ and, for every $((\overline{\sigma_i(C)},R),[\oc_u]) \onlyJump ((\sigma(C),R),[\oc_u])$ step of the path, we know $i$. 

\begin{lemma}\label{injection_lemma_jump}
Let $S$ be a subset of boxes. We suppose that $C_e=((\overline{\sigma_i(B)},P),[\oc_t]) \onlyJump C_f$, $C'_e=((\overline{\sigma_i(B)},P'),[\oc_{t'}]) \onlyJump C'_f$ and $\restrCont{\csRel_S}{C_f}=\restrCont{\csRel_S}{C'_f}$ then $\restrCont{\csRel_S}{C_e}=\restrCont{\csRel_S}{C'_e}$.
\end{lemma}
\begin{proof}
  Quite similar to the proof of Theorem~\ref{injection_lemma_simple} (cf. the study of the $\parr$ case).
\end{proof}

A dependence control condition is a criterion on proof-nets entailing a bound on the number of $\onlyJump$ steps for which we need to know the auxiliary edge to be able to trace back a $\csRel$-path. Instead of a syntactic criterion (like the ones of the type-systems $LLL$ and $SLL$), we propose here a semantic criterion on proof-nets. As in Section~\ref{chapter_3}, the criterion is defined as the acyclicity of a relation (written $\dcSim$) on boxes. Our criterion is more general than previous systems: every proof-net of (the multiplicative fragments of) $LLL$, $SLL$ and every $Ptime$ sound system of $MS$ satisfies our dependence control condition.

Intuitively $B \dcSim B'$ means that residues $B_1$ and $B_2$ of $B$ are cut, along reduction, with two distinct auxiliary doors ($\sigma_{i}(\_)$ and $\sigma_{j}(\_)$) of residues ($B'_1$ and $B'_2$) of $C$. From a context semantics point of view, it corresponds to the existence of $\csRel$-paths from the principal door of $B$ to two distinct auxiliary doors of $B'$. 

Let us observe that the relation $\dcSim$ is defined by considering $\csRel$-paths ending by a context on an (reversed) auxiliary edges of a box $B'$ while the relation $\stratSNLL$ (Definition~\ref{def_stratPot} in page~\pageref{def_stratPot}) was defined by considering $\noJump$-paths passing through the (reversed) principal edge of a box $B'$.

\begin{definition}\label{def_kjoins_simple}We set $B \dcSim B'$ iff there exist $i \neq j$ and paths of the shape:
  \begin{align*}
    ((\sigma(B),P), [\oc_t]) \csRel^+ ( (\overline{\sigma_i(B')}, P'_1), [\oc_{\sige}])\\ 
    ((\sigma(B),P), [\oc_u]) \csRel^+ ((\overline{\sigma_j(B')}, P'_2), [\oc_{\sige}]) 
  \end{align*}
\end{definition}

\begin{figure}
  \centering
  \tikzsetnextfilename{exp_proofnetb}
  \begin{tikzpicture}
    \tikzstyle{door}=[draw, circle, inner sep=0.02cm]
      \node [princdoor] (bang1) at (0,0) {};
      \node at ($(bang1)+(0.2,1.3)$) {$\mathbf{C_2}$}; 
      \node [auxdoor] (ancr1) at ($(bang1)+(-0.65,0)$) {};
      \node [auxdoor] (ancl1) at ($(ancr1)+(-1.2,0)$) {};
      \node [cont] (why1)  at ($(ancl1)!0.3!(ancr1)+(0,-0.7)$) {};
      \node [weak] (weak1) at ($(ancl1)!0.5!(ancr1)+(0,-0.2)$) {};
      \draw [ar] (weak1) -- (why1);
      \node [tensor]  (tens1) at ($(bang1)+(0,0.6)$) {};
      \node [ax] (ax1b) at ($(tens1)+(-0.4,0.3)$) {};
      \node [ax] (ax1h) at ($(tens1)+(-0.8,0.6)$) {};
      \draw [ar] (tens1) -- (bang1);
      \draw [ar] (ax1h) to [out=0,in=70]  (tens1);
    \draw [ar] (ax1b) to [out=0,in=110] (tens1);
    \draw [ar] (ax1h) to [out=180,in=90] (ancl1);
    \draw [ar] (ax1b) to [out=180,in=100] (ancr1);
    \node [princdoor] (bang2) at (2.7,0) {};
    \node at ($(bang2)+(0.2,1.3)$) {$\mathbf{C_1}$}; 
    \node [auxdoor] (ancr2) at ($(bang2)+(-0.65,0)$) {};
    \node [auxdoor] (ancl2) at ($(ancr2)+(-1.2,0)$) {};
    \node [cont] (why2)  at ($(ancl2)!0.3!(ancr2)+(0,-0.7)$) {};
    \node [weak] (weak2) at ($(ancl2)!0.5!(ancr2)+(0,-0.2)$) {};
    \draw [ar] (weak2) -- (why2);
    \node [tensor]  (tens2) at ($(bang2)+(0,0.6)$)   {};
    \node [ax]      (ax2b)  at ($(tens2)+(-0.4,0.3)$){};
    \node [ax]      (ax2h)  at ($(tens2)+(-0.8,0.6)$){};
    \draw [ar] (tens2)--(bang2);
    \draw [ar] (ax2h) to [out=0,in=70]  (tens2);
    \draw [ar] (ax2b) to [out=0,in=110] (tens2);
    \draw [ar] (ax2h) to [out=180,in=90] (ancl2);
    \draw [ar] (ax2b) to [out=180,in=100] (ancr2);
    \node [princdoor] (bang3) at (5.4,0) {};
    \node at ($(bang3)+(0.2,1.3)$) {$\mathbf{C_0}$}; 
    \node [auxdoor] (ancr3) at ($(bang3)+(-0.65,0)$) {};
    \node [auxdoor] (ancl3) at ($(ancr3)+(-1.2,0)$) {};
    \node [cont] (why3)  at ($(ancl3)!0.3!(ancr3)+(0,-0.7)$) {};
    \node [weak] (weak3) at ($(ancl3)!0.5!(ancr3)+(0,-0.2)$) {};
    \draw [ar] (weak3) -- (why3);
    \node [cut](cut1) at ($(bang1)!0.5!(why2)+(-0,-0.7)$) {};
    \node [cut](cut2) at ($(bang2)!0.5!(why3)+(-0,-0.7)$) {};
    \draw [ar](bang1) to [out=-80, in=180] (cut1);
    \draw [ar] (why2) to [out=-100, in=0]  (cut1);
    \draw [ar](bang2) to [out=-80, in=180] (cut2);
    \draw [ar](why3)  to [out=-100,in=0]   (cut2);
    \node [tensor] (tens3) at ($(bang3)+(0,0.6)$)    {};
    \node [ax]     (ax3b)  at ($(tens3)+(-0.4,0.3)$) {};
    \node [ax]     (ax3h)  at ($(tens3)+(-0.8,0.6)$) {};
    \draw [ar] (tens3)--(bang3);
    \draw [ar] (ax3h) to [out=0,in=70]  (tens3);
    \draw [ar] (ax3b) to [out=0,in=110] (tens3);
    \draw [ar] (ax3h) to [out=180,in=90] (ancl3);
    \draw [ar] (ax3b) to [out=180,in=100] (ancr3);
    \draw [ar] (bang3) -- ++ (0,-1);
    \draw (bang1) -|++(0.4,1.5) -| ($(ancl1)+(-0.3,0)$) -- (ancl1)--(ancr1)--(bang1);
    \draw (bang2) -|++(0.4,1.5) -| ($(ancl2)+(-0.3,0)$) -- (ancl2)--(ancr2)--(bang2);
    \draw (bang3) -|++(0.4,1.5) -| ($(ancl3)+(-0.3,0)$) -- (ancl3)--(ancr3)--(bang3);
    \draw [ar] (ancl1) to [bend right] (why1);
    \draw [ar] (ancr1) --++ (0,-1)  (why1);
    \draw [ar] (ancl2) to [bend right] (why2);
    \draw [ar] (ancr2) --++ (0,-1) (why2);
    \draw [ar] (ancl3) to [bend right] (why3);
    \draw [ar] (ancr3) --++ (0,-1) (why3);
    \node (suite) at ($(why1)+(-0.6,-0.3)$) {};
    \draw [ar,dashed] (why1) to [out=-90,in=0] (suite);
  \end{tikzpicture}
  \caption{\label{expb}This proof-net (if extended to $n$ boxes) reduces in ${\cal O}(n)$ reduction steps.}
\end{figure}
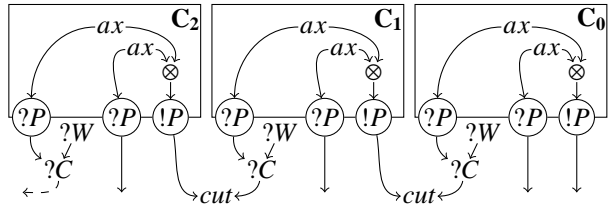

In Figure~\ref{exp}, we have $B_i \dcSim B_{i-1}$ because $((\sigma(B_i),[]),[\oc_{\sigl(\sige)}]) \csRel^2 ((\overline{\sigma_1(B_{i-1})},[]),[\oc_{\sige}])$ and $((\sigma(B_i),[]),[\oc_{\sigr(\sige)}]) \csRel^2 ((\overline{\sigma_2(B_{i-1})},[]),[\oc_{\sige}])$. Similarly, the proof-net of Figure~\ref{reallyExp} is not $\dcSim$-stratified because $B \dcSim B$. On the contrary, in Figure~\ref{expb}, $\dcSim=\varnothing$.

    \begin{figure}
      \centering
      \begin{tikzpicture}
        \node [princdoor]  (bang) at (0,0) {};
        \node at ($(bang)+(0.32,0.3)$) {$\mathbf{B}$};
        \node [auxdoor]    (aux2)  at ($(bang)+(-0.8,0)$) {};
        \node [auxdoor]    (aux1)  at ($(aux2) +(-0.6,0)$) {};
        \node [cont]       (cont)  at ($(aux1)!0.5!(aux2)+(0,-0.6)$) {};
        \node [tensor]     (tens)  at ($(bang)+(0,0.6)$)  {};
        \node [ax]         (axb)   at ($(tens)+(-0.7,0.3)$) {};
        \node [ax]         (axh)   at ($(tens)+(-0.7,0.6)$) {};
        \draw (bang) -| ++(0.5,1.4) -| ($(aux1)+(-0.32,0)$) -- (aux1) -- (aux2) -- (bang);
        \draw [ar] (tens)--(bang);
        \draw [ar] (axh) to [out=  0,in= 75]  (tens);
        \draw [ar] (axb) to [out=  0,in=120]  (tens);
        \draw [ar] (axh) to [out=180,in= 90] (aux1);
        \draw [ar] (axb) to [out=180,in= 90] (aux2);
        \draw [ar] (aux1) -- (cont);
        \draw [ar] (aux2) -- (cont);
        \node [par]        (par)   at ($(cont)!0.5!(bang)+(0,-1.2)$) {};
        \node [forall]     (fa)    at ($(cont)!0.5!(par)$) {};
        \draw [name path=bangPar,opacity=0] (bang)--(par);
        \draw [name path=horizontFa,opacity=0] ($(fa)+(-2,0)$) -- ($(fa)+(2,0)$);
        \node [exists, name intersections={of=bangPar and horizontFa}] (ex) at (intersection-1) {};
        \draw [ar] (cont) -- (fa); \draw [ar] (fa)--(par);
        \draw [ar] (bang) -- (ex); \draw [ar] (ex)--(par);

        \node [princdoor] (bangf) at ($(par)+(0,-0.6)$) {};
        \draw [ar] (par) -- (bangf);
        \draw (bangf) -| ++(1.15,3.6) -| ($(bangf)+(-1.25,0)$) -- (bangf);
        \node [tensor]    (tensf) at ($(bangf)+(-13:1.2)$) {};
        \node [der]       (derf)  at ($(tensf)+( 50:0.7)$) {};
        \node [tensor]    (tensfx)at ($(derf) +( 60:0.7)$) {};
        \node [exists]    (exx)   at ($(tensfx)+(120:0.7)$) {};
        \node [par]       (parx)  at ($(exx)  +(0,0.6)$) {};
        \node [ax]        (axx)   at ($(parx) +(0,0.7)$) {};
        \node [ax]        (axfx)  at ($(tensfx)+(0.6,0.5)$) {};
        \node [exists]    (exf)   at ($(tensf)+(-60:0.6)$)  {};
        \node [par]       (parf)  at ($(exf)  +(-60:0.6)$)  {};
        \draw [ar] (axx.0)   to [out= -20,in= 60] (parx);
        \draw [ar] (axx.180) to [out=-160,in=120] (parx);
        \draw [ar] (parx)-- (exx); 
        \draw [ar] (exx) -- (tensfx);
        \draw [ar] (tensfx)--(derf);
        \draw [ar] (derf)  --(tensf);
        \draw [ar] (bangf) to [out=-80, in=162] (tensf);
        \draw [ar] (tensf) --(exf);
        \draw [ar] (exf)   --(parf);
        \draw [ar] (axfx) to [out=-170,in=60] (tensfx);
        \draw [ar] (axfx) to [out=-10, in=60] (parf);
        
        \node [tensor]  (app) at ($(parf)+(5,0)$) {};
        \node [cut]     (cut) at ($(parf)!0.5!(app)+(0,-0.6)$) {};
        \draw [ar] (parf) to [out=-40,in=180] (cut);
        \draw [ar] (app) to [out=-140,in=  0] (cut);
        \node [ax]     (finax) at ($(app)+(0.8,0.3)$) {};
        \draw [ar]     (finax) to [out=180,in= 60] (app);
        \draw [ar]     (finax) to [out=  0,in=120] ($(finax)+(0.8,-0.5)$);

        \node [forall]       (forall)   at ($(app)+(-0.3,0.8)$) {};
        \draw [ar] (forall) -- (app);
        \node [par]          (parg)     at ($(forall)+(0,0.65)$) {};
        \node [princdoor]    (princ3g)  at ($(parg)+(1, 1.5)$) {};
        \node at ($(princ3g)+(0.5,0.4)$) {$\mathbf C$}; 
        \node [auxdoor]      (aux1g)    at ($(princ3g)+(-2.5,0)$) {};
        \node [auxdoor]      (aux2g)    at ($(aux1g)!0.333!(princ3g)$) {};
        \node [auxdoor]      (aux3g)    at ($(aux1g)!0.666!(princ3g)$) {};
        \node [cont]         (cont1g)   at ($(aux1g)!0.5!(aux2g)+(0,-0.9)$) {};
        \node [cont]         (cont2g)   at ($(cont1g)!0.5!(parg)$) {};
        \nvar{\hautTens}{1.1cm}
        \node [tensor]       (tens1g)   at ($(aux1g)+(0,\hautTens)$) {};
        \node [tensor]       (tens2g)   at ($(aux2g)+(0,\hautTens)$) {};
        \node [tensor]       (tens3g)   at ($(aux3g)+(0,\hautTens)$) {};
        \nvar{\decAx}{0.4cm}
        \node [ax]        (ax1n)     at ($(tens1g)+(-0.5,\decAx)$) {};
        \node [ax]        (ax2n)     at ($(tens1g)!0.5!(tens2g)+(0,\decAx)$) {};
        \node [ax]        (ax3n)     at ($(tens2g)!0.5!(tens3g)+(0,\decAx)$) {};
        \node [ax]        (ax4n)     at ($(princ3g)+(-0.3,\decAx + \hautTens)$) {};
        \node [par]          (parx)     at ($(princ3g)+(0,0.6)$) {};
        \draw [ar] (parg) -- (forall);
        \draw [ar] (cont2g) -- (parg);
        \draw [ar] (cont1g) -- (cont2g);
        \draw [ar] (aux1g) -- (cont1g);
        \draw [ar] (aux2g) -- (cont1g);
        \draw [ar] (aux3g) -- (cont2g);
        \draw [ar] (parx) -- (princ3g);
        \draw [ar] (tens1g) -- (aux1g);
        \draw [ar] (tens2g) -- (aux2g);
        \draw [ar] (tens3g) -- (aux3g);
        \draw [ar,out=-160,in= 170] (ax1n) to (parx);
        \draw [ar,out=0 ,in= 120] (ax1n) to (tens1g);
        \draw [ar,out=180 ,in=  60] (ax2n) to (tens1g);
        \draw [ar,out=0   ,in=120 ] (ax2n) to (tens2g);
        \draw [ar,out=180 ,in=  60] (ax3n) to (tens2g);
        \draw [ar,out=0   ,in=120 ] (ax3n) to (tens3g);
        \draw [ar,out=180 ,in=  60] (ax4n) to (tens3g);
        \draw [ar,out=-20 ,in=  60] (ax4n) to (parx);
        \draw (princ3g) -| ++(0.7,\hautTens + \decAx +0.3cm) -| ($(aux1g)+(-1.2,0)$) -- (aux1g) -- (aux2g) -- (aux3g) -- (princ3g);
        \draw (princ3g) -- (parg);

        \draw [dashed] ($(aux1)+(-0.6,1.6)$) rectangle ($(parf)+(1.6,-0.5)$);
        \node at ($(parf)+(1.3,0)$) {$\mathbf{G}$};
        \draw [dashed] ($(forall)+(2,-0.3)$) rectangle ($(ax1n)+(-0.8,0.55)$);
        \node at ($(forall)+(1.7,0)$) {$\mathbf{H}$};
      \end{tikzpicture}
      \caption{\label{reallyExp}The complexity of $G$ is not polynomial.}
      \end{figure}
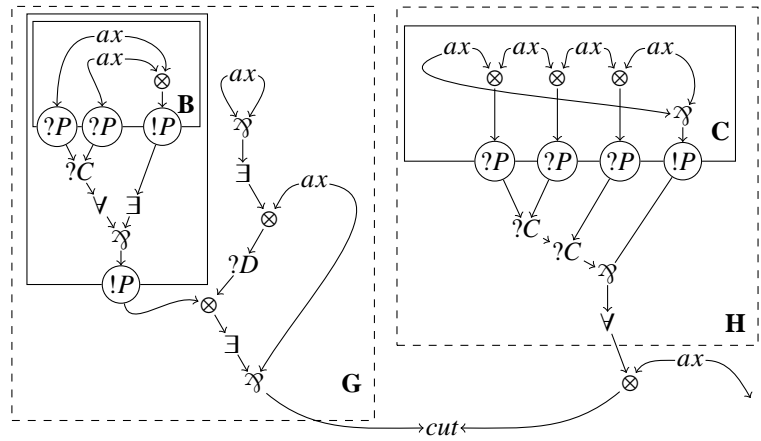

\begin{lemma}\label{lemma_sequences_edges_simple}
  Let $G$ be a $\stratSNLL$-stratified proof-net, $s\in \mathbb{N}$ and $(B,P)$ be a potential box with $d=\stratu{\dcSim}{B}$. There are at most $\left|\canRel{s-1}{\dirEdges{G}}\right|^d$ sequences $(e_i)_{1\leq i \leq l}$ of directed edges such that, there exists a potential sequence $(P_i)_{1 \leq i \leq l}$, a trace sequence $(T_i)_{1 \leq i < l}$ and $t \in \sig$ such that:
  \begin{equation*}
    ((\sigma(B),P),[\oc_t]) \csRel_{s} ((e_1,P_1),T_1) \csRel_{s} \cdots \csRel_{s} ((e_{l-1},P_{l-1}),T_{l-1}) \csRel_{s} ((e_l,P_l),[\oc_{\sige}])
  \end{equation*}
\end{lemma}
\begin{proof}
  We prove it by induction on $d$. We suppose that there exists a path of the shape $((\sigma(B),P),[\oc_t]) \hspace{-0.15em}\csRel_s\hspace{-0.15em} ((e_1,P_1),T_1) \hspace{-0.15em}\csRel_s\hspace{-0.15em} \cdots \hspace{-0.15em}\csRel_s\hspace{-0.15em} ((e_{l-1},P_{l-1}),T_{l-1}) \hspace{-0.15em}\csRel_s\hspace{-0.15em} ((e_l,P_l),[\oc_{\sige}])$. If there is a context in the path of the shape $((\sigma(C),Q),[\oc_{\_}])$ with $\stratu{\dcSim}{C} < \stratu{\dcSim}{B}$, we set $k$ as the smallest index such that $((e_{k+1},P_{k+1}),T_{k+1})$ is such a context. Otherwise, we set $k=l$.
  
  First, let us notice that by induction hypothesis, there are at most $|\canRel{s-1}{\dirEdges{G}}|^{d-1}$ possibilities for $e_{k+1},\cdots,e_l$. Then, let us determine the number of possibilities for $e_1,\cdots,e_k$. There are at most $|\canRel{s-1}{\dirEdges{G}}|$ choices for $(e_k,P_k)^{s-1}$. Once $(e_k,P_k)^{s-1}$ is determined, we will prove by contradiction that it determines $e_1,\cdots,e_k$. Let us suppose that there exists two possible sequences: $e_1,\cdots,e_k$ and $e'_1,\cdots,e'_{k'}$. Then we consider the lowest $j$ such that $\restrCont{s-1}{((e_{k-j},P_{k-j}),T_{k-j})} \neq \restrCont{s-1}{((e'_{k'-j},P'_{k'-j}),T'_{k'-j})}$. By assumption we have $k>0$ and, by Theorem~\ref{injection_lemma_simple}, the ``$k-j$ and $k'-j$ steps'' must be $\onlyJump$ steps:
  \begin{align*}
    C_{k-j}=((\overline{\sigma_{i_1}(D)},P_{k-j}),[\oc_{v}])   &\onlyJump ((\sigma(D),P_{k-j}),[\oc_v])=C_{k+1-j}\\
    C'_{k'-j}=((\overline{\sigma_{i_2}(D)},P'_{k'-j}),[\oc_{v'}]) &\onlyJump ((\sigma(D),P'_{k'-j}),[\oc_{v'}])=C'_{k'+1-j}
  \end{align*}
  with $\restrPot{s-1}{C_{k+1-j}}=\restrPot{s-1}{C'_{k'+1-j}}$ and $\restrPot{s-1}{C_{k-j}} \neq \restrPot{s-1}{C'_{k'-j}}$. By Lemma~\ref{lemma_restrpot_eq}, the difference is not on the potential and by Lemma~\ref{lemma_restr_eq} the difference is not on the trace, so the difference is on the edge: $i_1 \neq i_2$. By definition of $\dcSim$, it means that $B \dcSim D$ and $\stratu{\dcSim}{D}<\stratu{\dcSim}{B}$. This contradicts the definition of $k$. So our hypothesis is false: if we fix $\restrPot{s-1}{e_k,P_k}=\restrPot{s-1}{e_{k'},P_{k'}}$, then $[e_1;\cdots;e_k]=[e'_1;\cdots;e'_{k'}]$.
 
  Thus, we proved that there are at most $|\canRel{s-1}{\dirEdges{G}}|$ possibilities for $e_1,\cdots,e_k$ and at most $|\canRel{s-1}{\dirEdges{G}}|^{d-1}$ possibilities for $e_{k+1},\cdots,e_l$. In total, there are at most $|\canRel{s-1}{\dirEdges{G}}|^d$ possibilities for $e_1,\cdots,e_l$.
\end{proof}

\subsection{Nesting}\label{subsection_depcontrol_digging}
Lemma~\ref{lemma_sequences_edges_simple} bounds the number of paths corresponding to copies, provided that $\stratSNLL$ and $\dcSim$ are acyclic. In the absence of $\digLab$ nodes, a copy $t$ of $(B,P)$ only contains $\sigl(\_)$, $\sigr(\_)$ and $\sige$ constructions. One can reconstruct $t$ by observing the list of contexts in the path, of the shape $((e_i,P_i),[\oc_{t_i}])$ with $\overline{e_i}$ being a premise of a contraction node. This is entirely determined by the sequence $e_1,\cdots,e_l$ of edges of the path $((\sigma(B),P),[\oc_{t}]) \csRel ((e_1,\_),\_) \csRel \cdots ((e_l,\_),[\oc_{\sige}])$. Thus, if there is no $\digLab$ node, Lemma~\ref{lemma_sequences_edges_simple} bounds the number of copies of $(B,P)$.

To understand why the $\digLab$ nodes break this property, we can consider an example in Figure~\ref{exp2_dig}. We can notice that $\stratSNLL$  and $\dcSim$ are both the empty relation so $\stratu{\stratSNLL}{B_2}=\stratu{\stratSNLL}{B_1}=\stratu{\stratSNLL}{B_0}=1$ and $\stratu{\dcSim}{B_2}=\stratu{\dcSim}{B_1}=\stratu{\dcSim}{B_0}=1$. However, if extended to $n$ boxes, $|\cop{B_n,[]}|\geq 3^n$ and the number of $\cutRel$ steps is not polynomial in $n$.

To guide intuition, we can study a similar situation in $\lambda$-calculus with pairs. The $\lambda$-term $(\lambda x.(\lambda y.\la y,y\ra)x)\cdots (\lambda x.(\lambda y.\la y,y\ra)x)z$ reduces to a $\lambda$-term of size $O(2^n)$ (with $n$ the number of successive applications of $\lambda x.(\lambda y.\la y,y\ra)x$). In this case $x$ has only one free occurrence in $\lambda x.(\lambda y.\la y,y\ra)x$ (it corresponds to the fact that there is only one auxiliary door in the boxes of Figure~\ref{exp2_dig}) however $x$ is duplicated inside $\lambda x.(\lambda y.\la y,y\ra)x$ (this term reduces to $\lambda x.\la x,x\ra$). This corresponds to the $\wn C$ node inside the boxes $B_i$ of Figure~\ref{exp2_dig}, which duplicates the box $B_{i-1}$. This is possible because the box $B_{i-1}$ gets inside the box $B_i$ because of the $\wn N$ node.

We call {\em nesting} any restriction on linear logic which aims to tackle this kind of chains. The nesting in $LLL$~\cite{girard1995light}, $SLLL$~\cite{lafont2004soft}, $mL^4$~\cite{baillot2010linear} and $MS$~\cite{roversi2009some} is the absence of $\digLab$ node. Lemma~\ref{lemma_sequences_edges_simple} states that there are at most $|\dirEdges{G}|$ sequences of edges corresponding to copies of $(B_2,[])$, the sequence being entirely determined by the last edge\footnote{Indeed $S_0=\varnothing$ and, for every potential edge $\restrPot{\csRel_{\varnothing}}{e_l,P_l}=(e_l,[\sige;\cdots;\sige])$. So knowing $\restrPot{0}{e_l,P_l}$ is equivalent to knowing the last edge of the path.}. For instance, knowing that $((\sigma(B_2),[]),[\oc_t]) \csRel_1^* ((\overline{l},[p]),[\oc_{\sige}])$ is enough to deduce that: 
\begin{equation*}
  ((\sigma(B_2),[]),[\oc_t]) \csRel_1 ((\overline{f},[]),[\oc_t]) \csRel_1 ((\overline{\sigma_1(B_1)},[]),[\oc_{\sigl(\sige)};\oc_{p}]) \csRel_1^2 ((\overline{l},[p]),[\oc_{\sige}])
\end{equation*}

Thus, we can deduce that $t$ is of the shape $\sign(\sigl(\sige),p)$, but we do not know $t$ entirely because $p$ can be any element of $\cop{B_1,[]}=\{\sige,\sign(\sige,\sige),\sign(\sigl(\sige),\sige),\sign(\sigr(\sige),\sige)\}$.  
 
Following the paths backwards we can observe that the most important step is $((\overline{f},[]),[\oc_{\sign(\sige, \sign(p,\sige))}]) \csRel_1 ((\overline{\sigma_1(B)},[],[\oc_{\sige};\oc_{\sign(p,\sige)}])$ where a difference on the second trace element (which comes from $B_1$ with $\stratu{\stratSNLL}{B_1}=1$) becomes a difference on the first trace element, which corresponds to $t$. The paths corresponding to  $\sign(\sige, \sign(\sigl(\sige),\sige))$ and $\sign(\sige, \sign(\sigr(\sige),\sige))$ are the same, but the paths corresponding to their simplifications are different.

The dependence of $|\cop{B_2,[]}|$ on $|\cop{B_1,[]}|$ in Figure~\ref{exp2_dig} is similar to the dependence in Figure~\ref{exp}. We define a relation $\nestSim$ on  boxes capturing this dependence. Intuitively $B \nestSim C$ means that $B$ is cut with a $\digLab$ node along reduction and the outer residue $B_e$ of $B$ is cut with an auxiliary door of $C$. The acyclicity of $\nestSim$ is a nesting condition.

\begin{figure}
  \centering    \tikzsetnextfilename{exp_proofnetdig2} 
  \begin{tikzpicture}
    \tikzstyle{door}=[draw, circle, inner sep=0.02cm]
    \node [princdoor] (bang1) at (0,0) {};
    \node (name1) at ($(bang1)+(0.4,0.3)$) {$\mathbf{B_2}$}; 
    \node [auxdoor]   (aux1)  at ($(bang1)+(-1,  0)$) {};
    \node [cont]      (cont1) at ($(aux1) +(   0,0.7)$) {};
    \node [tensor]    (tens1) at ($(bang1)+(   0,0.7)$) {};
    \node [ax]        (ax1b)  at ($(tens1)!0.45!(cont1)+(0,0.4)$) {};
    \node [ax]        (ax1h)  at ($(tens1)!0.45!(cont1)+(0,0.7)$) {};
    \node [dig]       (dig1)  at ($(aux1) +(0,-0.7)$){};
    \draw (bang1) -| ++(0.65,1.6) -| ($(aux1)+(-0.4,0)$) -- (aux1) -- (bang1);
    \draw [ar] (tens1)--(bang1);
    \draw [ar] (ax1b) to [out=  0,in=120] (tens1);
    \draw [ar] (ax1h) to [out=-10,in=70]  (tens1);
    \draw [ar] (ax1b) to [out=-170,in=60] (cont1);
    \draw [ar] (ax1h) to [out=-170,in=120](cont1);
    \draw [ar] (cont1) -- (aux1);        
    \draw [ar] (aux1)  -- (dig1);
    \draw [ar,dashed] (dig1) to [out=-90,in=0] ($(dig1)+(-0.4,-0.4)$);
    
    \node [princdoor] (bang2) at ($(bang1)+(2.6,0)$) {};
    \node (name2) at ($(bang2)+(0.4,0.3)$) {$\mathbf{B_1}$}; 
    \node [auxdoor]   (aux2)  at ($(bang2)+(-1.2,  0)$) {};
    \node [cont]      (cont2) at ($(aux2) +(   0,0.7)$) {};
    \node [tensor]    (tens2) at ($(bang2)+(   0,0.7)$) {};
    \node [ax]        (ax2b)  at ($(tens2)!0.5!(cont2)+(0,0.4)$) {};
    \node [ax]        (ax2h)  at ($(tens2)!0.5!(cont2)+(0,0.7)$) {};
    \node [dig]       (dig2)  at ($(aux2) +(   0,-0.7)$){};
    \draw (bang2) -| ++(0.65,1.6) -| ($(aux2)+(-0.4,0)$) -- (aux2) -- (bang2);
    \draw [ar] (tens2)--(bang2);
    \draw [ar] (ax2b) to [out=  0,in=120] (tens2);
    \draw [ar] (ax2h) to [out=-10,in=70]  (tens2);
    \draw [ar] (ax2b) to [out=-170,in=60] (cont2);
    \draw [ar] (ax2h) to [out=-170,in=120] node [edgename,above left=-0.05cm] {$l$} (cont2);
    \draw [ar] (cont2) -- (aux2) node [edgename,right] {$e$};        
    \draw [ar] (aux2)  -- (dig2);
    \node [cut](cut12) at ($(bang1)!0.5!(aux2)+(0,-1)$) {};
    \draw [ar] (bang1) to [out=-90,in=180] (cut12);
    \draw [ar] (dig2)  to [out=-100,in=  0] node [edgename,pos=0.4,right] {$f$} (cut12);
    
    \node [princdoor] (bang3) at ($(bang2)+(2.6,0)$) {};
    \draw [ar] (bang3) --++ (0,-1);
    \node (name3) at ($(bang3)+(0.4,0.3)$) {$\mathbf{B_0}$}; 
    \node [auxdoor]   (aux3)  at ($(bang3)+(-1.2,  0)$) {};
    \node [cont]      (cont3) at ($(aux3) +(   0,0.7)$) {};
    \node [tensor]    (tens3) at ($(bang3)+(   0,0.7)$) {};
    \node [ax]        (ax3b)  at ($(tens3)!0.5!(cont3)+(0,0.4)$) {};
    \node [ax]        (ax3h)  at ($(tens3)!0.5!(cont3)+(0,0.7)$) {};
    \node [dig]       (dig3)  at ($(aux3) +(   0,-0.7)$){};
    \draw (bang3) -| ++(0.65,1.6) -| ($(aux3)+(-0.4,0)$) -- (aux3) -- (bang3);
    \draw [ar] (tens3)--(bang3);
    \draw [ar] (ax3b) to [out=  0,in=120] (tens3);
    \draw [ar] (ax3h) to [out=-10,in=70]  (tens3);
    \draw [ar] (ax3b) to [out=-170,in=60] (cont3);
    \draw [ar] (ax3h) to [out=-170,in=120](cont3);
    \draw [ar] (cont3) -- (aux3) node [edgename] {};
    \draw [ar] (aux3)  to (dig3);
    \node [cut](cut23) at ($(bang2)!0.5!(aux3)+(0,-1)$) {};
    \draw [ar] (bang2) to [out=-90,in=180] (cut23);
    \draw [ar] (dig3)  to [out=-100,in=  0] node [edgename,pos=0.4,right] {$g$} (cut23);
  \end{tikzpicture}
  \caption{This proof-net (if extended to $n$ boxes) reduces in ${\cal O}(2^n)$ reduction steps.} \label{exp2_dig}
\end{figure}
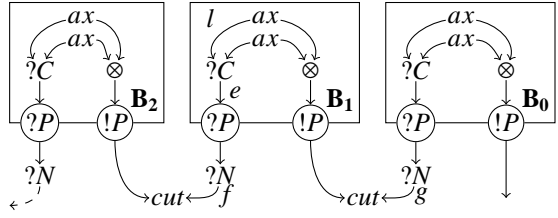                                      

\begin{definition}\label{def_nestsim}
  We set $B \nestSim C$ if there exists a non-standard signature $t$ and a path of the shape:
  \begin{equation*}
    ((\sigma(B),P),[\oc_t]) \csRel^+ ((\sigma(C),Q),[\oc_{\sige}])
  \end{equation*}
\end{definition}

For example, in Figure~\ref{exp2_dig}, we have $B_2 \nestSim B_1$ because $\sigp(\sige)$ is non-standard and $((\sigma(B_2),[]),[\oc_{\sigp(\sige)}]) \csRel^3 ((\sigma(B_1),[]),[\oc_{\sige}])$. To prove that $\nestSim$-stratification (together with $\stratSNLL$-stratification and $\dcSim$-stratification) implies polynomial time, we will need some technical lemmas to handle simplifications of copies.

In the following, we consider a $\stratSNLL$-stratified, $\nestSim$-stratified, $\dcSim$-stratified proof-net $G$. Let $s,n \in \mathbb{N}$, we set\label{def_Tsn} $T_{s,n}= \Set*{ B \in \boxset{G} }{ (\stratu{\stratSNLL}{B},\stratu{\nestSim}{B}) \leqlex (s,n)  }$ with $\leqlex$ the usual lexicographic order: \label{def_leqlex}$(a,b) \leqlex (a',b')$ iff $a < a'$ or ($a \leq a'$ and $b \leq b'$). To simplify notations, we write $\csRel_{s,n}$ for $\csRel_{T_{s,n}}$, $\copRel{s,n}{C}$ for $\copRel{\csRel_{T_{s,n}}}{C}$ and so on.

\begin{lemma}\label{lemma_bound_poly_toutca_simple}
  For $s,n \in \mathbb{N}-\{0\}$ and $(B,P) \in \can{\boxset{G}}$,
  \begin{equation*}
    \left | \copRel{s,n}{B,P} \right | \leq  \left | \canRel{s-1}{\dirEdges{G}} \right |^{\left|\dcSim\right|} \cdot  \left |\dirEdges{G}\right | \cdot  \left ( \max_{(C,Q) \in \pot{\boxset{G}}}   \left| \copRel{s,n-1}{C,Q} \right| \right )^{\depthG{G}}
  \end{equation*}
\end{lemma}
\begin{proof}
  If $\stratu{\nestSim}{B}>n$, then $((\sigma(B),P),[\oc_{\_}]) \not \csRel_{s,n}$ and $\copRel{s,n}{B,P}=\{\sige\}$ so the lemma stands. Otherwise (if $\stratu{\nestSim}{B}\leq n$), let us consider $t,t' \in \copRel{s,n}{B,P}$. By definition, there exists paths of the shape:
  \begin{align*}
    ((\sigma(B),P),[\oc_{t}]) &\csRel_{s,n} \hspace{-0.15em}((e_1,P_1),T_1) \hspace{-0.7em}&&\csRel_{s,n} \cdots \hspace{-0.7em}&&\csRel_{s,n} ((e_k,P_k),T_k) \hspace{-0.2em}&&\csRel_{s,n} ((e,Q),[\oc_{\sige}])\\
    ((\sigma(B),P),[\oc_{t'}]) &\csRel_{s,n} \hspace{-0.15em}((e'_{1},P'_{1}),T'_{1})\hspace{-1.5em} &&\csRel_{s,n} \cdots \hspace{-1.7em}&& \csRel_{s,n} ((e'_{k'},P'_{k'}),T'_{k'}) \hspace{-1.5em}&&\csRel_{s,n} ((e',Q'),[\oc_{\sige}])
 \end{align*}
 By Lemma~\ref{lemma_sequences_edges_simple}, there are at most $\left| \canRel{s-1}{\dirEdges{G}} \right |^{\left|\dcSim\right|}$ possible choices for $[e_1;\cdots;e_k]$. Let us suppose that $[e_1;\cdots;e_k]=[e'_{1};\cdots;e'_{k'}]$ and $\restrPot{s,n-1}{e,Q}=\restrPot{s,n-1}{e',Q'}$. We will prove by contradiction that $t=t'$.
 
Let us suppose that $t \neq t'$ and consider them as trees. Because $[e_1;\cdots;e_k]=[e'_{1};\cdots;e'_{k'}]$, their leftmost branches are the same. We consider the leftmost branches, $b$ and $b'$, which are different in $t$ and $t'$. Let us consider the simplifications $u$ and $u'$ of $t$ and $t'$ whose leftmost branches are $b$ and $b'$. Thus $u \complStrict t$, $u' \complStrict t'$, and the leftmost branches of $u$ and $u'$ are different.

By definition of copies $((\sigma(B),P),[\oc_u]) \csRel^* ((\_,\_),[\oc_{\sige}])$ and $((\sigma(B),P),[\oc_{u'}]) \csRel^* ((\_,\_),[\oc_{\sige}])$. We consider the first step of those paths which differs from the paths corresponding to $t$ and $t'$. Formally, we consider the lowest $i \in \mathbb{N}$ such that $((\sigma(B),P),[\oc_u]) \csRel^i ((f_i,\_),\_)$ with $f_i \neq e_i$. We are in the following case (with $v \simpl w$ and $v' \simpl w'$):
\begin{align*}
  ((\sigma(B),P),[\oc_t]) \hspace{0.3em} \csRel^{i-1} \hspace{0.2em}((\overline{\sigma_a(C)},P_i),V.\oc_{v}) \hspace{0.6em}&\noJump ((e_i,P_i.v),V) \\
  ((\sigma(B),P),[\oc_{t'}]) \csRel^{i-1} ((\overline{\sigma_a(C)},P'_i),V'.\oc_{v'}) &\noJump ((e_i,P'_i.v'),V')\\
  ((\sigma(B),P),[\oc_u]) \csRel^{i-1} \hspace{0.2em}((\overline{\sigma_a(C)},P_i),[\oc_{w}]) \hspace{0.4em}&\onlyJump ((\sigma(C),P_i),[\oc_{w}]) \\
  ((\sigma(B),P),[\oc_{u'}]) \csRel^{i-1} ((\overline{\sigma_a(C)},P'_i),[\oc_{w'}]) \hspace{0.2em}&\onlyJump ((\sigma(C),P'_i),[\oc_{w'}])
\end{align*}

We supposed $\restrPot{s,n-1}{e,Q}=\restrPot{s,n-1}{e',Q'}$. By induction on $k-j$, for $1 \leq j \leq k$, we have $\restrCont{s,n-1}{((e_j,P_j),T_j)}=\restrCont{s,n-1}{((e_j,P'_j),T'_j)}$: we use Theorem~\ref{injection_lemma_simple} for $\noJump$ steps and Lemma~\ref{injection_lemma_jump} for $\onlyJump$ steps (because $[e_1;\cdots;e_k]=[e'_1;\cdots;e'_k]$). In particular, $\restrCont{s,n-1}{((e_i,P_i.v),V)}=\restrCont{s,n-1}{((e_i,P'_i.v'),V')}$ so the signatures $\restrSig{s,n-1}{(\restrPot{s,n-1}{\sigma(C),P_i},[\oc_v])}$ and $\restrSig{s,n-1}{(\restrPot{s,n-1}{\sigma(C),P'_i},[\oc_{v'}])}$ are equal.

$u$ is a strict simplification of $t$ so $u$ is not standard. By definition of $\nestSim$, we have $B \nestSim C$ so $\stratu{\nestSim}{C}<\stratu{\nestSim}{B} \leq n$. One can verify that, for every box $D$ such that $((\sigma(C),Q_i),[\oc_w]) \csRel^* ((\sigma(D),\_),[\oc_{\_}])$ or $((\sigma(C),Q'_i),[\oc_{w'}]) \csRel^* ((\sigma(D),\_),[\oc_{\_}])$ we also have $B \nestSim D$ (so $\stratu{\nestSim}{D} \leq n-1$). Thus 

\begin{align*}
  v &= \restrSig{s,n}{((\sigma(C),P_i),[\oc_v])} &&\text{Because }t \in \copRel{s,n}{B,P}\\
    &= \restrSig{s,n-1}{((\sigma(C),P_i),[\oc_v])} &&\text{Because ``$\stratu{\nestSim}{D} \leq n-1$''}\\
    &= \restrSig{s,n-1}{((\sigma(C),P'_i),[\oc_{v'}])} &&\text{Proved in the previous paragraph.} \\
    &= \restrSig{s,n}{((\sigma(C),P'_i),[\oc_{v'}])} && \text{Because ``$\stratu{\nestSim}{D} \leq n-1$''} \\
  v &=v' &&\text{Because }t' \in \copRel{s,n}{B,P}
\end{align*}
Because $u=u'$ and the $i-1$ first steps from $((\sigma(B),P),[\oc_u])$ and $((\sigma(B),P'),[\oc_{u'}])$ take the same edges ($e_1,\cdots,e_{i-1}$) the leftmost branches of $u$ and $u'$ are equal. This is a contradiction. Our supposition was false, under our assumptions $t$ and $t'$ are equal. 

So we proved that, if we choose $[e_1;\cdots;e_k]$ and $\restrPot{s,n-1}{e,Q}$ then $t$ is uniquely determined. Thus,
\begin{align*}
    \left | \copRel{s,n}{B,P} \right | &\leq  \left | \canRel{s-1}{\dirEdges{G}} \right |^{\left|\dcSim\right|} \cdot  \left | \canRel{s,n-1}{\dirEdges{G}} \right|\\
    \left | \copRel{s,n}{B,P} \right | &\leq  \left | \canRel{s-1}{\dirEdges{G}} \right |^{\left|\dcSim\right|} \cdot  \left |\dirEdges{G}\right | \cdot  \left ( \max_{(C,Q) \in \pot{\boxset{G}}}   \left| \copRel{s,n-1}{C,Q} \right| \right )^{\depthG{G}}
\end{align*}
\end{proof}

\begin{theorem}\label{theo_bound_poly_nest} Let $x=\left| \dirEdges{G} \right|$, $S=\left|\stratSNLL\right|$, $D=\left|\dcSim\right|$, $N=\left|\nestSim \right|$, and $\partial=1+\depthG{G}$, then:
  \begin{equation*}
    \max_{(B,P) \in \pot{\boxset{G}}} \left| \cop{B,P}   \right| \leq x^{D^S \cdot \partial^{N \cdot S+N+S}-1}
  \end{equation*}
\end{theorem}
\begin{proof}
For $s,n \in \mathbb{N}$, we set $u_{0,n}=u_{s,0}=1$ and $u_{s,n}= u_{s-1,N}^{\depthG{G} \cdot D} \cdot x \cdot u_{s,n-1}^{\depthG{G}}$. Then, thanks to Lemma~\ref{lemma_bound_poly_toutca_simple}, we can prove by induction on $(s,n)$ that $\max_{(B,P) \in \pot{\boxset{G}} }\left|\copRel{s,n}{B,P} \right| \leq u_{s,n}$. One can verify by induction on $n$ that:
\begin{equation*}
  u_{s,n} = \left(u^{\depthG{G} \cdot D}_{s-1,N} \cdot x \right)^{\sum_{i=0}^{n-1} \depthG{G}^i}
\end{equation*}
Thus, for every $n \in \mathbb{N}$,
\begin{equation*}
  u_{s,n} \leq \left(u^{(\partial-1) \cdot D}_{s-1,N} \cdot x\right)^{\partial^n} = x^{\partial^n} \cdot u_{s-1,N}^{D\cdot \partial^{n} \cdot (\partial-1)}
\end{equation*} 
Thus, we prove by induction on $s$, that:
\begin{align*}
  u_{s,N} &\leq \left(x^{\partial^N} \right)^{\sum_{j=0}^{s-1} (D \cdot \partial^{N} \cdot (\partial - 1))^j} \leq \left(x^{\partial^N}\right)^{(D \cdot \partial^{N} (\partial-1))^s} \leq x^{D^s \cdot \partial^{s\cdot N + s + N}- 1}
\end{align*}
Finally, let us notice that $\csRel_{S,N}=\csRel$, so $\copRel{S,N}{B,P}=\cop{B,P}$. Thus,
\begin{align*}
  \max_{(B,P) \in \pot{\boxset{G}}} \left| \cop{B,P}   \right| &= \max_{(B,P) \in \pot{\boxset{G}} }\left|\copRel{S,N}{B,P} \right| \\
  & \leq u_{S,N} \\
\max_{(B,P) \in \pot{\boxset{G}}} \left| \cop{B,P}   \right| &\leq x^{D^S \cdot \partial^{S \cdot N + S + N}-1}
\end{align*}
\end{proof}

\begin{coro}\label{coro_bound_poly_nest}Let us consider a $\stratSNLL$-stratified, $\nestSim$-stratified, $\dcSim$-stratified proof-net $G$. Let $x=\left| \dirEdges{G} \right|$, $S=\stratG{G}$, $D=\left| \dcSim \right|$, $N=\left| \nestSim \right|$, and $\partial=1+\depthG{G}$, then:
  \begin{equation*}
    W_G \leq x^{D^{S} \cdot \partial^{(N+1) \cdot (S+1)}}
  \end{equation*}
\end{coro}
\begin{proof}
  By Theorem~\ref{theo_bound_poly_nest}, we have
  \begin{align*}
    W_G &= \left | \can{\dirEdges{G}} \right | \leq \left | \dirEdges{G} \right| \cdot \max_{(B,P) \in \pot{\boxset{G}}} \left | \cop{B,P} \right |^{\partial} \leq x \cdot \left(x^{D^S \cdot \partial^{N \cdot S+N+S}-1}\right)^{\partial}\\
    W_G &\leq x^{D^{S} \cdot \partial^{N \cdot S + N + S +1}} = x^{D^{S} \cdot \partial^{(N+1) \cdot (S+1)}}
  \end{align*}
\end{proof}

The polynomial in the bound only depends on $\stratG{G}$, $\left| \dcSim \right|$, $\left| \nestSim \right|$, and $\depthG{G}$. Those four parameters are bounded by the number of boxes. So a stratified nested proof-net controlling dependence normalizes in a time bounded by a polynomial on the size of the proof-net, the polynomial depending only on the number of boxes of the proof-net. 

In the usual encoding of binary words (or other inductive types) in linear logic, the number of boxes is independent of a term. Let us suppose that for every binary word $w$, the proof-net $(G)\underline{w}$ (representing the application of $G$ to the encoding of $w$) is $\stratSNLL$-stratified, $\dcSim$-stratified and $\nestSim$-stratified. Then $W_{(G)\underline{w}}$ is bounded by $P\left(|w|\right)$ with $P$ a polynomial which does not depend on $w$. This is the definition of polynomial soundness. However, those semantic criteria are not useful per se: 
\begin{itemize}
\item The only method we know to check the acyclicity of those relations on a proof-net $H$ is to normalize $H$ to compute the $\csRel$-paths. Normalizing a proof-net to obtain a bound on the length of its normalization has no practical use.
\item Given a proof-net $G$, we have no method to check if there exists a binary word such that one of those relation is cyclic on $(G)\underline{w}$.
\end{itemize}

In the next section we will define a decidable subsystem of linear logic (named $SDNLL$) such that $\stratSNLL$, $\dcSim$ and $\nestSim$ are acyclic on every proof-net of $SDNLL$. For $s,d,n\in \mathbb{N}$, we define a formula $\BSDN{s}{d}{n}$ such that every binary word can be encoded by a proof-net typed by $\BSDN{s}{d}{n}$. Determining if a given proof-net $G$ can be typed by a formula of the shape $\BSDN{s}{d}{n} \multimap A$ is decidable. And, if this is the case $(G)\underline{w}$ is $\stratSNLL$-stratified, $\dcSim$-stratified and $\nestSim$-stratified for every binary word $w$.

\section{Linear logic subsystems and $\lambda$-calculus type-systems}\label{chapter_type_systems}  
\label{section_def_sdnll}
\subsection{Definition of $SDNLL$}\label{section_def_sdnll}
We define a linear logic subsystem, called $SDNLL$ (for {\em S}tratification {\em D}ependence control {\em N}esting Linear Logic) characterizing polynomial time. In $SDNLL$, to enforce $\stratSNLL$-stratification, $\dcSim$-stratification and $\nestSim$-stratification, we label the $\oc$ and $\wn$ modalities with integers $s$, $d$ and $n$. Let us consider a $SDNLL$ proof-net $G$ and boxes $B$ and $B'$ with $\beta(\sigma(B))=\oc_{s,n,d}A$ and $\beta(\sigma_G(B'))=\oc_{s',d',n'}A'$. Then we will have the following implications: ($B \stratSNLL B'$ implies $s > s'$), ($B \dcSim B'$ implies $d>d'$) and ($B \nestSim B'$ implies $n > n'$). This implies that $G$ is $\stratSNLL$-stratified, $\dcSim$-stratified and $\nestSim$-stratified.

\begin{definition}\label{def_sdnll} For $s \in \mathbb{N}$, we define $\form{s}$ by the following grammar (with $t,d,n \in \mathbb{N}$, $t \geq s$ and $X$ ranges over a countable set of variables). Notice that $\form{0} \supseteq \form{1} \supseteq \cdots$
  \begin{equation*}
    \form{s} := X_{t} \mid X^\perp_t \mid \form{s} \otimes \form{s} \mid \form{s} \parr \form{s} \mid \forall X_{t}. \form{s} \mid \exists X_{t}.\form{s} \mid \oc_{t,d,n} \form{t+1} \mid \wn_{t,d,n} \form{t+1}
  \end{equation*}
\end{definition}

In this section, a {\em formula context} is a formula where a subterm has been replaced by $\circ$ (e.g. $\oc_{2,1,3}X \parr \circ$). If $h$ is a formula context and $A$ is a formula, then $h[A]$ refers to the formula obtained by replacing $\circ$ by $A$ in $h$. We gave another definition for ``context'' in Section~\ref{section_def_cont_sem}, and we will define yet another in Section~\ref{section_sdnll_lambda}. Because those terms are well-established terms, we chose not not create new terms. Because these definitions are very different, there should be little confusion.

For any formula of the shape $A=\oc_{s',d',n'}A'$, we write $\straF{A}$ for $s'$, $\dcF{A}$ for $d'$ and $\nesF{A}$ for $n'$. For $A \in \form{0}$, $s_A^{min}$ refers to the minimum $s \in \mathbb{N}$ such that $A=h[\oc_{s,\_,\_}\_]$ with $h$ a formula context. To gain expressivity, we define a subtyping relation $\leq$ on $\form{0}$. The relation $\leq$, defined as the transitive closure of the following $\leq^1$ relation, follows the intuition that a connective $\oc_{s,d,n}$ in a formula means that this connective ``comes'' from a box $B$ with $\stratu{\stratSNLL}{B} \geq s$, $\stratu{\dcSim}{B} \geq d$ and $\stratu{\nestSim}{B} \geq n$.

\begin{equation*}
  A \leq^1 B \Leftrightarrow \left \{ 
  \begin{array}{ll}
    \text{Either }& A = g[\oc_{s,d,n}D],  B = g[\oc_{s',d',n'}D],  s \geq s', d \geq d' \text{ and } n \geq n'\\
    \text{Or }    & A = g[\wn_{s,d,n}D],  B = g[\wn_{s',d',n'}D],  s \leq s', d \leq d' \text{ and } n \leq n'
  \end{array}
  \right .
\end{equation*}    

\begin{lemma}\label{subtyping_perp}
  If $A \leq B$ then $A^\perp \geq B^\perp$
\end{lemma}
\begin{proof}
  Immediate from the definition of $\leq$.
\end{proof}

\begin{definition}\label{def_sdnll_proofnet}  A {\em $SDNLL$ proof-net} is a proof-net whose edges are labelled by a $\form{0}$ formula, the labels respecting the constraints of Figure~\ref{rules_labelling_snll} modulo subtyping.

More precisely, the labelling of a proof-net $G$ is correct if for every node/box whose premises are labelled by $A_1,\cdots,A_k$ and whose conclusions are labelled by $C_1,\cdots,C_l$, there exists an instance of the constraint of Figure~\ref{rules_labelling_snll} whose premises are labelled by $A_1,\cdots,A_k$ and conclusions are labelled by $B_1,\cdots,B_l$ with $B_1 \leq C_1$,...$B_l \leq C_l$.
\end{definition}

For instance, let us suppose that $d_1 \geq d$ and $d_2,\cdots,d_k \geq d+1$ then $\wn_{s_1,d,n}A_1 \leq \wn_{s_1,d_1,n}A_1$, $\wn_{s_2,d+1,n}A_2 \leq \wn_{s_2,d_2,n}A_2$, ..., $\wn_{s_k,d+1,n}A_k \leq \wn_{s_k,d_k,n}A_k$. So, a box with premises $A_1,\cdots,A_k,C$ and conclusions $\wn_{s_1,d_1,n}A_1$,...$\wn_{s_k,d_k,n}A_k,\oc_{s,d,n}C$ satisfies the conditions of $SDNLL$ proof-nets.

The $SDNLL$ labels are compatible with cut-elimination as shown by the rules of Figure~\ref{cut_sdnll}. For most rule, the only difficulty is to handle subtyping. We explain it for the $\otimes / \parr$ case (the $ax$, $\forall / \exists$, $\oc P / \wn D$, $\oc P / \wn W$ and $\oc P / \wn C$ rules are adapted in the same way): by definition of $SDNLL$ proof-nets, $C \geq A \parr B$ so $C$ is of the shape $A_1 \parr B_1$ with $A \leq A_1$ and $B \leq B_1$. So $C^\perp = A_1^\perp \otimes B_1^\perp$ and, by definition of $SDNLL$ proof-nets, $A_2^\perp \leq A_1^\perp$ and $B_2^\perp \leq B_1^\perp$. By Lemma~\ref{subtyping_perp}, $A_1^\perp \leq A^\perp$ and $B_1^\perp \leq B^\perp$. So $A_2^\perp \leq A^\perp$ and $B_2^\perp \leq B^\perp$. So, whatever are the nodes with conclusions $A_2^\perp$ and $B_2^\perp$, we can replace their conclusions with $A^\perp$ and $B^\perp$ without breaking $SDNLL$ constraints. One can observe that we could have labelled the edges with $A_1,B_1,A^\perp_1,B_1^\perp$, or with $A_2,B_2,A^\perp_2,B_2^\perp$ or other formulae between $A$ and $A_2$ and between $B$ and $B_2$. We decide not to choose a canonical reduction: this only influences the indices on exponential connectives, and the conclusions of the proof-net are not concerned.

For the sake of readability, in the reductions of Figure~\ref{cut_sdnll} we suppose that subtyping only modifies the outermost exponential connectives (modification of labels on inner connectives are dealt as in the $\otimes / \parr$ case). Every letter in the indices represents a positive number, writing $d+d_1$ as the second index of an edge allows to represent any number greater than (or equal to) $d$.
\begin{itemize}
\item For the $\oc P / \wn P$ rule, we can notice that $d=d'+d'_1$ so $d \geq d'$, and $n=n'+n'_1$ so $n \geq n'$. Thus we have $d+d_1 \geq d'$, $d+d_k+1 > d'$, $n+n_1 \geq n'$ and $n+n_k \geq n'$. The box in the reduct satisfies the constraints of $SDNLL$.
\item For the $\oc P / \wn N$ rule, according to the definition of $SDNLL$ proof-nets, $d \geq d'$, and $n \geq n'$. So $n+n_1 \geq n'$, $n+n_k \geq n'$ and the outermost box of the reduct satisfies the constraints of $SDNLL$. We also have $n+n_1+1 \geq n'+1$ and $n+n_k+1 \geq n'+1$, so the innermost box of the reduct satisfies the constraints of $SDNLL$.
\end{itemize}

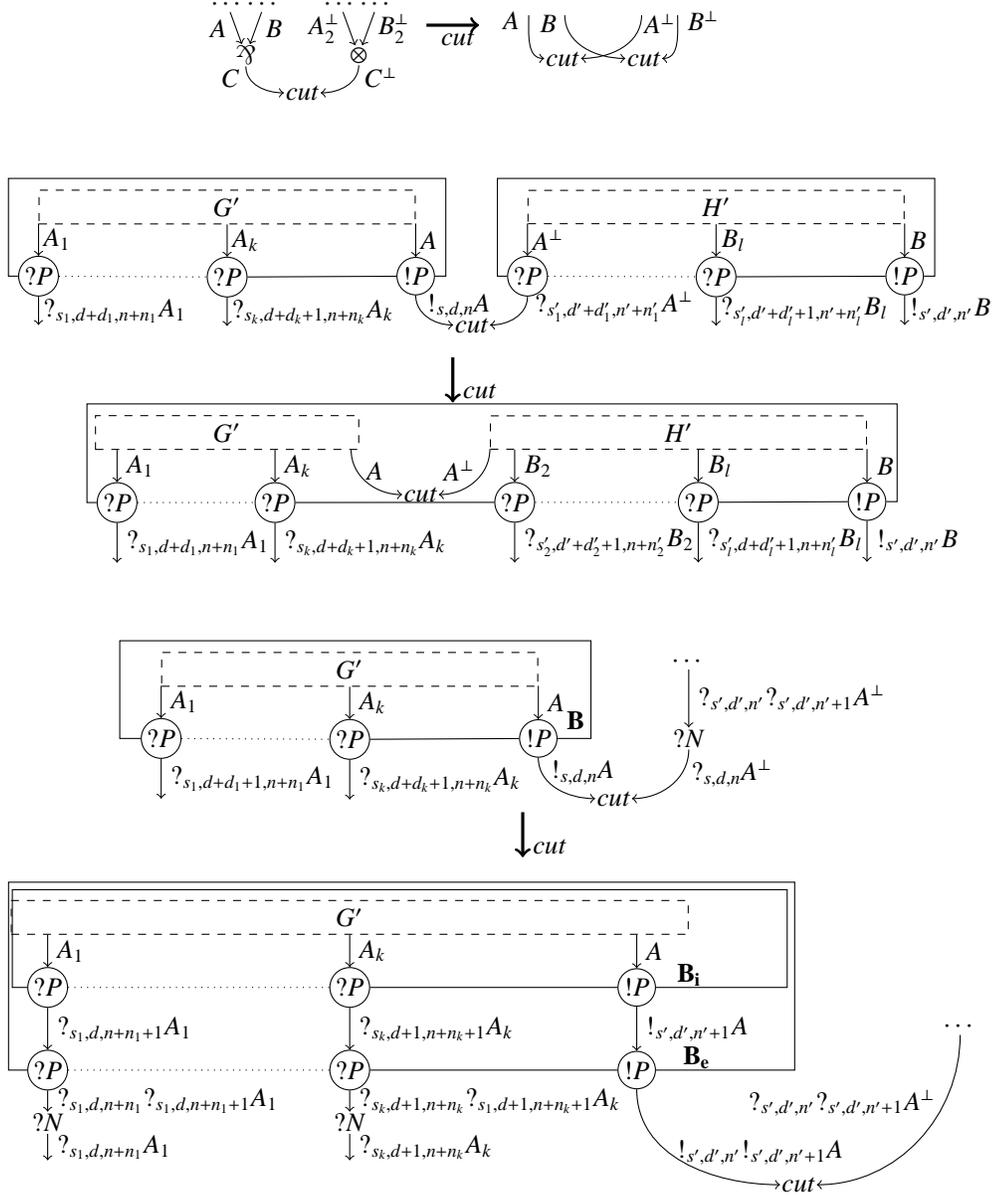
\begin{figure}\centering
  \tikzstyle{edgename}=[phantom]
  \begin{subfigure}{\textwidth}
    \centering
    \begin{tikzpicture}[baseline=0cm]
  \tikzstyle{edgename}=[opacity=0]
  \begin{scope}[shift={(3.2,-1.6)}]
    \node [par]  (par)  at (0,0) {};
    \node [tensor] (tens) at ($(par)+(1.5,0)$) {};
    \node [cut]  (cut)  at ($(par)!0.5!(tens)+(0,-0.5)$) {};
    \node [etc] (pl) at ($(par) +(110:0.7)$) {};
    \node [etc] (pr) at ($(par) +( 70:0.7)$) {};
    \node [etc] (tl) at ($(tens)+(110:0.7)$) {};
    \node [etc] (tr) at ($(tens)+( 70:0.7)$) {};
    \draw [ar] (pl) -- (par) node [type,left] {$A$} node [edgename] {$a$};
    \draw [ar] (pr) -- (par) node [type]      {$B$} node [edgename,right] {$b$};
    \draw [ar] (tl) -- (tens)node [type,left] {$A_2^\perp$} node [edgename] {$e$};
    \draw [ar] (tr) -- (tens)node [type]      {$B_2^\perp$} node [edgename,right] {$f$};
    \draw [ar,out=-90,in=180] (par)  to node [edgename,below left] {$c$} node [type,left,pos=0.25] {$C$} (cut);
    \draw [ar,out=-90,in=  0] (tens) to node [edgename,below right]{$d$} node [type,pos=0.25] {$C^\perp$}(cut);
  
    \draw [reduc] (2.4,0.4) --++(0.7,0) node [below left=-0.1cm] {$cut$};
    
    \node (pl2) at ($(pl)+(4, 0)$) {}; 
    \node (pr2) at ($(pr)+(4,0)$) {};
    \node (tl2) at ($(tl)+(4,0)$) {};
    \node (tr2) at ($(tr)+(4,0)$) {};
    \node [cut] (cutl) at ($(pl2)!0.5!(tl2)+(-0.3,-0.7)$) {};
    \node [cut] (cutr) at ($(pr2)!0.5!(tr2)+( 0.3,-0.7)$) {};
    \draw [ar,out=-90,in=180] (pl2) to node [type,pos=0.1,left] {$A$}      node [edgename,pos=0.15] {$a$} (cutl);
    \draw [ar,out=-90,in=  0] (tl2) to node [type,right=-0.08cm, pos=0.1]      {$A^\perp$} node [edgename,pos=0.1] {$e$} (cutl);
    \draw [ar,out=-90,in=180] (pr2) to node [type,pos=0.1,left] {$B$}      node [edgename,pos=0.1] {$b$} (cutr);
    \draw [ar,out=-90,in=  0] (tr2) to node [type,pos=0.1]      {$B^\perp$} node [edgename,pos=0.15] {$f$} (cutr);
  \end{scope}
\end{tikzpicture}
  \end{subfigure}
  
  \vspace{2em}

  \begin{subfigure}{\textwidth}
    \centering
     \begin{tikzpicture}[baseline=0.2cm]
  \tikzstyle{edgename}=[opacity=0]
  \begin{scope}[scale=1.0]
  \node [proofnet,minimum width=5cm] (G) at (0,0) {$G'$};
  \node [princdoor] (prg) at ($(G.-5)+(0,-0.7)$) {};
  \node [auxdoor]   (a2) at ($(G.-90)+(0,-0.7)$) {};
  \node [auxdoor]   (a1) at ($(G.-175)+(0,-0.7)$) {};
  \draw (a2)--(prg) -| ++(0.4,1.3) -| ($(a1)+(-0.4,0)$) -- (a1);
  \draw [dotted] (a1)--(a2);
  \draw [ar] (G.-5)--(prg)   node [type,right=-0.06cm] {$A$}       node [edgename] {$a$};     
  \draw [ar] (G.-90)--(a2)    node [type,right=-0.06cm] {$A_k$}     node [edgename] {$a_k$};
  \draw [ar] (a2)--++(0,-0.65) node [type,right=-0.06cm] {$\wn_{s_k,d+d_k+1,n+n_k} A_k$} node [edgename] {$c_k$} ;
  \draw [ar] (G.-175)--(a1)   node [type,right=-0.06cm] {$A_1$}     node [edgename] {$a_1$};    
  \draw [ar] (a1)--++(0,-0.65) node [type,right=-0.06cm] {$\wn_{s_1,d+d_1,n+n_1} A_1$} node [edgename] {$c_1$};

  \node [proofnet,minimum width=5cm] (H) at ($(G)+(6.5,0)$) {$H'$};
  \node [princdoor] (prh) at ($(H.-5)+(0,-0.7)$) {};
  \node [auxdoor]   (b2) at ($(H.-90)+(0,-0.7)$) {};
  \node [auxdoor]   (b1) at ($(H.-175)+(0,-0.7)$) {};
  \draw (b2)--(prh) -| ++(0.4,1.3) -| ($(b1)+(-0.4,0)$) -- (b1);
  \draw [dotted] (b1)--(b2);
  \draw [ar] (H.-5)--(prh)   node [type,right=-0.06cm] {$B$}       node [edgename] {$b$};     
  \draw [ar] (H.-90) --(b2)   node [type,right=-0.06cm] {$B_l$}     node [edgename] {$b_l$};    
  \draw [ar] (b2)--++(0,-0.65) node [type,right=-0.06cm] {$\wn_{s'_l,d'+d'_l+1,n'+n'_l} B_l$} node [edgename] {$d_l$};
  \draw [ar] (H.-175)--(b1)   node [type,right=-0.06cm] {$A^\perp$} node [edgename] {$b_1$};  
  \draw [ar] (prh)--++(0,-0.65) node [type,right=-0.06cm] {$\oc_{s',d',n'} B$}  node [edgename] {$d$};

  \node [cut] (cut) at ($(prg)!0.5!(b1)+(0,-0.65)$) {};
  \draw [ar,out=-90,in=180] (prg) to node [edgename, below left] {$c$}    node [type,pos=0.2] {$\oc_{s,d,n} A$}      (cut);
  \draw [ar,out=-90,in=  0] (b1)  to node [edgename, below right=-0.2cm] {$d_1$} node [type,pos=0.2] {$\wn_{s'_1,d'+d'_1,n'+n'_1} A^\perp$} (cut);

  \draw [reduc] (3,-2)--++(0,-0.6) node [near end, right] {$cut$};

  \node [proofnet,minimum width=3.5cm] (G) at ($(G)+(0,-3)$) {$G'$};
  \node [auxdoor]   (a2) at ($(G.-20)+(0,-0.7)$) {};
  \node [auxdoor]   (a1) at ($(G.-171)+(0,-0.7)$) {};
  \node [proofnet,minimum width=5cm] (H) at ($(G)+(6,0)$) {$H'$};
  \node [princdoor] (prh) at ($(H.-5)+(0,-0.7)$) {};
  \node [auxdoor]   (b2) at ($(H.-41)+(0,-0.7)$) {};
  \node [auxdoor]   (b1) at ($(H.-174)+(0,-0.7)$) {};
  \node [cut]       (cut)at ($(G.-8)!0.5!(H.-175)+(0,-0.6)$) {};
  \draw [ar,out=-90,in=180] (G.-8)  to node [edgename, below left =-0.1]  {$a$}   node [type,pos=0.3] {$A$}     (cut);
  \draw [ar,out=-90,in=  0] (H.-175) to node [edgename, below right=-0.15] {$b_1$} node [type,left, pos=0.3] {$A^\perp$}(cut);
  \draw (b2)--(prh) -| ++(0.4,1.3) -| ($(a1)+(-0.4,0)$) -- (a1);
  \draw (a2)--(b1);
  \draw [dotted] (a1)--(a2);
  \draw [dotted] (b1)--(b2);
  \draw [ar] (G.-20)--(a2)     node [type] {$A_k$}     node [edgename] {$a_k$};
  \draw [ar] (a2)--++(0,-0.8)  node [type] {$\wn_{s_k,d+d_k+1,n+n_k} A_k$} node [edgename] {$c_k$};
  \draw [ar] (G.-171)--(a1)    node [type] {$A_1$}     node [edgename] {$a_1$};
  \draw [ar] (a1)--++(0,-0.8)  node [type] {$\wn_{s_1,d+d_1,n+n_1} A_1$} node [edgename] {$c_1$};
  \draw [ar] (H.-41)--(b2)     node [type] {$B_l$}     node [edgename] {$b_l$};
  \draw [ar] (b2)--++(0,-0.8)  node [type] {$\wn_{s'_l,d+d'_l+1,n+n'_l} B_l$} node [edgename] {$d_l$};
  \draw [ar] (H.-174)--(b1)    node [type] {$B_2$}     node [edgename,left=-0.1cm] {$b_2$};
  \draw [ar] (b1)--++(0,-0.8)  node [type] {$\wn_{s'_2,d'+d'_2+1,n+n'_2} B_2$} node [edgename] {$d_2$};
  \draw [ar] (H.-5)--(prh)    node [type] {$B$}       node [edgename] {$b$};     
  \draw [ar] (prh)--++(0,-0.8) node [type] {$\oc_{s',d',n'} B$}   node [edgename] {$d$}; 
  \end{scope}
\end{tikzpicture}
  \end{subfigure}
  \vspace{2em}
  
  \begin{subfigure}{\textwidth}
    \centering
     \begin{tikzpicture}[baseline=0.2cm]
  \tikzstyle{edgename}=[opacity=0]
  \node [proofnet,minimum width=5cm] (G) at (0,0) {$G'$};
  \node [princdoor] (pr) at ($(G.-5)+(0,-0.7)$) {};
  \node [above] at ($(pr)+(0.5,0)$) {$\mathbf{B}$};
  \node [auxdoor]   (a2) at ($(G.-90)+(0,-0.7)$) {};
  \node [auxdoor]   (a1) at ($(G.-175)+(0,-0.7)$) {};
  \draw (a2)--(pr) -| ++(0.7,1.3) -| ($(a1)+(-0.55,0)$) -- (a1);
  \draw [dotted] (a1)--(a2);
  \draw [ar] (G.-5)--(pr)    node [edgename] {$a$}   node [type] {$A$};     
  \draw [ar] (G.-90)--(a2)    node [edgename] {$a_k$} node [type] {$A_k$};    
  \draw [ar] (G.-175)--(a1)   node [edgename] {$a_1$} node [type] {$A_1$};    
  \draw [ar] (a2)--++(0,-0.8) node [edgename] {$c_k$} node [type] {$\wn_{s_k,d+d_k+1,n+n_k} A_k$};
  \draw [ar] (a1)--++(0,-0.8) node [edgename] {$c_1$} node [type] {$\wn_{s_1,d+d_1+1,n+n_1} A_1$};
  \node [dig]  (dig) at ($(pr)+(2,0)$) {};
  \node [cut]   (cut)at ($(pr)!0.5!(dig)+(0,-0.8)$) {};
  \draw [ar,out=-90,in=180] (pr) to node [edgename,below left]  {$c$} node [type,pos=0.2] {$\oc_{s,d,n} A$} (cut);
  \draw [ar,out=-90,in=  0] (dig) to node [edgename,right] {$d$} node [type,pos=0.3] {$\wn_{s,d,n} A^\perp$} (cut);
  \node [etc] (etc) at ($(dig)+(0,1)$) {};
  \draw [ar] (etc) -- (dig) node [edgename,right] {$f$} node [type] {$\wn_{s',d',n'} \wn_{s',d',n'+1} A^\perp$};

  \draw [->,very thick] (2.3,-1.9) --++(0,-0.6) node [near end, right] {$cut$};

  \node [proofnet,minimum width=9cm] (G1) at ($(G)+(0,-3.3)$) {$G'$};
  \node [princdoor] (pri) at ($(G1.-4)+(0.5,-0.7)$) {};
  \node [above=-0.1] at ($(pri)+(0.7,0)$) {$\mathbf{B_i}$};
  \node [auxdoor]   (a2i) at ($(G1.-90)+(0,-0.7)$) {};
  \node [auxdoor]   (a1i) at ($(G1.-176)+(-0.7,-0.7)$) {};
  \draw (a2i)--(pri) -| ++(2,1.3) -| ($(a1i)+(-0.47,0)$) -- (a1i);
  \draw [dotted] (a1i)--(a2i);
  \draw [ar] ($(G1.-4)+(0.5,0)$)--(pri) node [edgename] {$a$}   node [type] {$A$};     
  \draw [ar] (G1.-90)--(a2i) node [edgename] {$a_k$} node [type] {$A_k$};
  \draw [ar] ($(G1.-176)+(-0.7,0)$)--(a1i)node [edgename] {$a_1$} node [type] {$A_1$};
  \node [princdoor] (pre) at ($(pri)+(0,-1.1)$) {};
  \node [above=-0.1] at ($(pre)+(0.8,0)$) {$\mathbf{B_e}$};
  \node [auxdoor]   (a2e) at ($(a2i)+(0,-1.1)$) {};
  \node [auxdoor]   (a1e) at ($(a1i)+(0,-1.1)$) {};
  \draw (a2e)--(pre) -| ++(2.1,2.5) -| ($(a1e)+(-0.52,0)$) -- (a1e);
  \draw [dotted] (a1e)--(a2e);
  \draw [ar] (pri)--(pre) node [edgename] {$b$}   node [type] {$\oc_{s',d',n'+1} A$};
  \draw [ar] (a2i)--(a2e) node [edgename] {$b_k$} node [type] {$\wn_{s_k,d+1,n+n_k+1} A_k$};
  \draw [ar] (a1i)--(a1e) node [edgename] {$b_1$} node [type] {$\wn_{s_1,d,n+n_1+1} A_1$};
  \node [dig] (dig1) at ($(a1e)+(0,-0.7)$) {};
  \node [dig] (dig2) at ($(a2e)+(0,-0.7)$) {};
  \draw [ar] (a1e)--(dig1) node [edgename] {$e_1$} node [type] {$\wn_{s_1,d,n+n_1} \wn_{s_1,d,n+n_1+1}A_1$};
  \draw [ar] (a2e)--(dig2) node [edgename] {$e_k$} node [type] {$\wn_{s_k,d+1,n+n_k} \wn_{s_1,d+1,n+n_k+1} A_k$};
  \draw [ar] (dig1)--++(0,-0.5) node [edgename] {$c_1$} node [type] {$\wn_{s_1,d,n+n_1} A_1$};
  \draw [ar] (dig2)--++(0,-0.5) node [edgename] {$c_k$} node [type] {$\wn_{s_k,d+1,n+n_k} A_k$};
  \node [etc]  (etc) at ($(pri)!0.5!(pre)+(4.3,0)$) {};
  \node [cut]  (cut) at ($(pre)!0.5!(etc)+(0,-1.8)$) {};
  \draw [ar,out=-90,in=180] (pre) to node [edgename,below left]  {$c$} node [type,pos=0.4] {$\oc_{s',d',n'} \oc_{s',d',n'+1} A$} (cut);
  \draw [ar,out=-90,in=  0] (etc) to node [edgename,right] {$f$} node [type,left, pos=0.3] {$\wn_{s',d',n'} \wn_{s',d',n'+1} A^\perp$} (cut) ;
\end{tikzpicture}
  \end{subfigure}
    
  \caption{\label{cut_sdnll}$SDNLL$ constraints are compatible with cut-elimination.}
\end{figure}
\tikzstyle{edgename}=[opacity=1, midway, left, black]

In order to prove the soundness of $SDNLL$ for polynomial time, we first have to prove a technichal lemma. Whenever $((e,P),[\oc_t]@T) \noJump^* ((f,Q),[\oc_u]@U)$, the formulae of $e$ and $f$ are related. To be more precise, if $G$ does not contain any $\exLab$ or $\faLab$ node, $\beta(e)_{|T} \leq \beta(f)_{|U}$ with $A_{|T}$ defined as follows:
\begin{definition}{\label{def_restr_trace}}
  Let $A$ be a formula and $T$ a trace, we define $A_{|T}$ by induction on $A$ as follows: $A_{|[]}=A$, $(A \otimes B)_{|T.\otimes_l}=(A \parr B)_{|T.\parr_l}=A_{|T}$, $(A \otimes B)_{|T.\otimes_r}=(A \parr B)_{|T.\parr_r}=B_{|T}$, $(\forall X.A)_{|T.\forall}=(\exists X.A)_{|T.\exists}=(\oc_{s,d,n}A)_{|T.\oc_t}=(\wn_{s,d,n}A)_{|T.\wn_t}=A_{|T}$.
\end{definition}

For instance, if $((e,P),T) \noJump ((f,P),T.\parr_r)$ then $\beta(e)$ and $\beta(f)$ are of the shape $B$ and $A' \parr B'$ with $B \leq B'$. Let us notice that $\beta(f)_{|T.\parr_r}=(B')_{|T} \geq (B)_{|T} = \beta(e)_{|T}$.

However, there is a problem with this definition when we cross a $\exists$ link downwards. For example, let us suppose that $((c,[~]),[\parr_r;\oc_{\sige};\otimes_r]) \mapsto ((d,[~]),[\parr_r;\oc_{\sige};\otimes_r;\exists])$ with $\beta(c)= \wn(X \otimes X^\perp) \otimes \oc(X^\perp \parr X)$ and $\beta(d)=\exists Y.Y \otimes Y^\perp$. Then $\beta(c)_{|[\parr_r;\oc_{\sige};\otimes_r]}=X$, but $\beta(d)_{|[\parr_r;\oc_{\sige};\otimes_r;\exists]}$ is undefined: the trace is not compatible with the syntactic tree of $\beta(d)$. In~\cite{perrinelMegathese}, we define a mapping $\realbeta{\_}$ from contexts to formulae (paying special attention to the substitutions caused by the $\forall/\exists$ cut-elimination) satisfying Lemmas~\ref{lemma_realbeta_simple} to~\ref{lemma_realbeta_csrel}:
\begin{itemize}
\item The first idea is to substitute, for some of the $\exists X.\_$ in $\beta(e)$, the occurrences of $X$ by its formula $B$: if $\beta(e)_{|T}=\exists X.A$ and $((\concl{l},\_),[\exists]) \noJump^* ((e,P),T.\exists @U)$ with $l$ a $\exists$ node whose associated formula is $B$, we replace $\exists X.A$ by $A[B/X]$. 
\item Moreover, if $\beta(e)$ contains a free occurrence of a variable $X$ associated with the $\forall$ node $m$, and $((\concl{m},P),[\forall]) \noJump^* ((\overline{\concl{l}},\_),[\forall])$ with $l$ a $\exists$ node whose associated formula is $B$, we replace $X$ by $B$.
\end{itemize}
Those two operations can be recursive: the formula $B$ can contain itself free occurrences of variables associated with $\forall$ nodes, or $\exists Y.C$ subformulas. The formal definition can be found in section 5.1.2 of~\cite{perrinelMegathese}.

\begin{lemma}{\label{lemma_realbeta_simple}}
  If $\beta(e)$ is of the shape $\oc_{s,d,n}A$, then there exists a substitution $\theta$ such that $\realbeta{(e,P),[\oc_t]@T}=\left(\oc_{s,d,n}A[\theta]\right)_{|T}$.
\end{lemma}

\begin{lemma}{\label{lemma_realbeta_nojump}}
  If $C \noJump^* C'$, we have $\realbeta{C} \leq \realbeta{C'}$.
\end{lemma}

The $\onlyJump$ step breaks this property: if $\beta(\sigma_i(B))=\wn_{s,d,n}A$ and $\beta(\sigma(B))=\oc_{s',d',n'}A'$ then $A$ and $A'$ are a priori unrelated. The only relation required on those formulae is that $d\geq d'$ and $n \geq n'$. Thus, whenever $((\sigma(B),P),[\oc_t]) \csRel^* ((\sigma(B'),P'),[\oc_{t'}])$ with $\beta(\sigma(B))=\oc_{s,d,n}A$ and $\beta(\sigma(B'))=\oc_{s',d',n'}A'$, we have $d \geq d'$ and $n \geq n'$:

\begin{lemma}{\label{lemma_realbeta_csrel}}
  If $C \csRel^* C'$ and $\realbeta{C}=\oc_{\_,d,n}A$ then $\realbeta{C'}=\oc_{\_,d',n'}B$ with $d \geq d'$ and $n \geq n'$.
\end{lemma}

\begin{figure}
  \centering
  \begin{tikzpicture}
    \tikzstyle{level}=[opacity = 0]
    \begin{scope}[scale = 0.85]
      \begin{scope}[shift={(0,1)}]
        \node [ax] (ax) at (0,0) {};
        \draw [ar,out= -10,in=90] (ax.0)   to node [type]       {$A^\perp$} ( 0.5,-0.8);
        \draw [ar,out=-170,in=90] (ax.180) to node [type, left] {$A$}      (-0.5,-0.8);
      \end{scope}

      \begin{scope}[shift={(2,0)}]
        \coordinate (G)   at (0,1);
        \coordinate (H)   at ($(G)+(1,0)$);
        \node [cut]      (cut) at ($(G)!0.5!(H)+(0,-0.8)$) {};
        \draw [ar,out=-90,in=180] (G. -40) to node [type, left] {$A$} (cut.180);
        \draw [ar,out=-90,in=  0]  (H.-140) to node [type] {$A^\perp$} (cut.0);
      \end{scope}

      \begin{scope}[shift={(4.4,0.6)}]
        \node [exists]   (ex) at (0,0) {};
        \draw [ar] ($(ex)+(0,0.6)$) to  node [type] {$A[B/X_s]$} (ex);
        \draw [ar] (ex) --++ (0,-0.6) node [type] {$\exists X_s. A$};
      \end{scope}
      
      \begin{scope}[shift={(6.8,0.6)}]
        \node [forall]   (fa) at (0,0) {};
        \draw [ar] ($(fa)+(0,0.6)$) to  node [type, opacity=1] {$A$} (fa);
        \draw [ar] (fa) --++ (0,-0.6) node [type, opacity=1] {$\forall X_s. A$};
      \end{scope}
    
        \begin{scope}[shift={(9.45,0.7)}]
          \node [tensor] (tens) at (0,0) {};
          \draw [ar,bend right] (-0.5,0.4) to node [type,left] {$A$} (tens);
          \draw [ar,bend left] (0.5,0.4) to node [type] {$B$} (tens);
          \draw [ar] (tens) --++ (0,-0.5) node [type] {$A \otimes B$};
        \end{scope}

        \begin{scope}[shift={(11.85,0.7)}]
          \node [par]  (par) at (0,0)  {};
          \draw [ar,bend right] (-0.5,0.4) to node [type,left] {$A$} (par);
          \draw [ar,bend left]  ( 0.5,0.4)  to node [type] {$B$} (par);
          \draw [ar] (par) --++ (0,-0.5) node [type] {$A \parr B$};
        \end{scope}

      \begin{scope}[shift={(-4,-2.5)}]
        
        \begin{scope}[shift={(4,1)}]
          \node[weak](weak) at (0,0) {};
          \draw[ar]  (weak) --++(0,-0.6)  node [type] {$?_{s,d,n}A$} node [level, midway, left] {$i$};
        \end{scope}
        
        \begin{scope}[shift={(7.3,1)}]
          \node[der] (der) at (0,0) {};
          \draw[ar] ($(der)+(0,0.6)$) -- (der) node [type] {$A$} node [level, left] {$i$};
          \draw[ar] (der) --++ (0,-0.6) node [type] {$?_{s,d,n}A$} node [level, left] {$i$-1};
        \end{scope}

        \begin{scope}[shift={(14,1)}]
          \node [dig] (dig) at (0,0)   {};
          \draw[ar] ($(dig)+(0,0.6)$) -- (dig) node [type] {$\wn_{s,d,n}\wn_{s,d,n+1}A$} node [level, left] {$i$};
          \draw[ar] (dig) --++ (0,-0.6) node [type] {$?_{s,d,n+1}A$} node [level, left] {$i$-1};
        \end{scope}

      \begin{scope}[shift={(11,1)}]  
        \node[cont] (cont) at (0,0) {};
        \draw[ar] ($(cont)+(120:0.6)$) -- (cont) node [type,left] {$\wn_{s,d,n}A$};
        \draw[ar] ($(cont)+( 60:0.6)$) -- (cont) node [type] {$\wn_{s,d,n}A$};
        \draw[ar] (cont) --++ (0,-0.6) node [type] {$\wn_{s,d,n}A$};
      \end{scope}

      \end{scope}

   
      \begin{scope}[shift={(6,-4.2)}]
        \draw[opacity=0] (0,1) node [pn,inner xsep=3.1cm] (G) {$G$};
        \draw (G)++(3.5,-1) node [princdoor] (bang)   {};
        \draw (G)++(1.2,-1) node [auxdoor] (whyn) {};
        \draw (G)++(-1.5,-1) node [auxdoor] (whyn2) {};
        \draw (G)++(-3.5,-1) node [auxdoor] (whyn3) {};
        \draw[ar] (whyn3 |- G.south) -- (whyn3) node [type] {$A_1$} node [level] {$i$};
        \draw[ar] (whyn2 |- G.south) -- (whyn2) node [type] {$A_2$} node [level] {$i$};
        \draw[ar] (whyn |- G.south) -- (whyn) node [type] {$A_k$} node [level] {$i$};
        \draw[ar] (bang |- G.south) -- (bang) node [type] {$C$} node [level] {$i$};
        \draw[ar] (whyn3) --++(0,-0.8) node [type,right=-0.05cm] {$\wn_{s_1,d,n} A_1$} node [level] {$i$-1};
        \draw[ar] (whyn2) --++(0,-0.8) node [type,right=-0.05cm] {$\wn_{s_2,d+1,n} A_2$} node [level] {$i$-1};
        \draw[ar] (whyn) --++(0,-0.8) node [type,right=-0.05cm]  {$\wn_{s_k,d+1,n} A_k$} node [level] {$i$-1};
        \draw[ar] (bang) --++(0,-0.8) node [type,right=-0.05cm] {$\oc_{s,d,n} C$} node [level] {$i$-1};
        \draw (whyn)--(bang) -| ++(0.7,1.1) -| ($(whyn3)+(-0.4,0)$) -- (whyn3) -- (whyn2);
        \draw [dotted] (whyn2) -- (whyn);
      \end{scope}
    \end{scope}
  \end{tikzpicture}

  \caption{ \label{rules_labelling_snll}Constraints for $SDNLL$ proof-nets. For the $\exists$ rule, we require $s_B^{min} \geq s$. For the promotion rule, one of the auxiliary door has the same second index as the principal door (in the figure we set arbitrarily this door to be the first one).}
\end{figure}
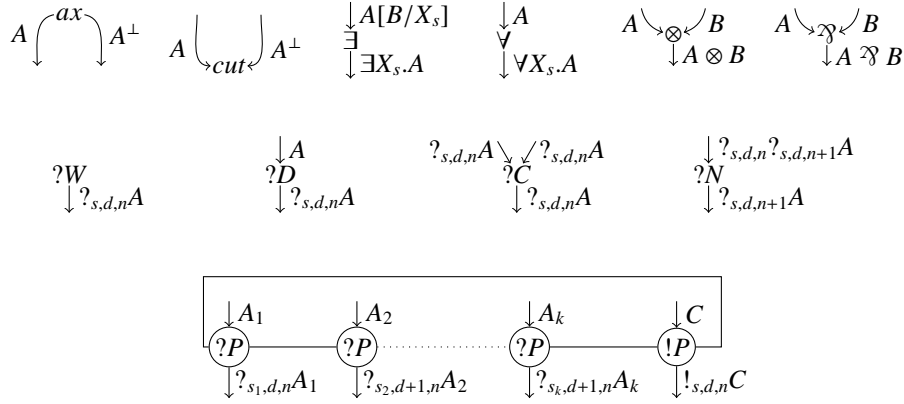

\subsection{$SDNLL$ is sound for $Poly$}
Thanks to $\realbeta{\_}$, we can prove the implications stated in the beginning of Section~\ref{section_def_sdnll}.

\begin{lemma}\label{lemma_snll_strat}
  If $B \stratSNLL B'$, $\beta(\sigma(B))= \oc_{s,d,n}A$ and $\beta(\sigma(B'))= \oc_{s',d',n'}A'$ then $s > s'$.
\end{lemma}
\begin{proof}
  Let us suppose that $B \stratSNLL B'$, $\beta(\sigma(B))=\oc_{s,d,n}A$ and $\beta(\sigma(B'))=\oc_{s',d',n'}A'$. By definition of $\stratSNLL$, there exist potentials $P$ and $P'$, and signatures $t$ and $u$ such that $((\sigma(B), P), [\oc_t]) \hspace{-0.1em}\noJump^*\hspace{-0.1em} ((\overline{\sigma(B')},P'),[\oc_u]@T) \hspace{-0.1em}\noJump^*\hspace{-0.1em} ((e,Q),[\oc_{\sige}])$. By Lemma~\ref{lemma_realbeta_simple}, there are substitutions $\theta$ and $\theta'$ such that $\realbeta{(\sigma(B),P),[\oc_t]}\hspace{-0.15em}=\hspace{-0.15em}\oc_{s,d,n}A[\theta]$ and $\realbeta{(\overline{\sigma(B')},P'),[\oc_u]@T}\hspace{-0.15em}=\hspace{-0.15em}(\oc_{s',d',n'}\overline{A'}[\theta'])_{|T}$. So, according to Lemma~\ref{lemma_realbeta_nojump},
  \begin{equation*}
    \oc_{s,d,n}A[\theta] \leq (\oc_{s',d',n'}\overline{A'}[\theta'])_{|T}
  \end{equation*}
  By definition of $\leq$ on formulae, $(\oc_{s',d',n'}\overline{A'}[\theta'])_{|T} = \oc_{s'',d'',n''}A''$ with $s \geq s''$. Thus $A'[\theta]$ is of the shape $H[\wn_{s'',d'',n''} A'']$. Either $A'$ is of the shape $H[\wn_{s'',d'',n''}A'']$ (in this case by definition of $\form{0}$, $s' < s''$ so $s > s'$), or there exist sequences $A'=A_0,A_1\cdots,A_k$ of formulae, $X^0,\cdots,X^k$ of variables and $s_0,\cdots,s_k$ such that for $0 \leq i < k$, $A_i$ is of the shape $H_i[X^i_{s_i}]$ and there exists a $\exLab$ node $n_i$ whose associated variable is $X^i_{s_i}$ and whose associated formula is $A_{i+1}$. And $A_{k}$ is of the shape $H_k[\oc_{s'',d'',n''}A'']$. In this case, $s' < s_0 \leq s_1 \leq \cdots \leq s_{k-1} \leq s''$ so $s' < s'' \leq s$.
\end{proof}

\begin{lemma}\label{lemma_snll_dep}
  If $B \dcSim B'$, $\beta(\sigma(B))= \oc_{s,d,n}A$ and $\beta(\sigma(B'))= \oc_{s',d',n'}A'$ then $d > d'$ 
\end{lemma}
\begin{proof}
  By definition of $\dcSim$, there exists $i \neq j$ and paths of the shape $((\sigma(B),P), [\oc_t]) \csRel^+ ( (\overline{\sigma_i(B')}, P'_1) [\oc_{\sige}])$ and $((\sigma(B),P), [\oc_u]) \csRel^+ ((\overline{\sigma_j(B')}, P'_2), [\oc_{\sige}])$. Either $i \neq 1$ or $j \neq 1$. We suppose without loss of generality that $i \neq 1$. By definition of $SDNLL$, $\beta(\sigma_i(B))=\oc_{\_,d'',\_}\_$ with $d'' > d'$. Then, by Lemma~\ref{lemma_realbeta_csrel}, $d \geq d'' > d'$.
\end{proof}

\begin{lemma}\label{lemma_snll_nesting}
  If $B \nestSim B'$, $\beta(\sigma(B))= \oc_{s,d,n}A$ and $\beta(\sigma(B'))= \oc_{s',d',n'}A'$ then $n > n'$ 
\end{lemma}
\begin{proof}
  By definition of $\nestSim$, there exist $P,Q \in \Pot$, and a non-standard signature $t$, such that $((\sigma(B),P),[\oc_t]) \csRel^* ((\sigma(C),Q),[\oc_{\sige}])$. Let us consider the first context of the path such that the leftmost trace element is $\oc_u$ with $u$ a standard signature. This step must be of the shape $((e,R),[\oc_{\sigp(u)}]) \noJump ((f,R),[\oc_u])$ with $\overline{e}$ the conclusion of a $\digLab$ node. By definition of $SDNLL$, $\beta(e)=\wn_{\_,\_,n_e}\_$ and $\beta(f)=\wn_{\_,\_,n_f}\_$ with $n_e> n_f$. Then, by Lemmas~\ref{lemma_realbeta_simple} and~\ref{lemma_realbeta_csrel}, $n \geq n_e > n_f \geq n'$.
\end{proof}

\begin{coro}\label{coro_sdnll_stratifieds}
  Let $G$ be a $SDNLL$ proof-net, then $G$ is $\stratSNLL$-stratified, $\dcSim$-stratified and $\nestSim$-stratified. Moreover, for every $B \in \boxset{G}$ with $\beta(\sigma(B))=\oc_{s,d,n}A$, $\stratu{\stratSNLL}{B} \leq s$, $\stratu{\dcSim}{B} \leq d$ and $\stratu{\nestSim}{B} \leq n$.
\end{coro}
\begin{proof}
  Immediate consequence of the three previous lemmas.
\end{proof}

    \begin{figure} \centering       \tikzsetnextfilename{encoding_sdnll_n3}
      \begin{tikzpicture}
        \node [forall] (fa) at (0,0) {};
        \node [par]       (parg)   at ($(fa)+(0,0.8)$) {};
        \draw [ar] (fa) --++ (0,-0.8) node [type,left] {$\NSDN{s}{d}{n}$};
        \draw [ar] (parg) -- (fa) node [type,left] {$\oc_{s,d+1,n}(X_{s+1} \multimap X_{s+1}) \multimap \oc_{s,d,n}(X_{s+1} \multimap X_{s+1})$};
        \node [princdoor]    (princ3g)  at ($(parg)+(2.5, 4.2)$) {};
        \node [auxdoor]      (aux1g)    at ($(princ3g)+(-8,0)$) {};
        \node [auxdoor]      (aux2g)    at ($(aux1g)!0.333!(princ3g)$) {};
        \node [auxdoor]      (aux3g)    at ($(aux1g)!0.666!(princ3g)$) {};
        \node [cont]         (cont1g)   at ($(aux1g)!0.5!(aux2g)+(0,-0.6)$) {};
        \node [cont]         (cont2g)   at ($(cont1g)!0.5!(parg)$) {};
        \nvar{\hautTens}{1.2cm}
        \node [tensor]       (tens1g)   at ($(aux1g)+(0,\hautTens)$) {};
        \node [tensor]       (tens2g)   at ($(aux2g)+(0,\hautTens)$) {};
        \node [tensor]       (tens3g)   at ($(aux3g)+(0,\hautTens)$) {};
        \nvar{\decAx}{0.5cm}
        \node [ax]        (ax1n)     at ($(tens1g)+(-0.5,\decAx)$) {};
        \node [ax]        (ax2n)     at ($(tens1g)!0.5!(tens2g)+(0,\decAx)$) {};
        \node [ax]        (ax3n)     at ($(tens2g)!0.5!(tens3g)+(0,\decAx)$) {};
        \node [ax]        (ax4n)     at ($(princ3g)+(-0.7,\decAx + \hautTens)$) {};
        \node [par]          (parx)     at ($(princ3g)+(0,0.7)$) {};
        \draw [ar] (cont2g) -- (parg)   node [type,below left=-0.1cm] {$\wn_{s,d+1,n}(X_{s+1} \otimes X^\perp_{s+1})$};
        \draw [ar] (cont1g) -- (cont2g) node [type,below left=-0.1cm] {$\wn_{s,d+1,n}(X_{s+1} \otimes X^\perp_{s+1})$};
        \draw [ar] (aux1g) -- (cont1g)  node [type,below left=-0.1cm,pos=0.6] {$\wn_{s,d+1,n}(X_{s+1} \otimes X^\perp_{s+1})$};
        \draw [ar] (aux2g) -- (cont1g)  node [type,below right=-0.05cm] {$\wn_{s,d+1,n}(X_{s+1} \otimes X^\perp_{s+1})$};
        \draw [ar] (aux3g) -- (cont2g)  node [type] {$\wn_{s,d+1,n}(X_{s+1} \otimes X^\perp_{s+1})$};
        \draw [ar] (parx) -- (princ3g);
        \draw [ar] (tens1g) -- (aux1g) node [type,pos=0.7] {$X_{s+1} \otimes X^\perp_{s+1}$};
        \draw [ar] (tens2g) -- (aux2g) node [type,pos=0.7] {$X_{s+1} \otimes X^\perp_{s+1}$};
        \draw [ar] (tens3g) -- (aux3g) node [type,pos=0.7] {$X_{s+1} \otimes X^\perp_{s+1}$};
        \coordinate (int) at ($(ax1n)+(0.5,-0.4)$);
        \draw [out=-180,in= 170] (ax1n.180) to (int);
        \draw [ar,out=-10,in= 170] (int) to (parx);
        \draw [ar,out=-0 ,in= 120] (ax1n.0) to (tens1g);
        \draw [ar,out=-180 ,in=  60] (ax2n.180) to (tens1g);
        \draw [ar,out=-0   ,in=120 ] (ax2n.0) to (tens2g);
        \draw [ar,out=-180 ,in=  60] (ax3n.180) to (tens2g);
        \draw [ar,out=-0   ,in=120 ] (ax3n.0) to (tens3g);
        \draw [ar,out=-180 ,in=  60] (ax4n.180) to (tens3g);
        \draw [ar,out=-0 ,in=  60] (ax4n.0) to (parx);
        \draw (princ3g) -| ++(0.5,\hautTens + \decAx +0.3cm) -| ($(aux1g)+(-1,0)$) -- (aux1g) -- (aux2g) -- (aux3g) -- (princ3g);
        \draw (princ3g) -- (parg) node [type,pos=0.85] {$\oc_{s,d,n}(X_{s+1} \multimap X_{s+1})$};
      \end{tikzpicture}
      \caption{Encoding $\encodeS{s,d,n}{3}$ of $3$}
      \label{fig_encoding_sdnll_n}
    \end{figure}

\begin{theorem}\label{theo_sdnll_pn_bound}
Let $G$ be a $SDNLL$ proof-net, then the maximal reduction length of $G$ (with $x=|\dirEdges{G}|$, $\partial=\depthG{G}$, $S=1+\max_{B \in \boxset{G}}\straF{\sigma(B)}$, $D=\max_{B \in \boxset{G}}\dcF{\sigma(B)}$ and $N=1+\max_{B \in \boxset{G}}\nesF{\sigma(B)}$) is bounded by
  \begin{equation*}
    W_G \leq x^{1+ D^{S} \cdot \partial^{N \cdot S}}
  \end{equation*}
\end{theorem}
\begin{proof}
  The bound is an immediate consequence of Corollaries~\ref{coro_sdnll_stratifieds} and~\ref{coro_bound_poly_nest}.
\end{proof}

To formalize the polynomial time soundness of $SDNLL$, we need to define an encoding of binary lists. For any $s,d,n \in \mathbb{N}$, we define the formulae $\NSDN{s}{d}{n}$ and $\BSDN{s}{d}{n}$ by
\begin{align*}
  \NSDN{s}{d}{n}&=\forall X_{s+1}, \oc_{s,d+1,n}(X_{s+1} \multimap X_{s+1}) \multimap \oc_{s,d,n}(X_{s+1} \multimap X_{s+1})\\
  \BSDN{s}{d}{n}&=\forall X_{s+1}, \oc_{s,d+1,n}(X_{s+1} \multimap X_{s+1}) \multimap \oc_{s,d+1,n}(X_{s+1} \multimap X_{s+1}) \multimap \oc_{s,d,n}(X_{s+1} \multimap X_{s+1})
\end{align*}

For any $s,d,n \in \mathbb{N}$, $k \in \mathbb{N}$ and binary list $l$, we can define an encoding $\encodeS{s,d,n}{k}$ of $k$ as in Figure~\ref{fig_encoding_sdnll_n}. The encoding $\encodeS{s,d,n}{l}$ of $l$ can be defined similarly. We can verify that the sizes of $\encodeS{s,d,n}{k}$ and $\encodeS{s,d,n}{l}$ depend linearly on the size of $k$ and $l$. Finally, for every $k$ and $l$ there is exactly one box in $\encodeS{s,d,n}{k}$ and $\encodeS{s,d,n}{l}$.

\begin{theorem}\label{theo_sdnll_polysound}
  For every $SDNLL$ proof-net $G$ whose only conclusion is labelled by $\BSDN{s}{d}{n} \multimap A$, there exists a polynomial $P$ such that for every binary list $l$,
  \begin{equation*}
    W_{(G)\encodeS{s,d,n}{l}} \leq P\left(|l|\right)
  \end{equation*}
\end{theorem}
\begin{proof}
  By Theorem~\ref{theo_sdnll_pn_bound}, $(G)\encodeS{s,d,n}{l}$ is $\stratSNLL$-stratified, $\dcSim$-stratified and $\nestSim$-stratified. The depths of those relations are bounded by $\left|\boxset{(G)\encodeS{s,d,n}{l}}\right|=\left|\boxset{G}\right|+1$. We can conclude by Corollary~\ref{coro_bound_poly_nest} (and the linearity of $|\edges{\encodeS{s,d,n}{l}}|$ on $|l|$).
\end{proof}

\subsection{Encoding of $mL^4$}
\label{section_encoding}
\label{section_encoding_l4}
There are already many subsystems of $LL$ characterizing polynomial time. We argue that the interest of $SDNLL$ over the previous systems is its intentional expressivity. To support our claim we define an encoding of $mL^4$~\cite{baillot2010linear}. The encoding of a maximal subsystem of $MS$~\cite{roversi2010local} is defined in~\cite{perrinelMegathese}. Baillot and Mazza already proved that $LLL$ can be embedded in $mL^4$, thus $SDNLL$ is at least as expressive as the union of those systems.

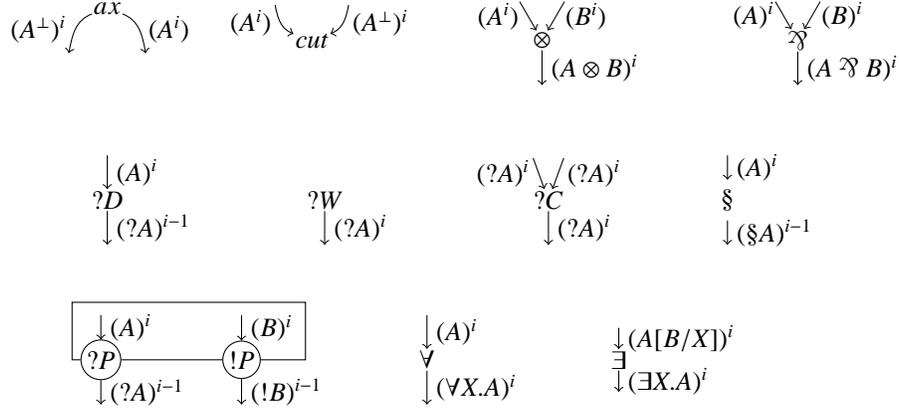
\begin{figure}
  \centering
  \begin{tikzpicture}
    \begin{scope}[scale = 0.85]
      \begin{scope}[shift={(0,0.5)}]
        \draw (0,0) node [ax] (ax) {};
        \draw[ar, out=-20,in=100] (ax) to node [type] {$(A^i)$}      node [level] {} (0.6,-0.7);
        \draw[ar, out=-160,in=80] (ax) to node [type,left] {$(A^\perp)^i$} node [level] {} (-0.6,-0.7);
      \end{scope}
      
      \begin{scope}[shift={(3.2,0)}]
        \draw (0,0) node [cut] (cut) {};
        \draw (cut) ++ (-0.6,0.7) node (G) {};
        \draw (cut) ++ (0.6,0.7) node (H) {};
        \draw[ar,out=-80,in=160] (G) to node [type,left] {$(A^i)$} node [level,pos=0.4] {} (cut);
        \draw[ar,out=-100,in=20] (H) to node [type] {$(A^\perp)^i$} node [level,pos=0.4] {} (cut);
      \end{scope}
      
      \begin{scope}[shift={(6.8,0)}]
        \draw (0,0) node [tensor] (tens)  {};
        \draw [ar] ($(tens)+(120:0.7)$) -- (tens) node [type,left] {$(A^i)$} node [level] {};
        \draw [ar] ($(tens)+( 60:0.7)$) -- (tens) node [type] {$(B^i)$} node [level,right] {};
        \draw [ar] (tens)--++(0,-0.7) node [type] {$(A \otimes B)^i$} node [level] {};
      \end{scope}
      
      \begin{scope}[shift={(10.8,0)}]
        \draw (0,0) node [par] (par)  {};
        \draw [ar] ($(par)+(120:0.7)$) -- (par) node [type,left] {$(A)^i$} node [level] {};
        \draw [ar] ($(par)+( 60:0.7)$) -- (par) node [type] {$(B)^i$} node [level,right] {};
        \draw [ar] (par)--++(0,-0.7) node [type] {$(A \parr B)^i$} node [level] {};
      \end{scope}

      \begin{scope}[shift={(5,-5)}]
        \draw (0,0) node [forall] (forall)  {};
        \draw [ar] ($(forall)+(0,0.7)$) -- (forall) node [type] {$(A)^i$} node [level] {};
        \draw [ar] (forall)--++(0,-0.7) node [type] {$(\forall X.A)^i$} node [level] {};
      \end{scope}
      
      \begin{scope}[shift={(8,-5)}]
        \draw (0,0) node [exists] (exists)  {};
        \draw [ar] ($(exists)+(0,0.5)$) -- (exists) node [type] {$(A[B/X])^i$} node [level] {};
        \draw [ar] (exists)--++(0,-0.5) node [type] {$(\exists X.A)^i$} node [level] {};
      \end{scope}

      \begin{scope}[shift={(-8.7,-2.5)}]
        \begin{scope}[shift={(8.70,0)}]
          \draw (0,0) node [der] (der)  {};
          \draw [ar] ($(der)+(0,0.7)$) -- (der) node [type] {$(A)^i$} node [level] {};
          \draw [ar] (der)--++(0,-0.7) node [type] {$(?A)^{i-1}$} node [level] {};
        \end{scope}

        \begin{scope}[shift={(12.1,0)}]
          \draw (0,0) node [weak] (weak)  {};
          \draw [ar] (weak)--++(0,-0.7) node [type] {$(?A)^{i}$} node [level] {};
        \end{scope}

        \begin{scope}[shift={(15.6,0)}]
          \draw (0,0) node [cont] (cont)  {};
          \draw [ar] ($(cont)+(110:0.7)$) --(cont) node [type,left] {$(?A)^i$} node [level] {};
          \draw [ar] ($(cont)+(70:0.7)$)  --(cont)  node [type] {$(?A)^i$} node [level,right] {};
          \draw [ar] (cont)--++(0,-0.7) node [type] {$(?A)^i$} node [level] {};
        \end{scope}
        \begin{scope}[shift={(18.4,0)}]
          \draw (0,0) node  (neut)  {$\S$};
          \draw [ar] ($(neut)+(0,0.7)$)--(neut) node [type] {$(A)^i$} node [level] {};
          \draw [ar] (neut)--++(0,-0.7) node [type] {$(\S A)^{i-1}$} node [level] {};
        \end{scope}
        
        \begin{scope}[shift={(10,-2.5)}]
          \draw (0.8,0) node [princdoor] (bang) {};
          \draw (-1.4,0) node [auxdoor] (whyn1) {};
          \draw[ar] ($(bang)+(0,0.7)$)--(bang) node [type] {$(B)^i$} node [level] {};
          \draw[ar] (bang) --++ (0,-0.7) node [type] {$(! B)^{i-1}$} node [level] {};
          \draw[ar] ($(whyn1)+(0,0.7)$)--(whyn1) node [type] {$(A)^i$} node  [level] {};
          \draw[ar] (whyn1) --++ (0,-0.7) node [type] {$(?A)^{i-1}$} node  [level] {};
          \draw (whyn1)--(bang) -| ++(1,0.9) -| ($(whyn1)+(-0.45,0)$) -- (whyn1);
        \end{scope}
      \end{scope}
    \end{scope}
  \end{tikzpicture}
  \caption{\label{figureLabelL3}Relations between levels of neighbour edges in $L^4$. We also allow boxes with $0$ auxiliary doors.}
\end{figure}

The formulae of $mL^4$ are defined as the formulae of linear logic with an additional modality $\S$ and an element of $\mathbb{N}$ indexing the formula. More formally, the set $\form{L4}$ of formulae of $mL^4$ is defined by the following grammar.\label{def_formulae_l4}
\begin{align*}
  \form{L4} &= \gorm{L4} \times \mathbb{N}\\
  \gorm{L4} &= X \mid X^\perp \mid \gorm{L4}\otimes\gorm{L4} \mid \gorm{L4}\parr\gorm{L4} \mid \forall X. \gorm{L4} \mid \exists X.\gorm{L4} \mid \oc \gorm{L4} \mid \wn \gorm{L4} \mid \S \gorm{L4}
\end{align*}
The index in $\mathbb{N}$ (called {\em level}) is usually written as an exponent. Intuitively, if the principal edge of $B$ is labelled with $(\oc A)^s$, the label $s$ represents the stratum of $B$ for $\stratSNLL$. More precisely, it corresponds to a formula of the shape $\oc_{s+\depth{B},\_,\_}A$ in $SDNLL$. Let us notice that, to connect two boxes $B$ and $B'$ labelled with $(\oc A)^s$ and $(\oc A')^{s'}$ with $s \neq s'$, we need to use $\S$ nodes.

Let us notice that every box of $mL^4$ proof-nets have only one auxiliary door. Thus $\dcSim=\varnothing$ and, for every box $B$ $\stratu{\dcSim}{B}=1$. We can also notice that there is no $\digLab$ node in $mL^4$ proof-nets, so for every box $B$, we have $\stratu{\nestSim}{B}=1$.

\label{def_expodepth_form}We define a mapping $\|\_\|$ from formulae contexts of $\gorm{L4}$ to $\mathbb{N}$ which will be used to decide the indices of variables and exponential modalities. For every formulae $A$ in $\gorm{L4}$ and formula context $H$, $\|\circ\|=0$, $\| C \otimes H \|=\|H \otimes C\|=\|C \parr H\|=\|H \parr C\|=\|H\|$, $\|\forall X.H\|=\|\exists X.H\|=\|H\|$, and $\|\oc H\|=\|\wn H\|=\|\S H\|=1+\|H\|$.

\begin{figure}
  \centering
  \AxiomC{$\hspace{1em}$}
  \RightLabel{\hspace{-0.2em}$\axLab$}
  \UnaryInfC{\hspace{-0.1em}$x: A^{\varnothing} \vdash x: A$\hspace{-0.1em}}  
  \DisplayProof \hspace{0.6em}
  \AxiomC{\hspace{-0.1em}$\Gamma \vdash t:A$\hspace{-0.4em}}
  \AxiomC{$X_s$ not free in $\Gamma$\hspace{-0.1em}}
  \RightLabel{\hspace{-0.2em}$\forall_i$}
  \BinaryInfC{$\Gamma \vdash t: \forall X_s.A$}
  \DisplayProof \hspace{0.6em}
  \AxiomC{\hspace{-0.1em}$\Gamma \vdash t:\forall X_s.A$\hspace{-0.4em}}
  \AxiomC{$s_B^{min} \geq s$\hspace{-0.1em}}
  \RightLabel{\hspace{-0.2em}$\forall_e$}
  \BinaryInfC{$\Gamma \vdash t:A[B/X]$}
  \DisplayProof

  \vspace{1em}

  \AxiomC{$\Gamma,x:A^{\varnothing} \vdash t:B$}
  \RightLabel{\hspace{-0.2em}$\derLab$}
  \UnaryInfC{$\Gamma,x:A^{s,d,n} \vdash t:B$}
  \DisplayProof \hspace{0.6em}
  \AxiomC{$\Gamma \vdash t:B$}
  \RightLabel{\hspace{-0.2em}$\weakLab$}
  \UnaryInfC{$\Gamma,x:A^{s,d,n} \vdash t:B$}
  \DisplayProof \hspace{0.6em}
  \AxiomC{$\Gamma,y:A^{s,d,n},z:A^{s,d,n} \vdash t:B$}
  \RightLabel{\hspace{-0.2em}$\contLab$}
  \UnaryInfC{$\Gamma,x:A^{s,d,n} \vdash t[x/y;x/z]:B$}
  \DisplayProof 

  \vspace{1em}

  \AxiomC{$\Gamma,x:A^{\varnothing} \vdash t:B$}
  \RightLabel{\hspace{-0.2em}$\multimap_i$}
  \UnaryInfC{\hspace{-0.15em}$\Gamma \vdash \lambda x.t : A \multimap B$\hspace{-0.15em}}
  \DisplayProof \hspace{0.6em}
  \AxiomC{$\Gamma,x:A^{s,d,n} \vdash t:B$}
  \RightLabel{\hspace{-0.2em}$\Rightarrow_i$}
  \UnaryInfC{\hspace{-0.15em}$\Gamma \vdash \lambda x.t : \oc_{s,d,n}A \multimap B$\hspace{-0.15em}}
  \DisplayProof \hspace{0.6em}
  \AxiomC{\hspace{-0.15em}$\Gamma \vdash t:A \multimap B$\hspace{-0.5em}}
  \AxiomC{$\Delta \vdash u:A$\hspace{-0.15em}}
  \RightLabel{\hspace{-0.2em}$\multimap_e$}
  \BinaryInfC{$\Gamma,\Delta \vdash (t)u : B$}
  \DisplayProof

  \vspace{1em}

  \AxiomC{$\Gamma \vdash t: \oc_{s,d,n}A \multimap B$}
  \AxiomC{$\Delta,\Sigma^{\varnothing} \vdash u:A$}
  \AxiomC{$d(\Delta\cup \Sigma)\geqq d \hspace{0.5em} n(\Sigma)\geq n \hspace{0.5em}n(\Delta)>n$}
  \RightLabel{$\Rightarrow_e$}
  \TrinaryInfC{$\Gamma,\Delta,\Sigma \vdash (t)u:B$}
  \DisplayProof \hspace{1em}
  \caption{\label{snll_lambda_rules}$SDNLL_{\lambda}$ as a $\lambda$-calculus type-system.}
\end{figure}

Any $mL^4$ proof-net $G$ can be transformed into a $SDNLL$ proof-net $G'$ as follows: for every variable $X$ appearing in the proof-net, we define $M_X$ as the maximum of the set 
\begin{equation*}
  \Set*{ s+ \|H\| }{ \beta(e)=(H[X])^s \text{ or }\beta(e)=(H[X^\perp])^s}
  \end{equation*}
  Then, we replace every occurrence of $X$ by $X_{M_X}$. If $\beta(e)=(H[\oc A])^s$ (resp. $(H[\wn A])^s$), we replace the modality by $\oc_{s+\|H\|,1,0}$ (resp. $\wn_{s+\|H\|,1,0}$). One can easily verify that $G$ is a valid $SDNLL$ proof-net. The $\S$ node becomes trivial (it does not change the sequent). 

  The most interesting constraint to check is the constraint on doors. Let us suppose that $e$ is the premise of the $i$-th auxiliary door of a box $B$ and $f$ is its conclusion. If $\beta_G(e)=(H[\oc A])^{s}$ then $\beta_G(f)=\wn (H[\oc A])^{s-1}$. We can notice that $\beta_{G'}(e)=H'[\oc_{s+\|H\|,\hspace{0.1em}1,0}\_]$ and $\beta_{G'}(f)=H'[\oc_{s-1+\|\wn H\|,\hspace{0.1em}1,0}\_]$. Those labels are the same because $s+\|H\|=(s-1)+(1+\|H\|)=(s-1)+\|\wn H\|$.

\subsection{$SDNLL$ as a type-system for $\lambda$-calculus}\label{section_sdnll_lambda}
As noticed by Baillot and Terui~\cite{baillot2004light}, translating naively a subsystem of linear logic into a type-system for $\lambda$-calculus can result in a type-system which enjoys neither subject reduction nor the complexity bound enforced by the linear logic subsystem. The subsystem we define is heavily inspired by $DLAL$. For instance, the proof of subject reduction follows the proof of subject reduction of $DLAL$ presented in~\cite{baillot2004lightlong}.

We restrict the formulae considered by only allowing $\oc$ modalities on the left side of $\multimap$ connectives.

\begin{definition}\label{def_sdnll_lambda} For $s \in \mathbb{N}$, we define $\formla{s}$ by the following grammar (with $t,d,n \in \mathbb{N}$, $t \geq s$ and $X$ ranges over a countable set of variables). Notice that $\formla{0} \supseteq \formla{1} \supseteq \cdots$
  \begin{equation*}
    \formla{s} := X_{t} \mid  \formla{s} \multimap \formla{s} \mid \oc_{t,d,n}\formla{t+1} \multimap \formla{s} \mid \forall X_{t}. \formla{s}
  \end{equation*}
\end{definition}

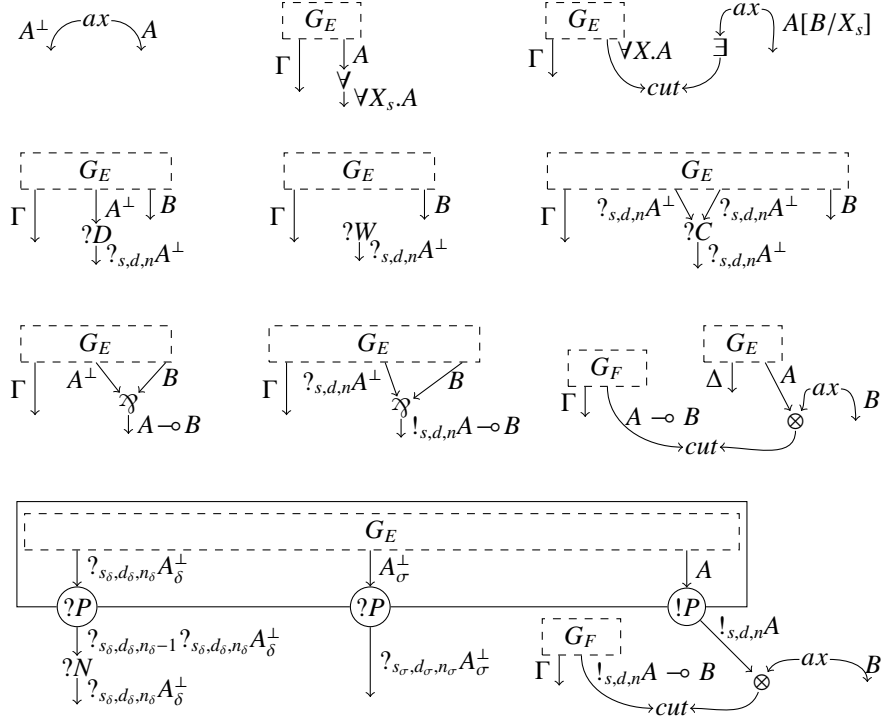
\begin{figure}\centering
  \begin{tikzpicture}
    \begin{scope}
      \node[ax] (ax) at (0,0) {};
      \draw[ar,out=  0,in=90] (ax) to node [type]      {$A$}      ($(ax)+( 0.6,-0.4)$) {};
      \draw[ar,out=180,in=90] (ax) to node [type,left] {$A^\perp$} ($(ax)+(-0.6,-0.4)$) {};
    \end{scope}
    \begin{scope}[shift={(3,0)}]
      \node[proofnet,minimum width=1cm] (G) at (0,0) {$G_E$};
      \node[forall] (fa) at ($(G.-40)+(0,-0.55)$) {};
      \draw[ar] (G.-40) --(fa)       node [type] {$A$};
      \draw[ar] (fa) --++ (0,-0.35)   node [type] {$\forall X_s.A$};
      \draw[ar] (G.-140)--++(0,-0.7) node [edgename] {$\Gamma$};
    \end{scope}
    \begin{scope}[shift={(6.5,0)}]
      \node[proofnet,minimum width=1cm] (G) at (0,0) {$G_E$};
      \node[exists] (ex)  at ($(G.-40)+(1.5,-0.1)$) {};
      \node[cut]    (cut) at ($(G.-40)!0.5!(ex)+(0,-0.6)$) {};
      \node         (end) at ($(ex)+(0.7,-0.2)$)   {};
      \node[ax]     (ax)  at ($(ex)!0.5!(end)+(0,0.6)$) {};
      \draw[ar,out=-90,in=180] (G.-40) to node [type,pos=0.15] {$\forall X.A$} (cut);
      \draw[ar,out=-90,in=  0] (ex)    to (cut);
      \draw[ar,out=180,in= 90] (ax)    to (ex);
      \draw[ar,out=  0,in= 90] (ax)    to node [type] {$A[B/X_s]$} (end);
      \draw[ar] (G.-150) --++(0,-0.5) node [edgename] {$\Gamma$};
    \end{scope}
    \begin{scope}[shift={(0,-2)}]
      \node[proofnet,minimum width=2cm] (G) at (0,0) {$G_E$};
      \node[der] (der) at ($(G.-90)+(0,-0.6)$) {};
      \draw[ar] (G.-90)--  (der)     node [type] {$A^{\perp}$};
      \draw[ar] (der)  --++(0,-0.4)  node [type] {$\wn_{s,d,n}A^{\perp}$};
      \draw[ar] (G.-19) --++(0,-0.4) node [type] {$B$};
      \draw[ar] (G.-163)--++(0,-0.7) node [edgename] {$\Gamma$};
    \end{scope}
    \begin{scope}[shift={(3.5,-2)}]
      \node[proofnet,minimum width=2cm] (G) at (0,0) {$G_E$};
      \node[weak] (weak) at ($(G.-90)+(0,-0.55)$) {};
      \draw[ar] (weak)  --++(0,-0.4)  node [type] {$\wn_{s,d,n}A^{\perp}$};
      \draw[ar] (G.-16) --++(0,-0.4) node [type] {$B$};
      \draw[ar] (G.-164)--++(0,-0.7) node [edgename] {$\Gamma$};
    \end{scope}
    \begin{scope}[shift={(8,-2)}]
      \node[proofnet,minimum width=4cm] (G) at (0,0) {$G_E$};
      \node[cont] (cont) at ($(G.-90)+(0,-0.55)$) {};
      \draw[ar] (G.-140)-- (cont)     node [type,pos=0.65,left] {$\wn_{s,d,n}A^{\perp}$};
      \draw[ar] (G.-40)--  (cont)     node [type,pos=0.65]      {$\wn_{s,d,n}A^{\perp}$};
      \draw[ar] (cont)  --++(0,-0.5)  node [type] {$\wn_{s,d,n}A^{\perp}$};
      \draw[ar] (G.-8) --++(0,-0.4) node [type] {$B$};
      \draw[ar] (G.-172)--++(0,-0.7) node [edgename] {$\Gamma$};
    \end{scope}
    \begin{scope}[shift={(0,-4.3)}]
      \node[proofnet,minimum width=2cm] (G) at (0,0) {$G_E$};
      \node[par] (par) at ($(G.-30)+(0,-0.55)$) {};
      \draw[ar] (G.-90) --(par) node [type,left] {$A^{\perp}$};
      \draw[ar] (G.-15) --(par) node [type]      {$B$};
      \draw[ar] (par)--++(0,-0.4) node [type] {$A\hspace{-0.2em} \multimap \hspace{-0.2em}B$};
      \draw[ar] (G.-163)--++(0,-0.7) node [edgename] {$\Gamma$};
    \end{scope}
    \begin{scope}[shift={(3.7,-4.3)}]
      \node[proofnet,minimum width=2.8cm] (G) at (0,0) {$G_E$};
      \node[par] (par) at ($(G.-35)+(0,-0.6)$) {};
      \draw[ar] (G.-60) --(par) node [type,left,pos=0.6] {$\wn_{s,d,n}A^{\perp}$};
      \draw[ar] (G.-12) --(par) node [type]      {$B$};
      \draw[ar] (par)--++(0,-0.4) node [type] {$\oc_{s,d,n}A \hspace{-0.2em}\multimap \hspace{-0.2em}B$};
      \draw[ar] (G.-168)--++(0,-0.7) node [edgename] {$\Gamma$};
    \end{scope}
    \begin{scope}[shift={(8.6,-4.3)}]
      \node[proofnet,minimum width=0.8cm] (G) at (0,0) {$G_{E}$};
      \node[ax]    (ax)   at ($(G)+(1.1,-0.6)$) {};
      \node[tensor](tens) at ($(G.-40)!0.5!(ax)+(0,-0.6)$) {};
      \coordinate  (end)  at ($(tens)+(0.8,0)$);
      \node[proofnet,minimum width=1cm] (H) at ($(tens)+(-2.5,0.7)$) {$G_F$};
      \node[cut] (cut) at ($(H)!0.5!(tens)+(0,-0.7)$) {};
      \draw[ar,out=-90,in=180] (H.-90) to node [type,pos=0.3] {$A \multimap B$} (cut);
      \draw[ar,out=-90,in=  0] (tens)  to (cut);
      \draw[ar] (G.-40) -- (tens) node [type,pos=0.25] {$A$};
      \draw[ar,out=180,in=60] (ax) to (tens);
      \draw[ar,out=  0,in=90] (ax) to node [type] {$B$} (end);
      \draw[ar] (G.-120) -- ++(0,-0.4) node [type,left] {$\Delta$};
      \draw[ar] (H.-140) -- ++(0,-0.4) node [type,left] {$\Gamma$};
    \end{scope}
    \begin{scope}[shift={(3.8,-6.8)}]
      \node[proofnet,minimum width=9.5cm] (G) at (0,0) {$G_{E}$};
      \node[princdoor] (pg) at ($(G.-4)+(0.5,-0.75)$) {};
      \node[auxdoor] (sigma) at ($(G.-122) +(0,-0.75)$) {};
      \node[auxdoor] (delta) at ($(G.-176)+(-0.5,-0.75)$) {};
      \node[dig] (dig) at ($(delta)+(0,-0.8)$) {};
      \draw [ar] ($(G.-4)+(0.5,0)$)   --(pg)    node [type] {$A$};
      \draw [ar] (G.-122)  --(sigma) node [type] {$A^{\perp}_{\sigma}$};
      \draw [ar] ($(G.-176)+(-0.5,0)$)--(delta) node [type] {$\wn_{s_{\delta},d_{\delta},n_{\delta}}A^{\perp}_{\delta}$};
      \draw [ar] (delta)--(dig) node [type] {$\wn_{s_{\delta},d_{\delta},n_{\delta}-1}\wn_{s_{\delta},d_{\delta},n_{\delta}}A^{\perp}_{\delta}$};
      \draw [ar] (dig)--++(0,-0.5) node [type] {$\wn_{s_{\delta},d_{\delta},n_{\delta}}A^{\perp}_{\delta}$};
      \draw [ar] (sigma)--++(0,-1.2) node [type] {$\wn_{s_{\sigma},d_{\sigma},n_{\sigma}}A^{\perp}_{\sigma}$};
      \draw (pg) -| ++(0.8,1.4) -| ($(delta)+(-0.8,0)$) -- (delta) -- (sigma) -- (pg);
      \node[ax]    (ax)   at ($(pg)+(1.7,-0.7)$) {};
      \node[tensor](tens) at ($(pg)+(1,-1)$) {};
      \coordinate  (end)  at ($(tens)+(1.4,0)$);
      \node[proofnet,minimum width=1cm] (H) at ($(tens)+(-2.4,0.6)$) {$G_F$};
      \node[cut] (cut) at ($(H)!0.5!(tens)+(0,-0.65)$) {};
      \draw[ar,out=-90,in=180] (H.-90) to node [type,pos=0.2] {$\oc_{s,d,n}A \multimap B$} (cut);
      \draw[ar,out=-110,in=  0] (tens)  to (cut);
      \draw[ar] (pg) -- (tens) node [type,pos=0.15] {$\oc_{s,d,n}A$};
      \draw[ar,out=180,in=60] (ax) to (tens);
      \draw[ar,out=  0,in=90] (ax) to node [type] {$B$} (end);
      \draw[ar] (H.-140) -- ++(0,-0.4) node [type,left] {$\Gamma$};
    \end{scope}
  \end{tikzpicture}
  \caption{\label{fig_pn_lambda_sdnll}Derivations of $SDNLL_\lambda$ can be translated into $SDNLL$ proof-nets}
\end{figure}

\label{def_contexts_others}We define contexts\footnote{Because we do not use context semantics in this subsection, there is no ambiguity.} as sets of the shape $\{x_1:A_1^{l_1},\cdots,x_k:A_k^{l_k}\}$ where the $x_i$s are pairwise distinct variables of $\lambda$-calculus, the $A_i$s are formulae of $\formla{0}$ and the $l_i$s are elements of $\{\varnothing\} \cup \mathbb{N}^3$. Intuitively $A^{s,d,n}$ represents $\oc_{s,d,n}A$ while $A^{\varnothing}$ represents $A$. The set of all contexts is written $\conla$, the set of contexts whose labels are all in $\mathbb{N}^3$ is written $\conoc$, the set of contexts whose labels are all equal to $\varnothing$ is written $\conlin$.

\label{def_op_on_contexts}In this paragraph, we consider $\Gamma=\{x_1:A_1^{s_1,d_1,n_1},\cdots,x_k:A_k^{s_k,d_k,n_k}\} \in \conoc$. Then we write $\Gamma^{\varnothing}$ for the context $\{x_1:A_1^{\varnothing},\cdots,x_k:A_k^{\varnothing}\}$. For $s,d,n \in \mathbb{Z}$, we write $\Gamma^{s,d,n}$ for $\{x_1:A_1^{s_1+s,d_1+d,n_1+n},\cdots,x_k:A_k^{s_k+s,d_k+d,n_k+d}\}$. We write $s(\Gamma)$ for the multiset of left indices, more formally $s(\Gamma)=\lmul x \mapsto \left|\Set*{i \in \mathbb{N}}{s_i=x}\right| \rmul$.\label{def_sgamma} We define $d(\Gamma)$ and $n(\Gamma)$ similarly.

\label{def_multiset_orders}If $M$ is a multiset, then we write $M \geq x$ for ``for every $y$ such that $M(y)>0$, we have $y\geq x$''. Similarly, we write $M > y$ for ``for every $y$ such that $M(y)>0$, we have $y> x$''. Finally, we write $M \geqq y$ for ``$M \geq x$ and $M(x) \leq 1$''.

We present the type system $SDNLL_{\lambda}$ in Figure~\ref{snll_lambda_rules}. In the type derivations, judgements are of the shape $\Gamma \vdash t:A$ with $\Gamma$ a context. If $x:B^{\varnothing}$ is in $\Gamma$ then $x$ appears exactly once in $t$.

To prove subject reduction and the polynomial bound we define (in Figure~\ref{fig_pn_lambda_sdnll}) for every type derivation $D$ of $A_1^{\varnothing},\cdots,A_k^{\varnothing},B_1^{s_1,d_1,n_1},\cdots,B_l^{s_l,d_l,n_l} \vdash t:C$, we define a $SDNLL$ proof-net $G_D$ with $k+l+1$ conclusions labelled with $A_1^{\perp}\hspace{-0.05em},\hspace{-0.1em}\cdots\hspace{-0.05em},\hspace{-0.1em}A_k^{\perp}\hspace{-0.05em},\hspace{-0.1em}\wn_{s_1,d_1,n_1}B_1^{\perp}\hspace{-0.05em},\hspace{-0.1em}\cdots\hspace{-0.1em},\hspace{-0.1em}\wn_{s_l,d_l,n_l}B_l^{\perp}$ and $C$. In Figure~\ref{fig_pn_lambda_sdnll}, we suppose that the derivation $D$ is obtained by applying a rule $r$ (the rule used is at the same position in Figure~\ref{snll_lambda_rules}) to the derivation $E$ (if the last rule is binary, the derivation on the left is named $F$).

\begin{lemma}[linear substitution]\label{lemma_linear_substitution}
  Let us consider derivations $D$ and $E$ of respective conclusions $\Delta \vdash u:A$ and $\Gamma,x:A^{\varnothing} \vdash t:B$. Then there exists a derivation $F$ of conclusion $\Gamma,\Delta \vdash t[u/x] :B$ and 

  \begin{equation*}
  \tikzsetnextfilename{snll_linear_substitution}
  \begin{tikzpicture}
    \node [proofnet,minimum width=2cm] (D) at (0,0) {$D$};
    \node [proofnet,minimum width=2.5cm] (E) at ($(D)+(3.5,0)$) {$E$};
    \node [cut]    (cut) at ($(D)!0.5!(E)+(0,-0.5)$) {};
    \draw [ar,out=-90,in=180] (D. -20) to node [type,pos=0.8,below left=-0.1cm] {$u:A$}       (cut);
    \draw [ar,out=-90,in=  0] (E.-160) to node [type,pos=0.8,below right=-0.1cm] {$x:A^{\perp}$} (cut);
    \draw [ar] (D.-160) --++(0,-0.5) node [edgename] {$\Delta$};
    \draw [ar] (E.-35)  --++(0,-0.5) node [edgename] {$\Gamma$};
    \draw [ar] (E.-15)  --++(0,-0.5) node [type] {$t:B$};

    \node [proofnet,minimum width=2cm] (F) at ($(E)+(5,0)$) {$F$};
    \draw [ar] (F.-160)--++(0,-0.5) node [edgename] {$\Delta$};
    \draw [ar] (F.-90) --++(0,-0.5) node [edgename] {$\Gamma$};
    \draw [ar] (F.-20) --++(0,-0.5) node [type]     {$t[u/x]:B$};
    
    \draw [->] ($(E.0)!0.25!(F.180)$) -- ($(E.0)!0.75!(F.180)$) node [below left] {$cut$} node [above] {$*$} ;
  \end{tikzpicture}
  \end{equation*}
\end{lemma}
\begin{proof}
  Simple induction on $E$. Because the label of $x$ is $\varnothing$, $x$ is not the conclusion of a $\derLab$, $\contLab$ or $\fauxLab$ node.
\end{proof}

\begin{lemma}[exponential substitution]\label{lemma_exponential_substitution}
  Let us consider derivations $D$ and $E$ of respective conclusions $\Delta,\Sigma^{\varnothing} \vdash u:A$ and $\Gamma,x:A^{s,d,n} \vdash t:B$ with $d(\Delta \cup \Sigma) \geqq d$, $n(\Sigma)\geq n$ and $n(\Delta)>n$. Then there exists a derivation $F$ of conclusion $\Gamma,\Delta,\Sigma \vdash t[u/x] :B$ and 

  \begin{equation*}
  \tikzsetnextfilename{snll_exponential_substitution}
  \begin{tikzpicture}
    \node [proofnet,minimum width=2.7cm] (D) at (0,0) {$D$};
    \node [princdoor] (pri) at ($(D.-13)+(0,-0.75)$) {};
    \node [auxdoor]   (sigma) at ($(D.-90)+(0,-0.75)$) {};
    \node [auxdoor]   (delta) at ($(D.-167)+(0,-0.75)$) {};
    \node [dig]       (dig)   at ($(delta)+(0,-0.6)$) {};
    \draw [ar] (D.-13)  -- (pri)   node [type] {$A$};
    \draw [ar] (D.-90)  -- (sigma) node [edgename] {$\Sigma$};
    \draw [ar] (D.-167) -- (delta) node [edgename] {$\Delta$};
    \draw [ar] (delta)--(dig);
    \draw [ar] (dig)  --++ (0,-0.3);
    \draw [ar] (sigma)--++ (0,-0.8);
    \draw (pri) -|++(0.6,1.35)-|($(delta)+(-0.6,0)$)--(delta)--(sigma)--(pri);
    \node [proofnet,minimum width=2.5cm] (E) at ($(pri)+(2.8,0)$) {$E$};
    \node [cut]    (cut) at ($(pri)!0.5!(E)+(0,-0.5)$) {};
    \draw [ar,out=-60,in=180] (pri) to node [type,pos=0.5,below] {$u:\oc_{s,d,n}A$}       (cut);
    \draw [ar,out=-90,in=  0] (E.-160) to node [type,pos=0.5,below] {$x:\wn_{s,d,n}A^{\perp}$} (cut);
    \draw [ar] (E.-35)  --++(0,-0.5) node [edgename] {$\Gamma$};
    \draw [ar] (E.-15)  --++(0,-0.5) node [type] {$t:B$};

    \node [proofnet,minimum width=2cm] (F) at ($(E)+(3.9,0)$) {$F$};
    \draw [ar] (F.-166)--++(0,-0.5) node [edgename] {$\Delta$};
    \draw [ar] (F.-145)--++(0,-0.5) node [edgename] {$\Sigma$};
    \draw [ar] (F.-35) --++(0,-0.5) node [edgename] {$\Gamma$};
    \draw [ar] (F.-14) --++(0,-0.5) node [type]     {$t[u/x]:B$};
    
    \draw [->] ($(E.0)!0.22!(F.180)$) -- ($(E.0)!0.78!(F.180)$) node [below left] {$cut$} node [above] {$*$};
  \end{tikzpicture}
  \end{equation*}
\end{lemma}
\begin{proof}
  By induction on $E$. The most interesting step is the $\Rightarrow_e$ step. In this case, let us write $C$ for the box created in this step, and let us set $\oc_{s',d',n'}\_=\sigma(C)$. Either $d>d'$ so $d(\Delta \cup \Sigma) \geq d > d'$. Or $d=d'$, $d(\Delta \cup \Sigma) \geqq d \geq d'$ and $d(\Gamma) > d'$, thus $d(\Delta \cup \Sigma \cup \Gamma) \geqq d'$.
\end{proof}

\begin{lemma}\label{lemma_forall_elim}
  Let us consider a derivation $D$ of conclusion $\Delta \vdash \lambda x.t:A$ then $G_{D} \cutRel^* G_{D'}$ (considering the untyped proof-nets) with $D'$ a derivation of conclusion $\Gamma \vdash \lambda x.t:A$ and the last rule $R$ introduces the top connective of $A$ (if $A=\forall A_1$ then $R=\forall_i$, if $A=\oc_{s,d,n}A_1 \multimap A_2$ then $R$ is $\Rightarrow_i$, if $A=A_1 \multimap A_2$ with $A_1 \in \formla{0}$ then $R=\multimap_i$)
\end{lemma}
\begin{proof}
  We prove it by induction on $D$. The last rule $R$ cannot be in $\{ax,\multimap_e,\Rightarrow_e\}$, because the $\lambda$-term of the conclusion would be of the shape $\lambda x.t$. If $R$ in $\{\forall_i,\multimap_i,\Rightarrow_i\}$, then the lemma is trivial. If $R$ is in $\{ \derLab, \weakLab, \contLab \}$, the derivation is of the shape:
  \begin{equation*}
    \AxiomC{$E$}
    \UnaryInfC{$\Delta \vdash \lambda x.t:A$}
    \RightLabel{$R$}
    \UnaryInfC{$\Gamma \vdash \lambda x.t:A$}
    \DisplayProof
  \end{equation*}
  By the induction hypothesis, $G_E \cutRel^* G_{E'}$ with $E'$ a derivation of conclusion $\Delta \vdash \lambda x.t:A$ and the last rule of $E'$ introduces the top connective of $A$. We will examine the case where this last rule is $\multimap_i$, the two other cases are similar. Let us examine the following derivation $D''$ (let us notice that, because $G_E \cutRel^* G_{E'}$, we have $G_D \cutRel^* G_{D''}$).
  \begin{equation*}
    \AxiomC{$F'$}
    \UnaryInfC{$\Delta,x:A_1 \vdash t: A_2$}
    \RightLabel{$\multimap_i$}
    \UnaryInfC{$\Delta \vdash \lambda x.t:A_1 \multimap A_2$}
    \RightLabel{$R$}
    \UnaryInfC{$\Gamma \vdash \lambda x.t:A_1 \multimap A_2$}
    \DisplayProof
  \end{equation*}
  In every case we can define $D'$ as the following derivation (let us notice that $G_{D''}=G_{D'}$)
  \begin{equation*}
    \AxiomC{$F'$}
    \UnaryInfC{$\Delta,x:A_1 \vdash t: A_2$}
    \RightLabel{$R$}
    \UnaryInfC{$\Gamma,x:A_1 \vdash t: A_2$}
    \RightLabel{$\multimap_i$}
    \UnaryInfC{$\Gamma \vdash \lambda x.t:A_1 \multimap A_2$}
    \DisplayProof
  \end{equation*}

  The last case to examine is $R=\forall_e$. In this case, the derivation is of the shape:
  \begin{equation*}
    \AxiomC{$E$}
    \UnaryInfC{$\Gamma \vdash \lambda x.t:\forall X_s.A_1$}
    \RightLabel{$R$}
    \UnaryInfC{$\Gamma \vdash \lambda x.t:A_1[B/X_s]$}
    \DisplayProof
  \end{equation*}
  By the induction hypothesis, $G_E \cutRel^* G_{E'}$ with $E'$ a derivation of conclusion $\Gamma \vdash \lambda x.t:\forall X_s.A$ and the last rule of $E'$ is $\forall_i$. Let us examine the following derivation $D''$ (let us notice that, because $G_E \cutRel^* G_{E'}$, we have $G_D \cutRel^* G_{D''}$)
  \begin{equation*}
    \AxiomC{$F'$}
    \UnaryInfC{$\Gamma \vdash \lambda x.t: A_1$}
    \RightLabel{$\forall_i$}
    \UnaryInfC{$\Gamma \vdash \lambda x.t: \forall X_s.A_1$}
    \RightLabel{$\forall_e$}
    \UnaryInfC{$\Gamma \vdash \lambda x.t: A_1[B/X]$}
    \DisplayProof
  \end{equation*}
  Then we can set $D'$ as the $SDNLL$ proof-net obtained from $F'$ by replacing $X_s$ by $B$ in the derivation. We can notice that $D'' \cutRel^2 D'$ (a $\forall/\exists$ step and an axiom step).
\end{proof}

\begin{figure}\centering
  \AxiomC{$~$}
  \UnaryInfC{\hspace{-0.4em}$g: (X_s \multimap X_s)^{\varnothing} \vdash g : X_s \multimap X_s$\hspace{-0.4em}}
  \AxiomC{$~$}
  \UnaryInfC{\hspace{-0.4em}$h: (X_s \multimap X_s)^{\varnothing} \vdash h : X_{s} \multimap X_{s}$\hspace{-0.4em}}
  \AxiomC{$~$}
  \UnaryInfC{\hspace{-0.4em}$a: (X_{s})^{\varnothing} \vdash a: X_{s}$\hspace{-0.4em}}
  \BinaryInfC{$h: (X_s \multimap X_s)^{\varnothing}, a: (X_{s})^{\varnothing }\vdash (h) a : X_{s}$}
  \BinaryInfC{\hspace{-0.4em}$g: (X_{s} \multimap X_{s})^{\varnothing},h: (X_{s} \multimap X_{s})^{\varnothing}, a: (X_{s})^{\varnothing}\vdash (g) (h) a : X_{s}$\hspace{-0.4em}}
  \UnaryInfC{$g: (X_{s} \multimap X_{s})^{s-1,d,n},h: (X_{s} \multimap X_{s})^{\varnothing}, a: (X_{s})^{\varnothing}\vdash (g) (h) a : X_{s}$}
  \UnaryInfC{$g: (X_{s} \multimap X_{s})^{s-1,d,n},h: (X_{s} \multimap X_{s})^{s-1,d,n}, a: (X_{s})^{\varnothing}\vdash (g) (h) a : X_{s}$}
  \UnaryInfC{$f: (X_{s} \multimap X_{s})^{s-1,d,n}, a: (X_{s})^{\varnothing}\vdash (f) (f) a : X_{s}$}
  \UnaryInfC{$f: (X_{s} \multimap X_{s})^{s-1,d,n} \vdash \lambda a.(f) (f)a : X_{s} \multimap X_{s}$}
  \UnaryInfC{$\vdash \lambda f. \lambda a.(f) (f)a : \oc_{s-1,d,n}(X_{s} \multimap X_{s}) \multimap X_{s} \multimap X_{s}$}
  \UnaryInfC{$\vdash \lambda f. \lambda a.(f) (f) a : \forall X_{s}, \oc_{s-1,d,n}(X_{s} \multimap X_{s}) \multimap X_s \multimap X_s$}
  \DisplayProof
  \caption{\label{fig_type_2}Type derivation of $2:\forall X_{s}, \oc_{s-1,d,n}(X_{s} \multimap X_{s}) \multimap X_{s} \multimap X_{s}$.}
\end{figure}

\begin{lemma}[subject reduction]\label{lemma_beta_cutrel_snll}
  If there exists a type derivation $D$ whose conclusion is $\Gamma \vdash t:B$ and $t \betared t'$ then there exists a type derivation $D'$ whose conclusion is $\Gamma \vdash t':B$ and $G_D \cutRel^+ G_{D'}$.
\end{lemma}
\begin{proof}
  We prove the lemma by induction on $D$. Because there is a redex in $t$, $t$ cannot be a variable so the last rule is not an $ax$ rule. Let us suppose that the last rule is a unary rule. Then $D$ is of the shape:
  \begin{equation*}
    \AxiomC{$E$}
    \UnaryInfC{$\Delta \vdash u: A$}
    \RightLabel{$R$}
    \UnaryInfC{$\Gamma \vdash t:B$}
    \DisplayProof
  \end{equation*}
  In every case, $u$ is a subterm of $t$ containing the redex. So, by induction hypothesis, $u$ reduces to a $\lambda$-term $u'$. By the induction hypothesis, there exists a derivation $E'$ of conclusion $\Delta \vdash u':B$ and $G_{E} \cutRel^k G_{E'}$ (with $k \geq 1$). We can verify that in every case we can define $D'$ as the following derivation.
  \begin{equation*}\centering
    \AxiomC{$E'$}
    \UnaryInfC{$\Delta \vdash u': A$}
    \RightLabel{$R$}
    \UnaryInfC{$\Gamma \vdash t':B$}
    \DisplayProof
  \end{equation*}
  If the last rule is a $\multimap_e$ or $\Rightarrow_e$ step which does not correspond to the redex, the lemma is proved similarly.

  If the last rule is a $\multimap_e$ rule corresponding to the redex then, by Lemma~\ref{lemma_forall_elim}, $G_D \csRel^* G_E$ with $E$ a derivation of the following shape:
  \begin{equation*}
    \AxiomC{$E_l$}
    \UnaryInfC{$\Gamma, x:A \vdash v : B$}
    \RightLabel{$\multimap_i$}
    \UnaryInfC{$\Gamma \vdash \lambda x.v : A \multimap B$}
    \AxiomC{$E_r$}
    \UnaryInfC{$\Delta \vdash u:A$}
    \BinaryInfC{$\Gamma,\Delta \vdash (\lambda x.v)u:B$}
    \DisplayProof
  \end{equation*}
  By Lemma~\ref{lemma_linear_substitution}, $G_E$ reduces to a derivation of conclusion $\Gamma,\Delta \vdash v[u/x]:B$. If the last rule is a $\Rightarrow_e$ rule corresponding to the redex then, the result follows similarly by Lemmas~\ref{lemma_forall_elim} and~\ref{lemma_exponential_substitution}.
\end{proof}

\begin{figure}\centering
  \AxiomC{}
  \UnaryInfC{$m:N\hspace{-0.15em}\vdash\hspace{-0.15em} m: N$}
  \UnaryInfC{$m:N\hspace{-0.15em}\vdash\hspace{-0.15em} m: \oc_{s,d,n}F \hspace{-0.2em}\multimap\hspace{-0.2em} F $}
  \AxiomC{}
  \UnaryInfC{$g: F \hspace{-0.15em}\vdash\hspace{-0.15em} g: F$}
  \BinaryInfC{$m:N,g:{F}^{s,d,n} \hspace{-0.15em}\vdash\hspace{-0.15em} (m)g: F$}
  \AxiomC{$\vdots$}
  \AxiomC{}
  \UnaryInfC{$x: X_s^{\varnothing} \hspace{-0.15em}\vdash\hspace{-0.15em} x:X_s$}
  \BinaryInfC{$n:N,h:{F}^{s,d,n},x:X_s^{\varnothing} \hspace{-0.15em}\vdash\hspace{-0.15em} ((n)h) x:  X_s$}
  \BinaryInfC{$m: N,n:N,g:{F}^{s,d,n},h:{F}^{s,d,n},x:X_s^{\varnothing} \hspace{-0.15em}\vdash\hspace{-0.15em} ((m)g) ((n)h)x: X_s$}
  \UnaryInfC{$m: N,n:N,f:{F}^{s,d,n},x:X_s^{\varnothing} \hspace{-0.15em}\vdash\hspace{-0.15em} ((m)f) ((n)f)x: X_s$}
  \UnaryInfC{$m: N,n:N,f:{F}^{s,d,n} \hspace{-0.15em}\vdash\hspace{-0.15em} \lambda x. ((m)f) ((n)f)x: X_s \hspace{-0.2em}\multimap\hspace{-0.2em} X_s$}
  \UnaryInfC{$m: N,n:N \hspace{-0.15em}\vdash\hspace{-0.15em} \lambda f.\lambda x. ((m)f) ((n)f)x: \NSDN{s}{d}{n}$}
  \UnaryInfC{$m: N \hspace{-0.15em}\vdash\hspace{-0.15em} \lambda n.\lambda f.\lambda x. ((m)f) ((n)f)x: \NSDN{s}{d}{n} \hspace{-0.2em}\multimap\hspace{-0.2em} \NSDN{s}{d}{n}$}
  \UnaryInfC{$\hspace{-0.15em}\vdash\hspace{-0.15em} \lambda m.\lambda n.\lambda f.\lambda x. ((m)f) ((n)f)x: \NSDN{s}{d}{n} \hspace{-0.2em}\multimap\hspace{-0.2em} \NSDN{s}{d}{n} \hspace{-0.2em}\multimap\hspace{-0.2em} \NSDN{s}{d}{n}$}
  \DisplayProof
  \caption{\label{fig_type_add}Type derivation of $add: \NSDN{s}{d}{n} \hspace{-0.2em}\multimap\hspace{-0.2em} \NSDN{s}{d}{n} \hspace{-0.2em}\multimap\hspace{-0.2em} \NSDN{s}{d}{n}$. To simlify the proof derivation, we write $F$ for $X_s \hspace{-0.2em}\multimap\hspace{-0.2em} X_s$ and $N$ for $\NSDN{s}{d}{n}^{\varnothing}$.}
\end{figure}

\begin{theorem}\label{sdnll_lambda_polytime}
  If there exists a type derivation $E$ whose conclusion is $\Gamma \vdash t:B$, $x$ is the size of $E$, $S-1$, $D-1$ and $N-1$ are the maximum indexes in $E$, and $\partial$ is the depth of $E$ (in terms of $\Rightarrow_e$ rules), then:
  \begin{equation*}
    t \betared^k t' \hspace{1.5em}\Rightarrow \hspace{1.5em}k \leq x^{1+ D^{S} \cdot \partial^{1+N \cdot S}}
  \end{equation*}
\end{theorem}
\begin{proof}
  Immediate from Theorem~\ref{theo_sdnll_pn_bound} and Lemma~\ref{lemma_beta_cutrel_snll}.
\end{proof}

We can notice that, contrary to $DLAL$, $SDNLL_{\lambda}$ does not allow weakening on linear variables. Thus one can never derive $\vdash \lambda x.t : A \multimap B$ when $x$ is not a free variable of $t$. We are confident\footnote{Adding generalized weakening (Conclusion of $\wn C$ can be any formula) would not change a single line in the proof of Theorem~\ref{theo_sdnll_polysound}. However, in order to prove Lemma~\ref{lemma_beta_cutrel_snll}, we would need to add cut-elimination rules: $\wn C$ cut with any node $n$, deletes $n$ and creates $\wn C$ nodes cut with every premise of $n$. 
$W_G$ would still be a valid bound on the number of steps during reduction, so Sections~\ref{chapter_3} and~\ref{section_polytime_simple} would be the same.} that adding the following rule to $SDNLL$ does not break Lemma~\ref{sdnll_lambda_polytime}. 

\begin{equation*}
  \AxiomC{$\Gamma \vdash t:B$}
  \UnaryInfC{$\Gamma,x:A^{\varnothing} \vdash t:B$}
  \DisplayProof
\end{equation*}
However, one cannot extend the encoding of Figure~\ref{fig_pn_lambda_sdnll} to this rule because Linear Logic does not allow weakening on a formula $A$ unless $A$ is of the shape $\wn A'$. Thus we would have to prove the bound directly on $\lambda$-calculus (or a similar language as in~\cite{baillot2004light}). Which makes the proof more difficult, because we cannot use the lemmas we proved on context semantics. If we had defined the context semantics and the criteria on a more general framework (for example interaction nets, for which we define a context semantics in~\cite{perrinel2014interactionnets}) we would not have problems to accomodate such a simple modification.

To give an intuition on the system, let us give some examples of proof derivations. For any $s \geq 1$, and $k \in \mathbb{N}$, $\underline{k}$ can be typed with the following type (see Figure~\ref{fig_type_2} for the type derivation of $\underline{2}$):
\begin{equation*}
  \Nll_{s,d,n}=\forall X_s, \oc_{s-1,d,n}(X_s \multimap X_s) \multimap X_s \multimap X_s
\end{equation*}
Addition can be typed as shown in Figure~\ref{fig_type_add}. Finally, although this type system has no built-in mechanism to type tuples, we can encode them by the usual church encoding (Figure~\ref{fig_snll_pairs}). Let us notice that this encoding does not require any additional constraint on the types, contrary to $mL^4$ (where the terms must have the same level).

\begin{figure}\centering 
.\vspace{2em}

\AxiomC{$~$}
\UnaryInfC{$f:(A \multimap B \multimap X)^{\varnothing} \vdash f:  A \multimap B \multimap X$}
\AxiomC{$\Gamma \vdash t: A $}
\BinaryInfC{$\Gamma,f: (A \multimap B \multimap X)^{\varnothing} \vdash (f)t :  B \multimap X $}
\AxiomC{$\Delta \vdash u: B$}
\BinaryInfC{$\Gamma,\Delta,f: (A \multimap B \multimap X)^{\varnothing} \vdash ((f)t)u:  X$}
\UnaryInfC{$\Gamma,\Delta \vdash \lambda f.((f)t)u: (A \multimap B \multimap X) \multimap X$}
\UnaryInfC{$\Gamma,\Delta \vdash \lambda f.((f)t)u:\forall X. (A \multimap B \multimap X) \multimap X$}
\UnaryInfC{$\Gamma,\Delta \vdash  \la t,u \ra : \la A,B \ra$}
\DisplayProof
\caption{\label{fig_snll_pairs}Simple encoding of pairs.}
\end{figure}

\label{sect_expressivity}
We isolate four constraints that previous logics ($LLL$, $SLL$, and $MS$) has and which $SDNLL$ does not have. We illustrate each constraint with an intuitive description and a $\lambda$-term which can be typed in $SDNLL$ but seemingly not in previous logics because of this constraint. We set $S=\lambda m.\lambda f.\lambda x. ((m)f) (f) x$ implementing succesor, and $+$ as the $\lambda$-term $\lambda m.\lambda n.\lambda f.\lambda x.((m) f) ((n)f) x$ implementing addition on Church integers and $x+y+z$ is a notation for $((+)((+)x)y)z$.
\begin{itemize}
\item In previous logics, in $\la t,u \ra$, $t$ and $u$ must have the same stratum indices (depth in $LLL$ and $MS$, level in $mL^4$). The term $(\underline{k})\lambda \la x,y,z \ra.\la x,((x) S) \underline{0}, x+y+z \ra$ is not typable in previous logics: because the function $\lambda \la x,y,z \ra.\la x,((x) S) \underline{0}, x+y+z \ra$ is iterated, we have $s(y)=s \left (((x)S)\underline{0} \right )$ so $s(y) > s(x)$. But, because they are in the same tuple, it must be $s(y)=s(x)$. The $x+y+z$ term ensures that the stratum indices of $y$ and $x$ cannot be modified by $\S$ modalities.
\item There is no $N$ rule in previous logics and in their encodings in $SDNLL$. This seems to prevent the typing of $\underline{k}(\lambda \la v,w,x,y,z\ra \la w,w,x,w+x+y,((x) (+)v) 0 \ra)$.
\item Contrary to $LLL$ and $mL^4$, one can have several variables in the context during a $\Rightarrow_e$ rule. So, $t=\underline{k}(\lambda \la x,y,z \ra . \la x,x+y,y \ra)$ is typable in $SDNLL$ but not in $LLL$ and $mL^4$. Moreover, the maximum nest of terms is not a priori bounded by the type system, so if we set $u=\lambda\la x,y,z \ra.\la z,z,z \ra$, then $(t)(u)(t)\cdots (u)t$ is typable in $SDNLL$ whatever the length of the chain of applications, whereas in $MS$ the maximum length of such a chain is bounded.
\item Previous logics had no subtyping. For example, in $mL^4$, a $A^i$ formula cannot be considered as a $A^{i-1}$ formula. The example in the first item of this list would be typable in $mL^4$ if it was allowed to decrease the level of a formula by mean of a subtyping relation.
\end{itemize}

\section{Conclusion and further work}
In order to address the potential applications given in the introduction (real-time systems, complexity debugging, mathematical proofs) we aim to create a type system for a programming language such that:
\begin{enumerate}
\item Programming in the language is practical. The language offers usual features such as built-in types (integers, boolean,...), control flow operations, recursive definitions, side effects,...\label{goal1}
\item Type inference is decidable in reasonable time.\label{goal2}
\item For most polynomial time program users will write, the type infered entails a polynomial bound.\label{goal3}
\item The bounds infered are often tight (very important for real-time systems, rather important for complexity debugging, unimportant for mathematical proofs).\label{goal4}
\end{enumerate}
We consider that goals~\ref{goal2} and~\ref{goal4} highly depend on the system. We could try to design a faster type inferrence algorithm for $LLL$, or infer tighter bounds for a $LLL$ program. However, because of its lack of expressivity, there is little chance that $LLL$ will be used in practice for the goals we have in mind. A new system must be created, and type inferrence and tight bound inferrence may be totally different in this new system.

We view goals 1 and 3 as mostly orthogonal. It is possible to define an expressive functional core (a linear logic subsystem or $\lambda$-calculus type-system) and add features to it. For example, previous works have added pattern-matching, recursive definitions~\cite{baillot2010polytime}, side-effects and concurrent features~\cite{madet2012polynomial} to $LLL$. However, previous works only extended a specific system ($LLL$ in those cases). We are not aware of any work proving that ``for every subsystem $S$ of linear logic sound for $Ptime$ and verifying some condition $C$, the system obtained by adding feature $F$ to $S$ is sound for $Ptime$''. So, if we added other features to $LLL$ without breaking the polynomial bound, it is unclear whether we would have been able to add those features to other subsystems of linear logic characterizing polynomial time.  In those conditions, it made more sense to first work on the functional core (goal~\ref{goal3}), and in a second step, add features to it (goal~\ref{goal1}).

Because $SDNLL$ is more expressive than previous subsystems of linear logic characterizing polynomial time, this work fits into goal 3. However, in the same way Baillot and Mazza considered that the ``fundamental contribution of'' \cite{baillot2010linear} is not the definition of the systems $mL^3$ and $mL^4$ themselves, but the demonstration that ``in linear-logical characterizations of complexity classes, exponential boxes and stratification levels are two different things'', we consider that the main contribution of our work is the idea to define semantic criteria based on the acyclicity of relation on boxes:
\begin{itemize}
\item Using those criteria, we separated three principles underlying $LLL$ and the works based on it. It sheds a new light on previous works: $mL^4$ relaxes the ``stratification'' criterion of $LLL$, while $MS$ relaxes its ``dependence control'' criterion (we are not aware of previous works relaxing the ``nesting'' condition). Realizing that those principles are mainly orthogonal can help further works on the expressivity of linear logic subsystems characterizing polynomial time: independent improvements on different principles can be combined. For instance, one can easily verify, that one can combine $mL^4$ and a maximal $Ptime$ system of $MS$\footnote{$mL^4$ with its ``at most one auxiliary door by box'' replaced by the indices criteria of the $MS$ system to control dependence}. In fact, $SDNLL$ can be seen as an extension of such a system.
\item Because the lemmas and results of sections~\ref{chapter_2},~\ref{chapter_3} and~\ref{section_polytime_simple} are valid for any untyped proof-net, they can be reused to prove polynomial bounds for other subsystems of linear logic in which $\stratSNLL$, $\dcSim$ or $\nestSim$ are acyclic (such as  $mL^4$, $MS$ and the multiplicative fragment of $LLL$) or to define new criteria: in~\cite{perrinelMegathese}, Perrinel builds upon these technical lemmas to define more expressive criteria entailing a polynomial bound and a criterion entailing a primitive recursive bound.
\item We separated the task of creating an expressive subsystem of linear logic characterizing polynomial time into subtasks: finding loose criteria on semantic entailing polynomial time, and finding syntactic criteria entailing those semantic criteria. One can closely examine the proofs leading to Corollary~\ref{coro_bound_poly_nest}, to find any unnecessary assumption on the proof-net behaviour. While this may be subjective, we found it much easier to reason about complex semantic criteria without having to consider the exact way in which they will be enforced.
\item While the syntax of $SDNLL$ (or any other syntactic subsystem of linear logic) may be difficult to adapt to richer languages where the notion of reduction differs from cut-elimination, those relations on boxes have a meaning going beyond linear logic itself: $B \stratSNLL C$ means that $B$ interacts with an element created by an interaction of $C$ (with nodes created when $C$ is opened/interacts), $B \dcSim C$ and $B \nestSim C$ represent two ways of having several duplicates of $B$ inside $C$. Thus it would be interesting to investigate the application of those principles to other models of computation based on reduction/rewriting.
\end{itemize}

The applications considered in the introduction are used to motivate the direction of our research, to explain why the intensional expressivity of our characterization is an important problem. We are still far from having a system expressive enough to handle them. We explained why we first focused on the expressivity of the functional core (subsystems of linear logic and type systems on plain $\lambda$-calculus), but we consider that the main challenge in the future of Implicit Computational Complexity will be to add features to this functional core (built-in types, pattern-matching, recursive definitions, side-effects,...) to get closer to the programming languages used in practice. Thus, the fact that those criteria might be easier to adapt to a more practical framework is especially important.

In a previous work~\cite{perrinel2014interactionnets}, we defined a context semantics for interaction nets: a well-behaved class of graph rewriting systems~\cite{lafont1989interaction} based on proof-nets. Interaction net is not a singe system, but a set of such systems. Thus, this framework seems particularly adapted to the progressive addition of features. An interesting problem for future work would be to use the context semantics of~\cite{perrinel2014interactionnets} to define relations on interaction nets corresponding to $\stratSNLL$, $\dcSim$ and $\nestSim$.

\section{Bibliography}
\bibliographystyle{plain}
\bibliography{megabib} 
\end{document}